\documentclass[manuscript,screen,acmtog]{acmart}

\AtBeginDocument{%
  }

\setcopyright{acmlicensed}
\copyrightyear{2024}
\acmYear{2024}
\acmDOI{XXXXXXX.XXXXXXX}

\acmISBN{978-1-4503-XXXX-X/18/06}

\usepackage{algorithm}
\usepackage{algpseudocode}
\usepackage{pifont}
\usepackage{upgreek}
\usepackage{subfigure} 
\usepackage{svg}
\usepackage{multirow}
\usepackage{wrapfig}
\usepackage{placeins}

\def \bezier {B{\'e}zier~}

\acmSubmissionID{1257}

\begin{document}

\title{Fast and Robust Point Containment Queries on Trimmed Surface}


\author{Anchang Bao}
\affiliation{%
  \institution{Tsinghua University}
  \city{Beijing}
  \country{China}
}
\email{baoanchang02@gmail.com}

\author{Enya Shen}
\affiliation{%
  \institution{Tsinghua University}
  \city{Beijing}
  \country{China}
  }
\affiliation{%
  \institution{Haihe Lab of ITAI}
  \city{Tianjin}
  \country{China}
  }

\email{shenenya@tsinghua.edu.cn}

\author{Jianmin Wang}
\affiliation{%
  \institution{Tsinghua University}
  \city{Beijing}
  \country{China}
  }
\email{jimwang@tsinghua.edu.cn}

\begin{abstract}
  Point containment queries on trimmed surfaces are fundamental to CAD modeling, solid geometry processing, and surface tessellation. Existing approaches—such as ray casting and generalized winding numbers—often face limitations in robustness and computational efficiency.

  We propose a fast and numerically stable method for performing containment queries on trimmed surfaces, including those with periodic parameterizations. Our approach introduces a recursive winding number computation scheme that replaces costly curve subdivision with an ellipse-based bound for Bézier segments, enabling linear-time evaluation. For periodic surfaces, we lift trimming curves to the universal covering space, allowing accurate and consistent winding number computation even for non-contractible or discontinuous loops in parameter domain.

  Experiments show that our method achieves substantial speedups over existing winding-number algorithms while maintaining high robustness in the presence of geometric noise, open boundaries, and periodic topologies. We further demonstrate its effectiveness in processing real B-Rep models and in robust tessellation of trimmed surfaces.
\end{abstract}



\received{20 February 2007}
\received[revised]{12 March 2009}
\received[accepted]{5 June 2009}

\newcommand{\R}{\mathbb{R}}

\begin{teaserfigure}
  \centering
  \includegraphics[width=\textwidth]{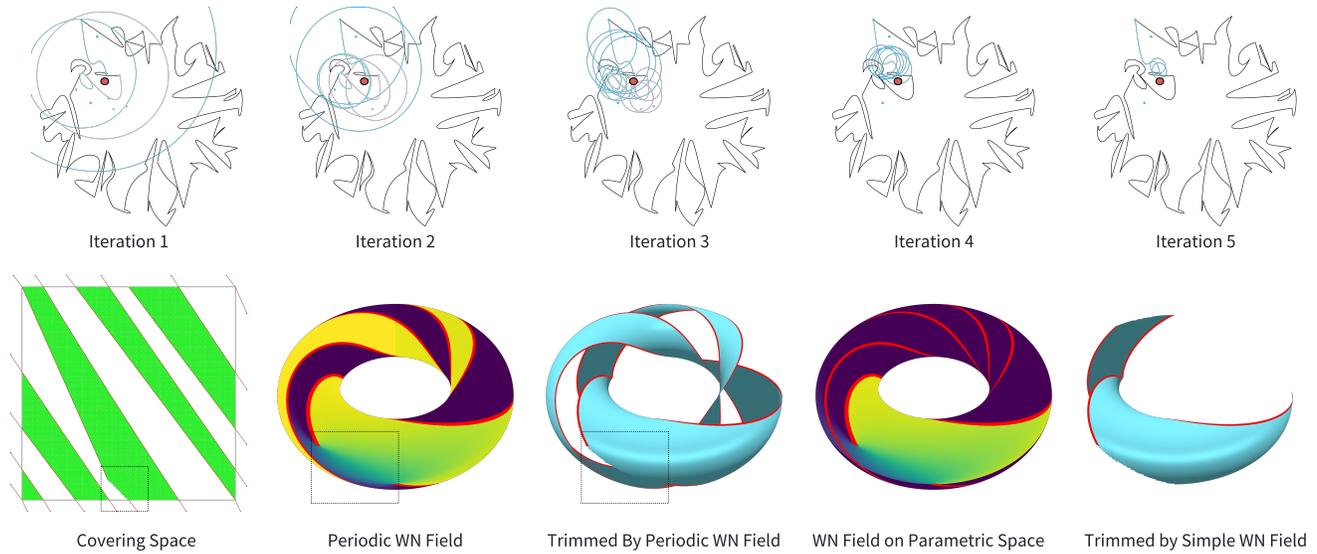}  
  \caption{We proposed an efficient winding number computation method for parametric curves and generalize to the periodic domain by lifting curves to the universal covering space. With this proposed method, we developed fast and robust point containment queries on trimmed surfaces. \textbf{First row}: Illustration of the winding number computation method described in section \ref{sec:method}. The red point is the testing point and the colored curves are ellipses used in our algorithms. Different from the \citet{spainhour_robust_2024}, we utilize a more flexible ellipse bound, enabling linear time \bezier curves evaluation and $O(1)$ inclusion test. \textbf{Second row}: By lifting curves to the universal covering space, we correctly compute the winding number on the periodic domain. With the proposed method, we can handle containment queries on the periodic trimmed surface robustly. Given open boundary curves, we present the winding number field calculated in covering space and parametric space respectively.}
  \label{fig:graphical_abstract}
\end{teaserfigure}

\maketitle

\section{Introduction}

Trimmed surfaces play a central role in free-form modeling and are commonly represented using Boundary Representation (B-Rep), the industry standard for free-form solids~\cite{pratt2001introduction}.
Unlike ordinary spline patches, a trimmed surface is defined by an underlying geometric surface and a set of boundary curves that restrict its visible region.
Processing boundary-represented solids lies at the core of free-form modeling, yet one key limitation of B-Rep is the absence of explicit containment relations between solids.
As a result, containment tests become essential whenever inclusion relations are required.
A fundamental instance of this task is determining whether a point on the underlying geometric surface lies within its trimmed region—known as a containment query on a trimmed surface.
Figure~\ref{fig:torus_point_in_face} illustrates containment queries on a trimmed torus, a problem with broad applications in computer graphics, including CAD model processing, B-Rep ray tracing, and tessellation of trimmed surfaces.

\begin{figure}[t]
\centering
\includegraphics[width=0.6\linewidth]{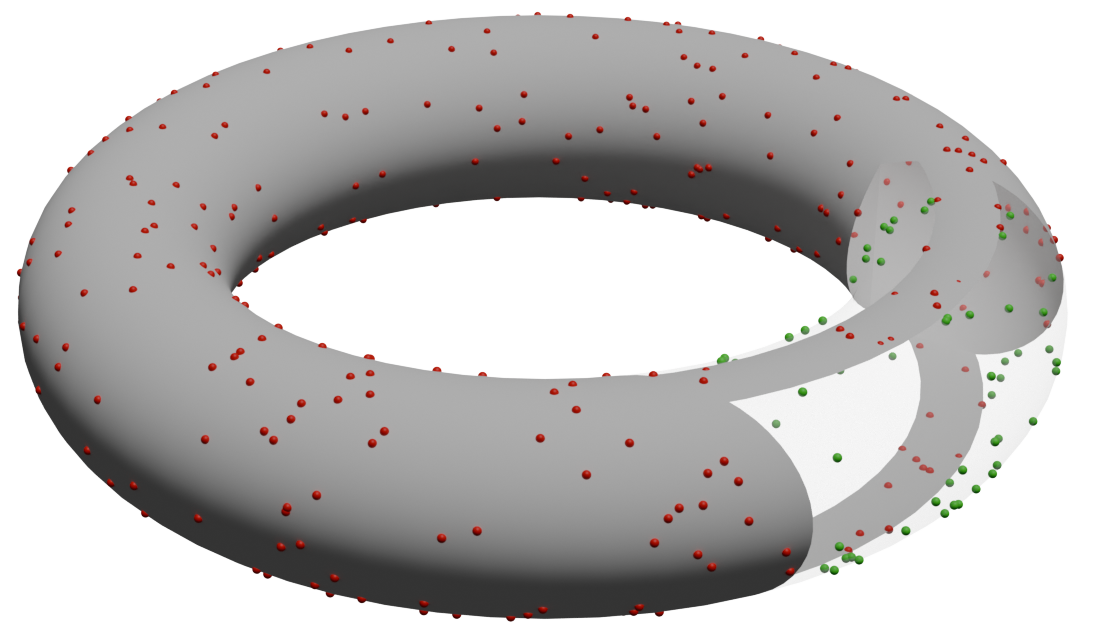}
\caption{Containment queries on a trimmed torus. The opaque region represents the retained surface, while the translucent part is trimmed away. Points on the torus are randomly sampled; those inside the trimmed region are marked in red, and those outside are shown in green.}
\label{fig:torus_point_in_face}
\end{figure}

Containment queries arise in a variety of CAD and modeling contexts.
During Boolean operations between parametrically represented solids, numerous intersection edges are generated (Figure~\ref{fig:topological_intersection}).
Curve segments that lie on the underlying geometric surface but outside the trimmed region must be removed, which requires point–surface containment tests.
Similarly, in ray-traced rendering of B-Rep models, rays may intersect the underlying surface at positions that must be verified as belonging to the trimmed region.
Containment queries are also integral to trimmed-surface tessellation, where sample points in the parameter domain are classified as inside or outside the trimmed region before rendering.

In B-Rep, curves on surfaces are represented by a pair of mappings: a curve in physical space and its corresponding projection in the parameter domain, known as a \textbf{p-curve}.
With p-curves, containment testing on non-periodic surfaces reduces to a 2D inside–outside test in the parameter domain.
Ray-casting methods are commonly used for this purpose, but they are prone to corner cases and numerical instability.
Winding-number–based methods, in contrast, are more robust but computationally expensive.

Containment queries on periodic domains are particularly challenging because the underlying topology alters how boundaries are represented.
For example, when the geometric surface is cylindrical (Figure~\ref{fig:periodic_failed}), the trimmed region is defined by separation loops.
In such cases, the p-curves may not form closed loops in parameter domain, making standard trimming techniques such as ray casting or direct winding-number evaluation inadequate.

In this paper, we present a fast and robust framework for containment queries on trimmed surfaces.
To ensure robustness against geometric and numerical errors, we adopt the winding number as the criterion for inside–outside classification.
We introduce an efficient acceleration structure for winding-number evaluation based on ellipse bounds of Bézier curves, and extend the formulation to periodic domains by lifting the computation to the universal covering space.

Our main contributions are as follows:

\begin{itemize}
\item We introduce a novel recursive scheme that accelerates winding-number evaluation by bounding Bézier segments with ellipse regions, avoiding $O(n^2)$ curve subdivision and achieving linear-time complexity.
\item We formulate containment queries on periodic surfaces in the universal covering space, enabling robust computation for loops that are non-contractible or discontinuous in parameter domain.
\item We demonstrate the use of our algorithm for robust tessellation of trimmed surfaces and interactive design of complex periodic geometries.
\end{itemize}

The remainder of this paper is organized as follows.
Section~\ref{sec:related_work} reviews previous work on containment queries and curve evaluation.
Section~\ref{sec:background} presents the preliminaries.
Section~\ref{sec:method} details the proposed winding-number computation and its extension to periodic domains.
Section~\ref{sec:experimentation} reports experimental comparisons and applications.
Section~\ref{sec:future} discusses limitations and directions for future work, and Section~\ref{sec:conclusion} concludes the paper.

\begin{figure}[t]
\centering
\includegraphics[width=0.8\linewidth]{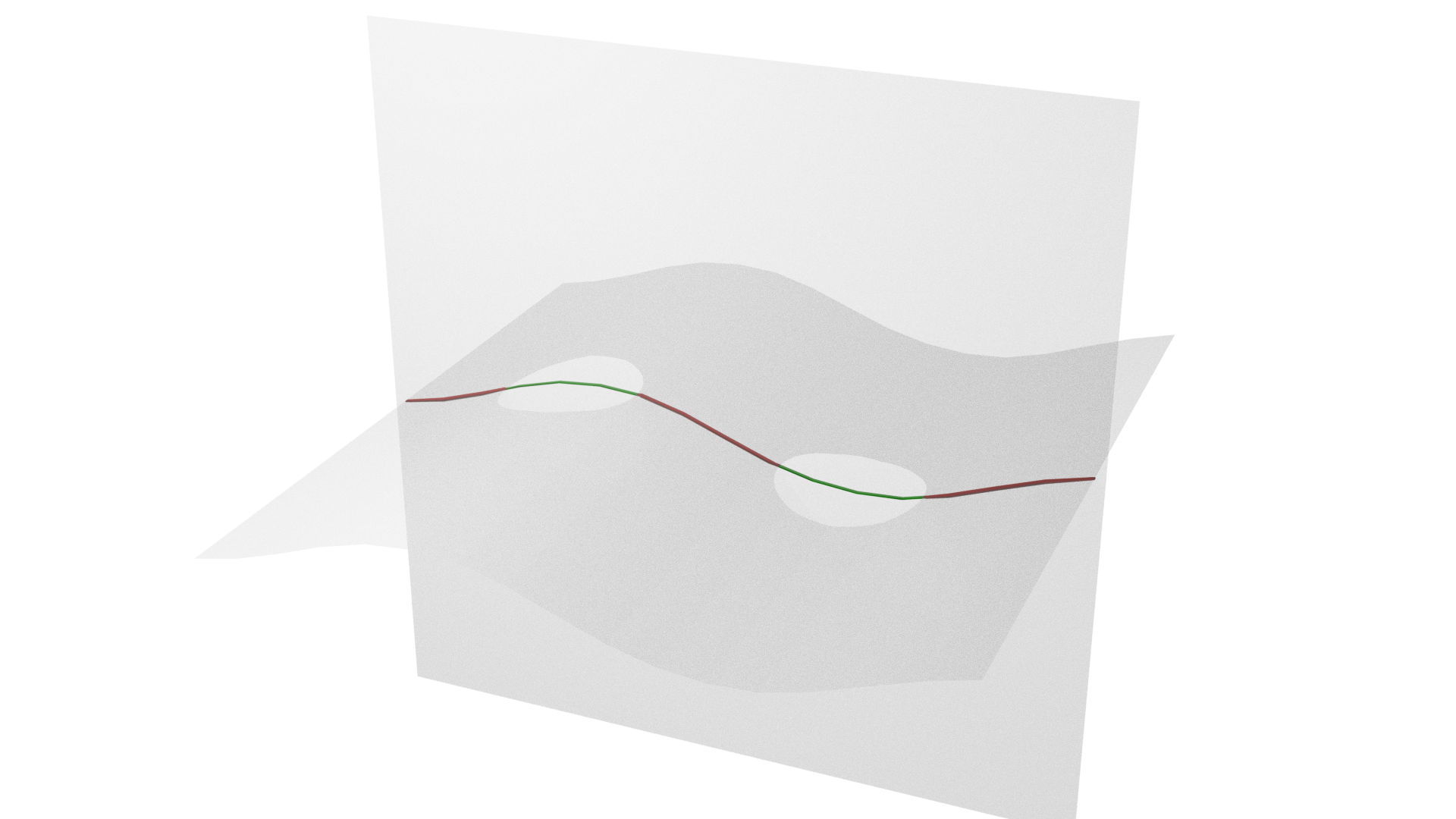}
\caption{Topological intersection between two faces. Extraneous portions of the intersection curve must be removed using point–surface containment queries.}
\label{fig:topological_intersection}
\end{figure}
\section{Related Work}
\label{sec:related_work}

\subsection{Containment Queries for Non-Parametric Boundaries}

We first review containment query methods developed for \emph{non-parametric} geometries such as line segments, polygons, and triangular meshes. These problems have been extensively studied in the computer graphics community, primarily through two families of techniques: ray casting and winding number computation.

\paragraph{Ray casting.}
Ray casting is the most widely used approach for point-in-solid classification, both in two-dimensional polygons and three-dimensional meshes. The method shoots a ray from the test point in a fixed direction and counts intersections with the object boundary. By the Jordan curve theorem, an odd number of crossings indicates an interior point, while an even number indicates an exterior one. 

Ray casting can be accelerated through spatial data structures such as KD-trees and bounding volume hierarchies (BVH)~\cite{rubin19803,goldsmith1987automatic}. For rendering purposes, modern frameworks also support GPU acceleration and parallel ray traversal~\cite{parker2010optix}.

Despite its success in photo-realistic rendering, ray casting suffers from robustness issues in geometry processing. It fails for open or non-manifold boundaries, since the Jordan theorem applies only to closed curves. In practice, ray-casting failures often occur at degenerate intersections, coincident edges, or geometric inconsistencies (Figure~\ref{fig:ray_casting}). A common workaround is to shoot multiple rays in random directions, which increases computational cost.

\paragraph{Winding number methods.}
Winding numbers provide a robust alternative to ray casting by offering a continuous measure of containment. 
\citet{jacobson_robust_2013} introduced the \emph{generalized winding number} to handle open boundaries robustly. 
\citet{barill2018fast} extended this idea to point clouds, and \citet{sun2024gauwn} defined a differentiable version based on Gaussian smoothing. 
Most relevant to our work, \citet{feng2023winding} formulated the \emph{surface winding number} (SWN) as a solution to a boundary value problem with a jump condition, enabling manifold-aware containment queries. 
However, this approach requires solving a global linear system, which is inefficient for local or single-point queries.

Compared to ray casting, winding-number-based methods are more robust in the presence of geometric errors, open boundaries, and non-manifold configurations. Spatial acceleration structures for non-parametric geometries have also been explored~\cite{jacobson_robust_2013,barill2018fast}, but analogous data structures for parametric geometries remain underdeveloped.

\begin{figure}[t]
\centering
\includegraphics[width=0.4\textwidth]{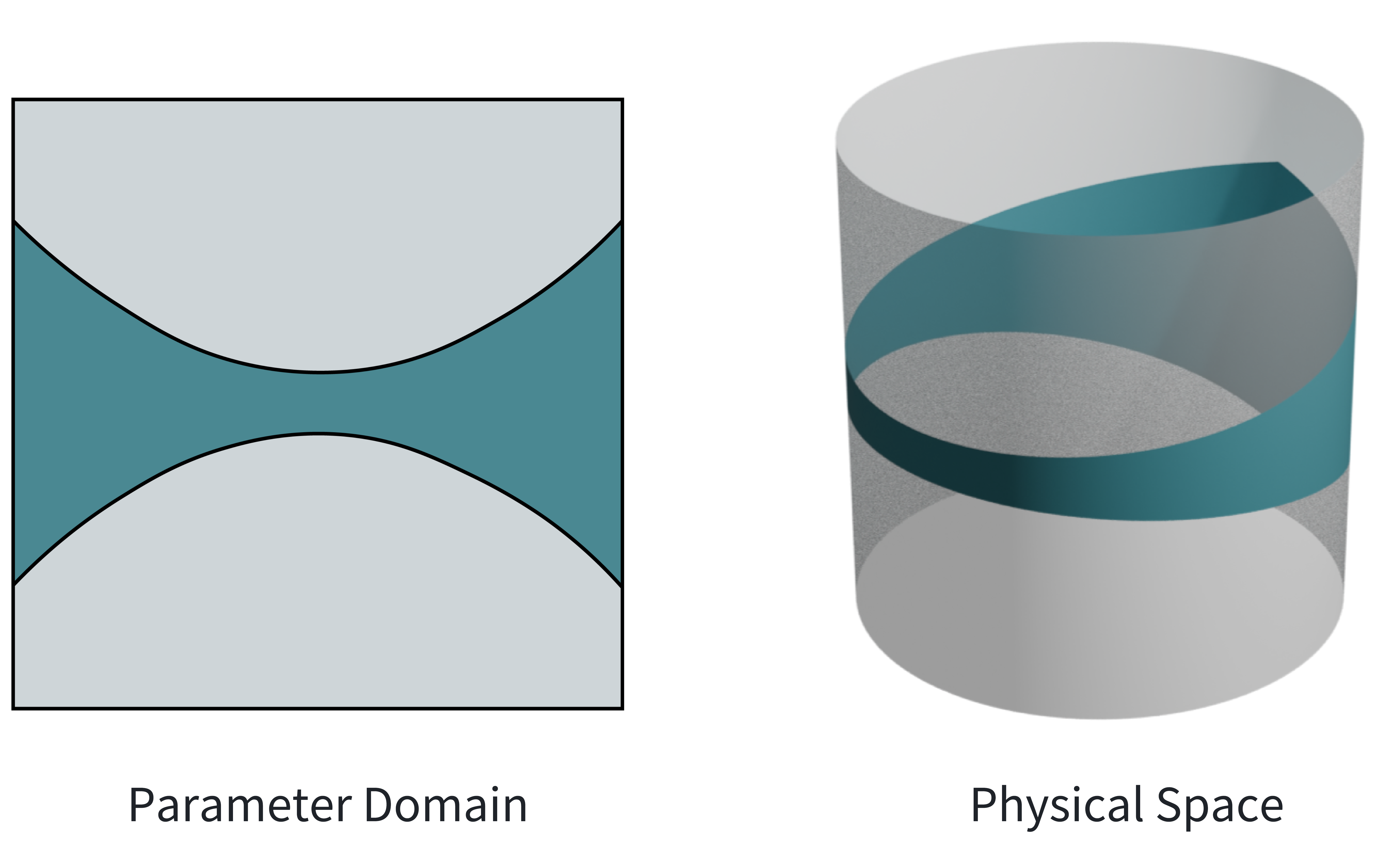}
\caption{A periodic trimmed surface defined on a cylindrical geometry. The boundary consists of two separation loops, which do not form closed loops in the parameter domain.}
\label{fig:periodic_failed}
\end{figure}

\subsection{Containment Queries for Parametric Boundaries}

We now discuss methods for containment queries involving \emph{curved parametric boundaries}, such as Bézier or NURBS curves.

\paragraph{Ray casting.}
Containment queries on trimmed surfaces are typically solved by ray-casting-based methods. 
The main computational challenge lies in computing the intersection between a ray and a curve, which reduces to finding the roots of a polynomial expressed in the Bernstein basis. 
\citet{nishita1990ray} introduced the \emph{Bézier clipping} algorithm, which exploits the control polygon to isolate root intervals. 
Subsequent works improved its convergence rate, such as quadratic clipping~\cite{bartovn2007computing} and cubic clipping~\cite{liu2009fast}. 
The current state-of-the-art is the \emph{Bézier zero factoring} method~\cite{machchhar2016revisiting}, which reduces polynomial degree by exploiting factorization properties of Bernstein polynomials, achieving superior efficiency.

\begin{figure}[t]
    \centering
    \includegraphics[width=0.45\textwidth]{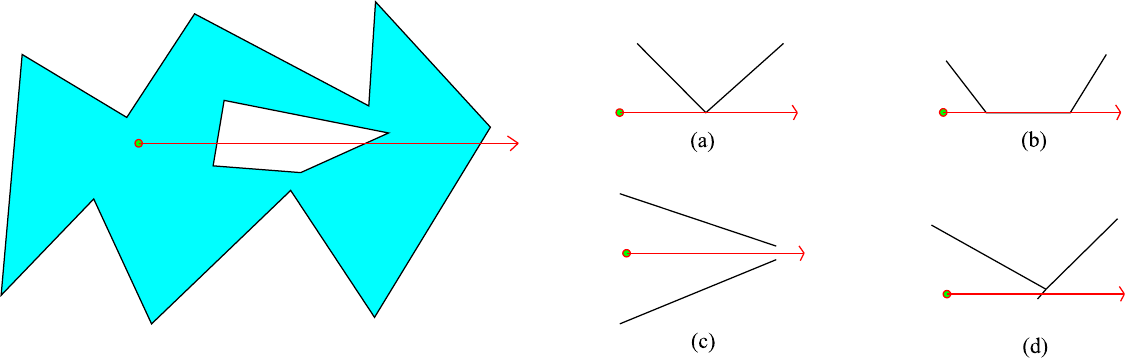}
    \caption{
    Ray casting and its failure cases. 
    Left: The classic even–odd rule determines inside/outside classification by counting ray–boundary intersections. 
    Right: Common failure cases include (a) invalid intersection points, (b) rays coincident with boundary segments, (c) open boundaries due to geometric errors, and (d) non-manifold configurations.
    }
    \label{fig:ray_casting}
\end{figure}

\paragraph{Curve hierarchical schemes.}
To accelerate repeated ray–curve queries, hierarchical structures have been developed for Bézier curves. 
Most are based on identifying monotone segments or local extrema to form spatial subdivisions. 
 
Earlier data structures include the \emph{arc tree}~\cite{gunther1990arc}, which uses arc-length parameterization, and the \emph{strip tree}~\cite{ballard1981strip}, which bounds curve segments with rectangles. 
\citet{schollmeyer_direct_2009} proposed splitting trimming curves into monotone pieces to avoid root finding. 
\citet{schollmeyer_efficient_2019} introduced KD-tree-based spatial subdivision, while \citet{sloup_optimizing_2021} refined KD-tree construction for faster point location.
For Bézier curves, recursive control-point hierarchies remain the most widely used strategy.

\paragraph{Winding number methods.}
All methods relying on crossing counts are inherently sensitive to geometric errors such as open or non-manifold boundaries. 
While repair techniques exist, such issues are still pervasive in CAD and B-Rep models. 
\citet{spainhour_robust_2024} proposed a robust containment method for parametric curves using the generalized winding number. 
Our method builds upon this line of work but introduces an \emph{ellipse bound} for faster Bézier evaluation and bound inclusion tests, as well as a complete treatment of periodic domains. 
\citet{liu2025closed} formulated winding-number computation through complex polynomial roots, and \citet{martens2025one} developed a hybrid approach combining ray casting with winding evaluation. 
\citet{spainhour2025robust} further extended winding numbers to trimmed surfaces; their focus is on determining whether a point lies inside a solid defined by trimmed surfaces, while our work focuses on classifying points \emph{on} the trimmed surface itself.

\paragraph{Periodic surfaces.}
Periodic surfaces are pervasive in CAD, including cylinders, tori, and spheres, all of which have periodic parameter domains (see Figure~\ref{fig:torus_point_in_face}). 
Such periodicity causes difficulties for both ray casting and winding-number-based methods, as trimming curves may cross domain boundaries or appear artificially open. 
To our knowledge, no prior work provides a systematic solution for containment queries on trimmed surfaces with periodic parameterizations. 
\citet{feng2023winding} partially addresses periodicity, but only for non-parametric geometries and via global linear solves.

\section{Preliminaries}
\label{sec:background}

\subsection{Boundary Representation and Trimmed Surfaces}

Boundary representation (B-Rep) is the de facto international standard~\cite{pratt2001introduction} for representing parametric CAD models.  
In B-Rep, a solid entity is defined by its boundaries.  
To represent the boundary of complex solids, flexible free-form surfaces are required, and the most common choice is trimmed non-uniform rational B-spline (NURBS) surfaces.  

A NURBS surface is described by a control net and defines a smooth mapping from a (possibly periodic) parameter domain, usually $[0,1]\times[0,1]$, to physical space $\mathbb{R}^3$.  
To represent surfaces with non-zero genus, a set of curves $\{\gamma_i\}$ is defined on the surface to trim it, yielding the concept of a \textbf{trimmed surface}.  
A trimmed surface consists of a parametric surface patch (typically a NURBS surface) and a set of boundary curves $\gamma_i$.  
The projection of these boundary curves onto the parameter domain is referred to as the \textbf{p-curves}.  

Mathematically, a trimmed surface can be completely specified by its geometric and topological data even without the p-curves.  
However, for computational efficiency in tasks such as tessellation or containment queries, p-curves are an essential component of B-Rep models.  
They can be obtained by sampling, inverting, and subsequently fitting the results with Bézier or NURBS curves.  
For further discussion on p-curve computation, we refer readers to~\citet{renner2004exact, flory2008constrained, hofer2004energy}.  
P-curves also naturally arise in surface–surface intersections computed via marching methods~\cite{barnhill1990marching}; see~\citet{marussig2018review} for a comprehensive review of intersection and trimming techniques.  

In this paper, the p-curves are represented by (rational) Bézier or NURBS curves.  
For simplicity, we focus on the Bézier case; the rational case is discussed in Appendix~\ref{app:ellipse}.  
B-spline or NURBS curves can always be subdivided into (rational) Bézier segments~\cite{farin2001curves}.

Curves on general surfaces can exhibit significant geometric complexity, but in computer graphics and CAD, the discussion is simplified by the topological constraints of spline parameterizations.  
Modern CAD systems parameterize surfaces via mappings from rectangular domains $[a,b]\times[c,d]$ to $\mathbb{R}^3$.  
Cylinder-like or torus-like surfaces are represented using periodic parameter domains.  
\citet{herron1985smooth} proved that smooth spline surfaces cannot represent shapes of genus greater than one.  
Consequently, a single spline patch is topologically equivalent to a disk, sphere, cylinder, or torus—all of which can be covered by periodic domains.  
To represent surfaces of higher genus, multiple patches must be employed.

\subsection{Generalized Winding Number}
\label{sec:wn}

The winding number has long been used in computer graphics.  
\citet{jacobson_robust_2013} introduced the \textit{generalized winding number} for robust inside/outside classification.  
For a closed curve $\gamma$ parameterized by $t\in[0,1]$ and a fixed point $p$, the winding number $\omega(p,\gamma)$ is defined as
\begin{equation*}
\omega(p,\gamma) = \frac{1}{2\pi} \int_{\gamma - p} \mathrm{d}\theta,
\end{equation*}
where $\theta$ is the polar angle with respect to the positive $x$-axis.

By definition, the winding number has an important property: it is \textit{additive}.  
If $\gamma = \gamma_1 + \gamma_2$ (defined by concatenation), then
\begin{equation}
\omega(p,\gamma) = \omega(p,\gamma_1) + \omega(p,\gamma_2).
\end{equation}

From the viewpoint of algebraic topology, the winding number is also \textit{homotopy invariant}.  
If $\gamma$ and $\bar{\gamma}$ are homotopy equivalent in $\mathbb{R}^2\setminus\{p\}$ (with the test point removed), then
\begin{equation}
\omega(p,\gamma) = \omega(p,\bar{\gamma}).
\end{equation}
Intuitively, the winding number remains unchanged under any continuous deformation of $\gamma$ that does not cross the test point $p$.  

Using this property, if $\gamma$ lies within a simply connected region $\Omega$ and $p\notin\Omega$, then
\begin{equation}
\omega(p,\gamma) = \omega(p,\gamma(0),\gamma(1)),
\end{equation}
where $\omega(p,a,b)$ denotes the winding number of the oriented segment $\overrightarrow{ab}$ with respect to $p$, computed from its angular difference.  
These two properties form the foundation of our method as well as of previous winding-number–based approaches~\cite{spainhour_robust_2024}.

\subsection{Homology of Loops on a Torus}

To precisely describe our winding-number algorithm on periodic domains and explain its correctness, we employ the language of homology.  
For readers unfamiliar with homology theory, a loop on a torus with homology coefficients $(p,q)$ can be understood geometrically as one that wraps $p$ times around one periodic direction of the surface (for example, the longitudinal direction) and $q$ times around the other (the meridional direction).  
In the corresponding \textit{universal covering space}—which unfolds the periodic surface into an infinite planar tiling—this loop no longer closes onto itself; instead, its start and end points differ by a translation of $(p,q)$ units.  
Loops with coefficients $(0,0)$ are contractible on the original surface, whereas nonzero coefficients indicate loops that cross periodic boundaries.  
For a formal and rigorous treatment of homology from the perspective of algebraic topology, we refer readers to~\citet{munkres2018elements}.

\section{Method}
\label{sec:method}

In this section, we first outline our containment query algorithm and introduce the \textit{ellipse bound}, which avoids expensive curve subdivision (Section~\ref{sec:ellipse}). 
We then discuss periodic parameterizations in CAD systems and interpret them in terms of universal covering spaces (Section~\ref{sec:covering}). 
Building on this, Section~\ref{sec:periodic_curves} establishes winding number primitives for periodically extended curves, and Section~\ref{sec:periodicity} describes the complete winding number algorithm for uni- and bi-periodic domains.

For convenience, we assume the parameter domain of the surface is $D = [0,1] \times [0,1]$, and all curve parameters vary within $[0,1]$.

\subsection{Algorithm Overview}

Given a test point $p$ and a set of trimming curves $\{\gamma_i\}$, our goal is to determine whether $p$ lies within the trimmed region. 
We do this by computing the winding number of the trimming curves with respect to $p$, and then classifying $p$ as inside or outside based on this value.  

As discussed in Section~\ref{sec:wn}, the winding number is both \textit{additive} and \textit{homotopy invariant}. 
Thus, for each trimming curve, if we can enclose it within a convex region in the parametric domain, then for any test point outside this region, the curve can be replaced by a straight line segment without affecting the winding number. 
For test points inside the region, we subdivide the curve at its midpoint, compute new bounds for the two subsegments, and apply the same process recursively.
In the recursion, if we detected that the distance between $p$ and current endpoints is smaller than a predefined tolerance $\epsilon$, the point is reported as lying on the boundary. 
Note that not all points within the $\epsilon$-shell of boundary curves are labeled \texttt{on-boundary}, since points that can be confidently classified as inside or outside are excluded. 

The outline of the recursive winding number computation is shown in Algorithm~\ref{alg:winding}. 
The two main primitives are (1) how to bound a Bézier segment, and (2) how to classify whether a point lies inside the bound.

\begin{algorithm}
\caption{ComputeWindingNumber}
\label{alg:winding}
\begin{algorithmic}[l]
\Require Test point $p$, 2D parametric Bézier curve $\gamma$, interval $[a,b]$, tolerance $\epsilon$
\Ensure Winding number $\omega(p, \gamma_{[a,b]})$
\If{$d(p, \gamma(a)) < \epsilon$ or $d(p, \gamma(b)) < \epsilon$}
    \State Return \texttt{on-boundary}
\ElsIf{$p$ not in bound of $\gamma_{[a,b]}$}
    \State $\omega \gets WindingNumberLineSegment(p, \gamma(a), \gamma(b))$
\Else
    \State $m \gets (a+b)/2$
    \State $\omega_1 \gets ComputeWindingNumber(p, \gamma, a, m, \epsilon)$ 
    \State $\omega_2 \gets ComputeWindingNumber(p, \gamma, m, b, \epsilon)$
    \State $\omega \gets \omega_1 + \omega_2$
\EndIf
\end{algorithmic}
\end{algorithm}

\subsection{Ellipse Bound}
\label{sec:ellipse}

Let $\gamma$ be a Bézier curve defined by control points $\{P_i\}_{i=0}^n$. 
It is well known that the convex hull of the control points bounds the curve. 
However, using this bound in each recursion requires invoking De Casteljau’s algorithm to subdivide the curve and update the control points of each split. 
For containment tests, a point-in-polygon query is also needed. 
This results in $O(n^2)$ cost for curve subdivision and $O(n)$ for point-in-bound tests.

To reduce this cost and enable efficient geometric primitives such as~\citet{wozny2020linear}, we propose an \textit{ellipse bound} that replaces the convex-hull test with a simple analytic condition.

\begin{lemma}[Ellipse Bound]
\label{thm:ellipse}
Every (rational) Bézier curve $\gamma$ lies within the ellipse region $\Omega_\gamma$ defined by
\begin{equation}
\label{equ:ellipse}
\Omega_\gamma = \{\, x \in \mathbb{R}^2 \mid d(x, P_0) + d(x, P_n) \le B \,\},
\end{equation}
where $B$ is an upper bound on the curve’s derivative magnitude. 
For a non-rational Bézier curve, $B = \max_i \| n (P_i - P_{i-1}) \|$. 
The rational case is given in Appendix~\ref{app:ellipse}.
\end{lemma}

The proof follows directly from the Lipschitz bound of Bézier derivatives and is provided in Appendix~\ref{app:ellipse}. 
Conceptually, this bound defines an ellipse whose foci are the curve’s endpoints, and whose major axis is determined by the curve’s maximum derivative. 
This ellipse fully encloses the curve.

For a subsegment $\gamma' : [a,b] \to \mathbb{R}^2$, the corresponding bound becomes
\begin{equation}
\label{equ:sub_ellipse}
\Omega_{\gamma'} = \{\, x \in \mathbb{R}^2 \mid d(x, \gamma(a)) + d(x, \gamma(b)) \le B(b-a) \,\}.
\end{equation}
This formulation enables the recursive winding number algorithm to test containment using only the endpoints of the curve—\textit{avoiding both De Casteljau subdivision and point-in-polygon queries}—thereby improving performance while keeping numerical robustness.

With the ellipse bound, the termination of Algorithm~\ref{alg:winding} can be guaranteed which is formally described by following lemma and proved in Appendix~\ref{proof:termination}.

\begin{lemma}[Termination]
\label{thm:termination}
The maximum recursion depth in Algorithm~\ref{alg:winding} is 
\[
O\!\left(\log \frac{B}{\max(\epsilon, d)}\right),
\]
where $\epsilon$ is the geometric tolerance, $d$ is the distance between the test point $p$ and the curve $\gamma$, and $B$ is the bound constant same with the Equation~\ref{equ:ellipse}.
\end{lemma}

In practice, far fewer subdivisions are needed, since most subsegments are pruned during recursion. 
Empirically, the average runtime grows logarithmically with respect to $1/\epsilon$ for boundary-adjacent test points.

\paragraph{Path hierarchy.}
Beyond single curves, we define a \emph{path} as a sequence of end-to-end continuous Bézier segments $\{\gamma_i\}_{i=1}^n$. 
An ellipse bound for any subpath $\{\gamma_i\}_{i=j}^k$ can be written as
\begin{equation*}
\Omega_{j,k} = \bigl\{\, x \in \mathbb{R}^2 \ \big|\  d(x, \gamma_j(0)) + d(x, \gamma_k(1)) 
    \le \max_{j \le m \le k} B_m (k-j+1) \,\bigr\}.
\end{equation*}
This hierarchical structure enables efficient pruning in both CAD and vector graphics applications.  
For example, a NURBS curve can be decomposed into rational Bézier segments, and a hierarchy of ellipse bounds can be built for fast containment queries.

Table~\ref{tbl:bound_comparison} compares the computational cost of the ellipse bound versus the traditional control-points bound. 
The ellipse bound reduces evaluation and in-bound testing from $O(n^2)$ and $O(n)$ to $O(n)$ and $O(1)$, respectively. 
A similar hierarchy can also be evaluated in constant time. 
We note that concurrent work by~\citet{wozny2025fast} proposed a stable $O(n\log n)$ subdivision method based on the Fast Fourier Transform, but it applies only to non-rational Bézier curves; rational curves still require $O(n^2)$ subdivision.

\begin{table}[h]
\centering
\caption{Comparison of computational cost between ellipse bounds and control-points bounds, where $n$ is the number of control points.}
\begin{tabular}{ccc}
\toprule
& Ellipse Bound & Control-Points Bound \\
\midrule
Evaluation      & $O(n)$   & $O(n^2)$ \\
In-bound Test   & $O(1)$   & $O(n)$   \\
Path Hierarchy  & $O(1)$   & $O(n)$   \\
\bottomrule
\end{tabular}
\label{tbl:bound_comparison}
\end{table}

\paragraph{Relation to arc trees.}
Our ellipse bound differs from that used in arc trees~\citep{gunther1990arc}, which rely on arc length. 
Computing arc-length parameterizations for general Bézier curves is often intractable~\citep{sakkalis2009non}, making our derivative-based ellipse bound more practical for CAD applications.

\subsection{Periodic Surfaces and Their Covering Space}
\label{sec:covering}

Periodic parameterizations are ubiquitous in CAD models. 
Seamless cylinders, tori, and their deformations are represented by surfaces whose parameter domains are periodic rectangles. 
From a topological viewpoint, such surfaces fall into three main categories:
\begin{itemize}
  \item \textbf{Cylinder-like:} represented by a uni-periodic domain.
  \item \textbf{Torus-like:} represented by a bi-periodic domain.
  \item \textbf{Sphere or cone-like:} represented by a uni-periodic domain with one or more singular points.
\end{itemize}
For surfaces homeomorphic to the sphere or cone, we can puncture the singular points, effectively reducing them to the cylinder-like case.

Evaluating the winding number directly in the parameter domain is problematic because trimming curves may cross the domain boundary or appear artificially open due to periodic identification. 
To address this, we interpret trimming loops within the viewpoint of the \emph{universal covering space}.

The universal covering space of a uni- or bi-periodic rectangle is the infinite plane, tiled by translated copies of the base domain. 
Tiles repeat along one axis for uni-periodic domains and along both axes for bi-periodic domains. 
In this lifted space (Figure~\ref{fig:covering}):
\begin{itemize}
  \item Contractible loops remain contractible.
  \item Non-contractible loops lift to \emph{periodically extended curves}, which repeat infinitely across tiles.
\end{itemize}
This perspective transforms boundary-crossing trimming curves into well-defined entities whose winding numbers can be computed consistently.

\begin{figure}[h]
  \centering
  \includegraphics[width=\linewidth]{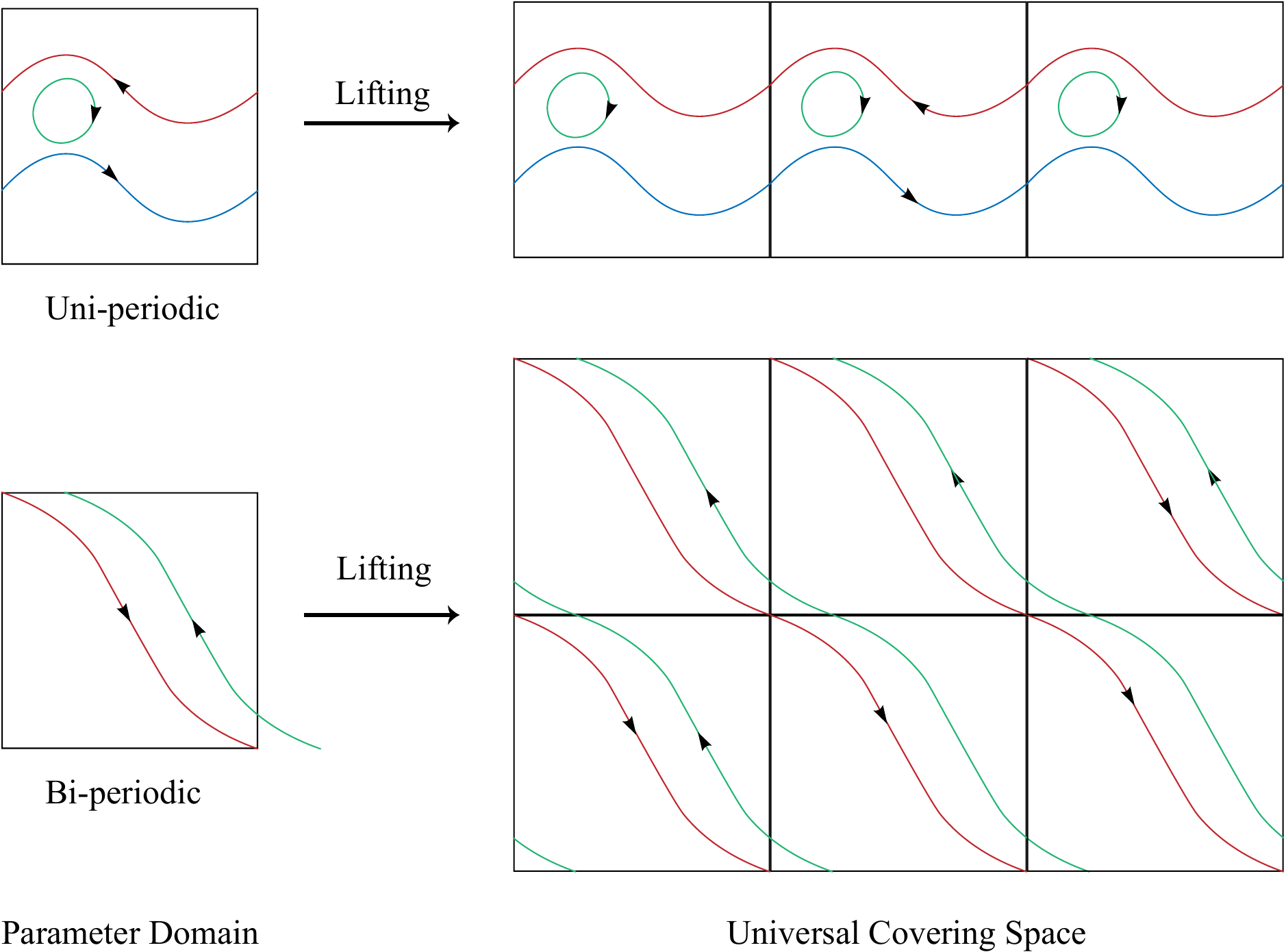}
  \caption{Universal covering spaces of periodic domains. 
  Contractible loops lift to closed loops, while non-contractible loops become periodically extended curves that repeat across tiles.}
  \label{fig:covering}
\end{figure}

\subsection{Periodically Extended Curves}
\label{sec:periodic_curves}

We now formalize the concept of periodically extended curves and show how to compute their winding numbers. 
The key idea is to truncate the infinite set of copies: segments sufficiently far from the query point can be replaced by straight lines without changing their contribution.

\paragraph{Definition.}
Let $\gamma : \mathbb{R} \to \mathbb{R}^2$ be a curve, and assume there exists a period $T \in \mathbb{R}$ such that
\begin{equation}
    \gamma(t + T) = \gamma(t) + \vec{v},
\end{equation}
for some nonzero translation vector $\vec{v} \in \mathbb{R}^2$. 
We call $\gamma$ a \emph{periodically extended curve} with \emph{extension vector} $\vec{v}$. 
We partition $\gamma$ into unit segments:
\[
  \gamma_n : [0, T] \to \mathbb{R}^2, \qquad \gamma_n(t) = \gamma(nT + t).
\]
Let $\beta_n = \gamma_n(0)$; all $\beta_n$ lie on a line $L$ parallel to $\vec{v}$. 
The winding number of $\gamma$ with respect to a test point $p$ can be written as an infinite series:
\begin{equation}
  \label{equ:periodic_wn_series}
  \omega(p, \gamma) = \sum_{k=-\infty}^{\infty} \omega(p, \gamma_k).
\end{equation}

Although Equation~\eqref{equ:periodic_wn_series} is not directly computable, the convergence is effectively finite. 
Since $\|\vec{v}\| > 0$, there exists $N$ such that for all $|k| > N$, the ellipse bounds of $\gamma_k$ do not contain $p$. 
By homotopy invariance, the contributions from these distant copies can be replaced by straight segments between $\beta_k$ and $\beta_{k+1}$. 
Hence,
\begin{equation}
    \label{equ:periodic_wn}
    \begin{aligned}
        \omega(p, \gamma) 
        &= \sum_{k=-N}^{N} \omega(p, \gamma_k)
         + \sum_{|k|>N} \omega(p, \gamma_k) \\[2pt]
        &= \sum_{k=-N}^{N} \omega(p, \gamma_k)
         + \sum_{|k|>N} \omega(p, \beta_k, \beta_{k+1}) \\[2pt]
        &= \sum_{k=-N}^{N} \omega(p, \gamma_k)
         + \omega(p, L) - \omega(p, \beta_{-N}, \beta_{N+1}).
    \end{aligned}
\end{equation}

Here $\omega(p,L)$ denotes the winding number of the infinite line $L$ with respect to $p$, computed by a simple left–right test: 
if $p$ lies on the left side of $L$, $\omega(p,L)=1/2$; otherwise, $\omega(p,L)=-1/2$. 
The final correction term $\omega(p,\beta_{-N},\beta_{N+1})$ subtracts the spurious contribution of the long chord spanning the truncated region. 
Figure~\ref{fig:periodic_wn} illustrates this construction.

\begin{figure}[h]
  \centering
  \includegraphics[width=\linewidth]{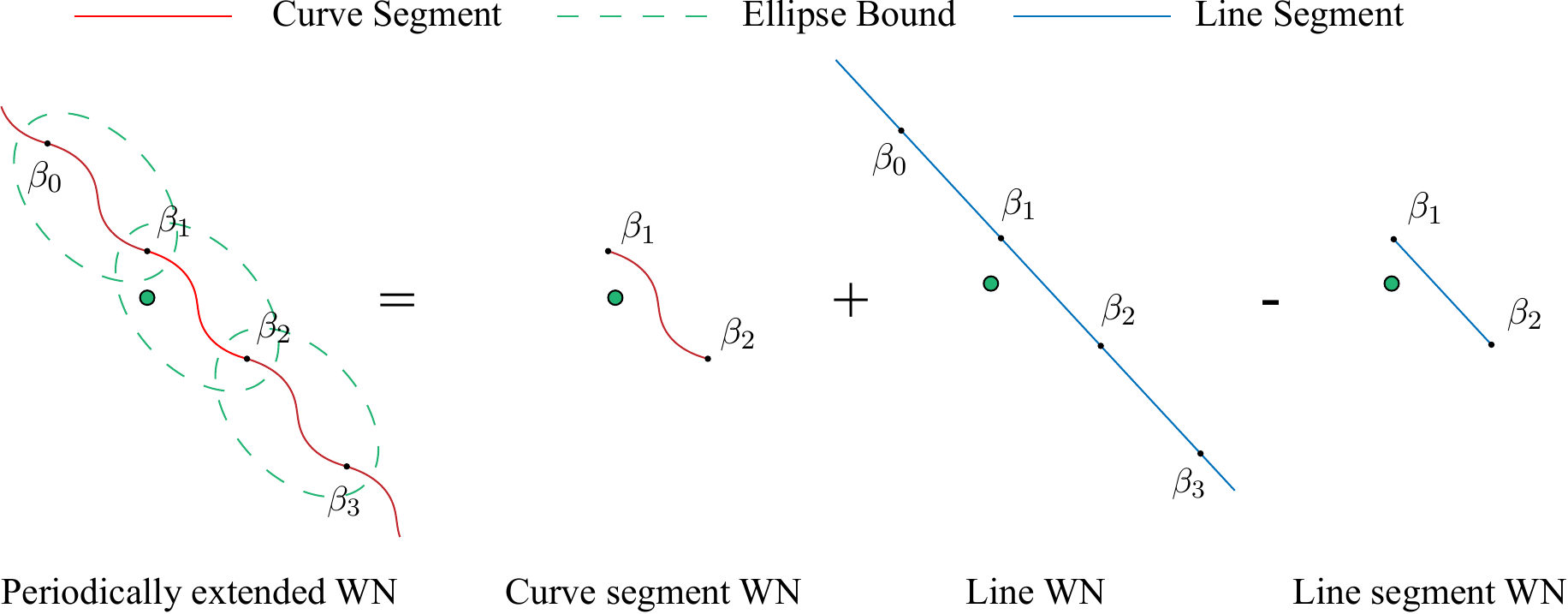}
  \caption{Winding number computation for a periodically extended curve. 
  Nearby copies of $\gamma$ are evaluated using the standard winding number primitive, while distant copies are replaced by straight-line contributions.}
  \label{fig:periodic_wn}
\end{figure}

In practice, only a small number of local segments of $\gamma$ need to be evaluated using the recursive ellipse-bound algorithm. 
All remaining copies reduce to a single straight-line contribution. 
This approach ensures that winding number queries on periodic domains remain robust, efficient, and mathematically consistent.

\begin{figure*}[ht]
  \centering
  \includegraphics[width=\linewidth]{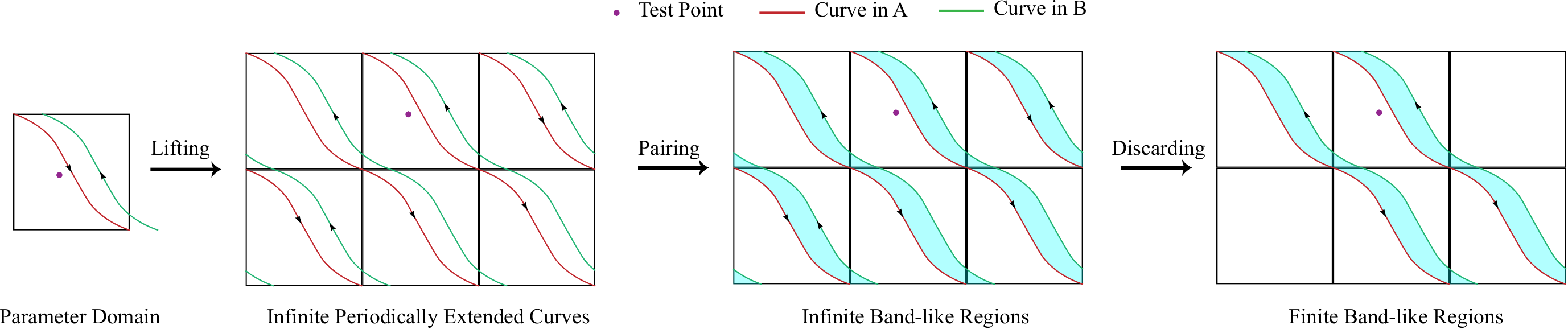}
  \caption{Pipeline of our algorithm on a bi-periodic domain. 
  Following Proposition~\ref{thm:main}, the p-curves are divided into two sets, $A$ and $B$. 
  Lifting them to the universal covering space yields infinite periodically extended curves, which are paired to form band-like regions. 
  Regions far from the test point are discarded, and winding numbers are computed using the primitives described in Section~\ref{sec:periodic_curves}.}
  \label{fig:bi_covering}
\end{figure*}

\subsection{Periodic Domain}
\label{sec:periodicity}

With the framework described in Section~\ref{sec:periodic_curves}, containment queries on a \emph{uni-periodic} domain can be solved directly. 
Let the tiles of the universal covering space be denoted by $D_n$, where $n \in \mathbb{Z}$. 
For a trimming loop $\gamma$, two cases arise:

\begin{itemize}
    \item \textbf{Contractible loop:} $\gamma$ lifts to closed loops $\gamma_n = \gamma + n\vec{e}_1$. 
    In this case, components that are far from the test point can be safely discarded.
    \item \textbf{Non-contractible loop:} the lifted copies of $\gamma$ form a periodically extended curve, and its winding number can be computed using the primitives introduced in Section~\ref{sec:periodic_curves}.
\end{itemize}

When the domain is \emph{bi-periodic}, the situation becomes more subtle. 
Each loop $\gamma$ is lifted to
\[
  \gamma + i\vec{e}_1 + j\vec{e}_2, \qquad i,j \in \mathbb{Z}.
\]
Unlike the uni-periodic case, even a contractible loop now has infinitely many lifted copies. 
A naive summation of winding numbers,
\begin{equation}
  \omega\!\left(p, \{\gamma_i\}_{1 \leq i \leq n}\right)
  = \sum_{i=-\infty}^{\infty}\sum_{j=-\infty}^{\infty}\sum_{k=1}^{n}
    \omega\!\left(p, \gamma_k + i\vec{e}_1 + j\vec{e}_2 \right),
  \label{eq:biperiodicWN}
\end{equation}
is not absolutely convergent. We therefore develop another way to properly formulate the winding number.


\paragraph{Key ideas.}
We observe that the trimming loops on a torus-like surface must adhere to topological constraints. In particular, non-contractible loops must appear in oppositely oriented pairs with the same homology coefficients. Each such pair defines a band-like region in the universal covering space—an infinitely repeating but finite-width strip. The winding number over the entire domain can then be defined as the sum of contributions from these bands, and only a finite number of bands near the test point need to be considered.
Figure~\ref{fig:bi_covering} shows the overall motivation of proposed algorithm on bi-periodic domain via a simple example.

The following proposition provides the topological foundation for our approach.

\begin{proposition}
  \label{thm:main}
  Let $S$ be a surface smoothly parametrized by a bi-periodic mapping to $\mathbb{R}^3$, and let $\gamma_{1}, \gamma_{2}, \dots, \gamma_{n}$ be disjoint, self-intersection-free loops on $S$ that together define a bounded region on $S$, i.e., $\gamma_{1} + \cdots + \gamma_{n}$ is null-homologous. 
  Then, the set $\{\gamma_i\}$ can be partitioned as
  \begin{equation}
    \{\gamma_i\} = A \cup B \cup C,
  \end{equation}
  where $A$ contains all loops with homology coefficients $(p, q)$, $B$ contains those with $-(p, q)$, and $C$ contains all contractible loops. 
  Moreover, $A$ and $B$ have the same cardinality.
\end{proposition}

Intuitively, Proposition \ref{thm:main} implies a strong structural restriction on trimming loops in a torus-like parameter domain: all non-contractible loops must appear in oppositely oriented pairs with identical homology coefficients. 
This allows us to group the loops into pairs $(\alpha, \beta)$ with $\alpha \in A$ and $\beta \in B$, and each pair $(\alpha, \beta)$ defines a band-like region in the universal covering space.
A formal proof is provided in Appendix~\ref{proof:main}.

To compute the winding number contributed by a pair $(\alpha, \beta)$, we must enumerate all their lifted copies. The following proposition shows how to partition the set of all lifted copies of a loop with homology coefficients $(p, q)$ into a family of periodically extended curves. Proof is provided in Appendix~\ref{proof:traverse}


\begin{proposition}
  \label{thm:traverse_version2}
  For a loop $\gamma$ with homology coefficients $(p, q)$, the set of all its lifted copies $\{ \gamma + i \vec{e}_1 + j \vec{e}_2 \}_{i, j \in \mathbb{Z}}$
  can be partitioned into disjoint periodically extended curves:
  \begin{equation}
    \bigcup_{i=-\infty}^{\infty}\bigcup_{j=-\infty}^{\infty}
      (\gamma + i\vec{e}_1 + j\vec{e}_2)
    = \bigcup_{r=0}^{p-1}\bigcup_{s=-\infty}^{\infty} \gamma_{r,s},
  \end{equation}
  where $\gamma_{r,s}$ is defined as
  \begin{equation}
    \gamma_{r,s} = \sum_{l=-\infty}^{\infty} 
      \big(\gamma + (r + pl)\vec{e}_1 + (s + ql)\vec{e}_2\big),
    \quad 0 \leq r \leq p - 1,\; s \in \mathbb{Z}.
    \label{equ:liftedgamma}
  \end{equation}
\end{proposition}

Now, consider a pair of loops $\alpha$ and $\beta$ with homology coefficients $(p,q)$ and $-(p,q)$, respectively. 
Their lifted copies $\alpha_{i,j}$ and $\beta_{i,j}$ bound infinitely repeating band-like regions in the covering space. 
The winding number of this pair is defined as
\begin{equation}
  \omega(p, \{\alpha, \beta\})
  = \sum_{i=0}^{p-1}\sum_{j=-\infty}^{\infty}
      \big(\omega(p, \alpha_{i,j}) + \omega(p, \beta_{i,j})\big).
  \label{eq:bandWN}
\end{equation}
As illustrated in Figure~\ref{fig:bi_covering}, only finitely many of these bands intersect the base domain $[0,1]\times[0,1]$, so the summation reduces to a finite set of terms. 
Each term $\omega(p, \gamma_{i,j})$ can be evaluated efficiently using the primitives described in Section~\ref{sec:periodic_curves}.

\paragraph{Summary.}
By enumerating all loop pairs in $A$ and $B$ (as defined in Proposition~\ref{thm:main}) and summing their contributions according to Equation~\eqref{eq:bandWN}, we obtain the total winding number for all non-contractible loops. 
The complete pairing algorithm, including pseudo-code and special cases, is provided in Appendix~\ref{sec:code}. 
Although the derivation assumes continuous and disjoint loops, the method remains robust under typical geometric imperfections, such as disconnected or overlapping p-curves (see Figure~\ref{fig:graphical_abstract}).

\section{Experimentation}
\label{sec:experimentation}

\subsection{Implementation}

We implemented the proposed algorithm in C++, and all experiments were executed on a single thread. 
Runtimes were measured on a PC equipped with an Intel Core i9-13900HX CPU and 32\,GB of RAM. 

\subsection{2D Containment Query with Parametric Boundaries}

We first evaluate the proposed winding number algorithm on planar containment queries for parametrically represented curves in a simply connected domain.

\subsubsection{Performance}

We compare our method with both ray-casting and winding number algorithms, including Axom~\cite{spainhour_robust_2024}, ClosedForm~\cite{liu2025closed}, and OneShot~\cite{martens2025one}. 
For the ray-casting baseline, we adopt the pruning strategy of~\citet{nishita1990ray} and use a state-of-the-art Bézier root solver~\cite{machchhar2016revisiting}. 
The default ray direction is along the positive $u$-axis.

\begin{figure}[t]
    \centering
    \includegraphics[width=\linewidth]{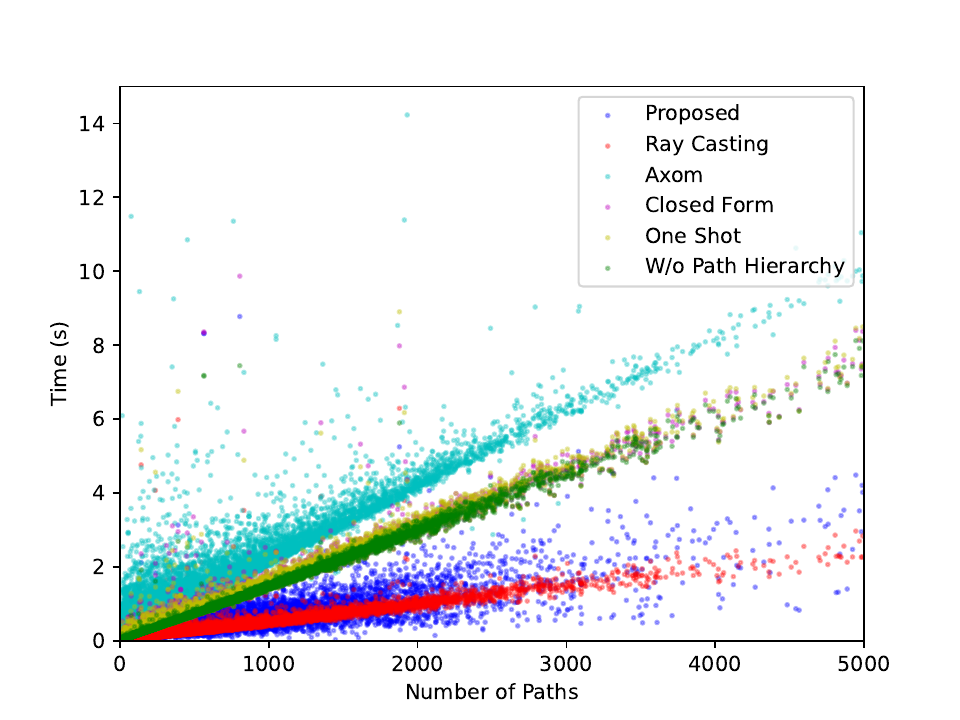}
    \caption{
    Runtime comparison across methods on the OpenClipArt20K dataset~\cite{hu2019triwild}.
    Extremely large or small values are clipped for readability. 
    The proposed method achieves performance comparable to non-robust ray casting and significantly outperforms recent robust winding number approaches such as Axom~\cite{spainhour_robust_2024}.
    }
    \label{fig:openclip}
\end{figure}

We evaluate performance on two datasets: the OpenClipArt20K dataset~\cite{hu2019triwild} and the publicly available Bézier Guarding (BG) dataset~\cite{mandad2020bezier}. 
OpenClipArt20K consists of 20,000 SVG files containing both simple and highly complex shapes. 
We extract the Bézier paths from each SVG, normalize them to the unit square $[0,1]\times[0,1]$, and treat them as shape boundaries. 
The BG dataset includes four categories of randomly generated Bézier curves, as illustrated in Figure~\ref{fig:bg_dataset}. 
For OpenClipArt20K, we preserve the original SVG path hierarchy, whereas for the BG dataset, we also evaluate our method without hierarchy acceleration. 
For each shape, $256\times256$ sample points are uniformly distributed across the parametric domain, and all algorithms are evaluated under identical settings.

Table~\ref{tbl:openclip} reports the average runtime, and Figure~\ref{fig:openclip} shows the runtime distribution on OpenClipArt20K. 
Compared to Axom~\cite{spainhour_robust_2024}, our method significantly accelerates winding number computation, reducing runtime to \textbf{27.3\%} of Axom's on real-world SVG shapes. 
The performance is also comparable to the optimized, yet non-robust, ray-casting method~\cite{nishita1990ray}. 
For BG subsets~A and~B, both methods perform similarly, while on subsets~C and~D, the pruning strategy in~\cite{nishita1990ray} gives ray casting a strong advantage due to the sparse distribution of curves.

\begin{table}[h]
\centering
\caption{
Average runtime (seconds) of different algorithms on OpenClipArt20K~\cite{hu2019triwild} and BG~\cite{mandad2020bezier} datasets.
Our method reduces winding number computation time to \textbf{27.3\%} of Axom's~\cite{spainhour_robust_2024} while achieving performance comparable to an optimized ray-casting approach.
}
\begin{tabular}{@{}lccccc@{}}
\toprule
Dataset & Ray & Axom & Closed & OneShot & Proposed \\
\midrule
OpenClipArt20K     & 0.318 & 1.488 & 0.890 & 0.903 & 0.407 \\
BG Dataset A        & 0.076 & 0.360 & 0.115 & 0.133 & 0.095 \\
BG Dataset B        & 0.091 & 0.426 & 0.125 & 0.149 & 0.104 \\
BG Dataset C        & 0.066 & 0.257 & 0.184 & 0.172 & 0.171 \\
BG Dataset D        & 0.045 & 0.190 & 0.128 & 0.119 & 0.114 \\
\bottomrule
\end{tabular}
\label{tbl:openclip}
\end{table}

\begin{figure}[t]
    \centering
    \includegraphics[width=\linewidth]{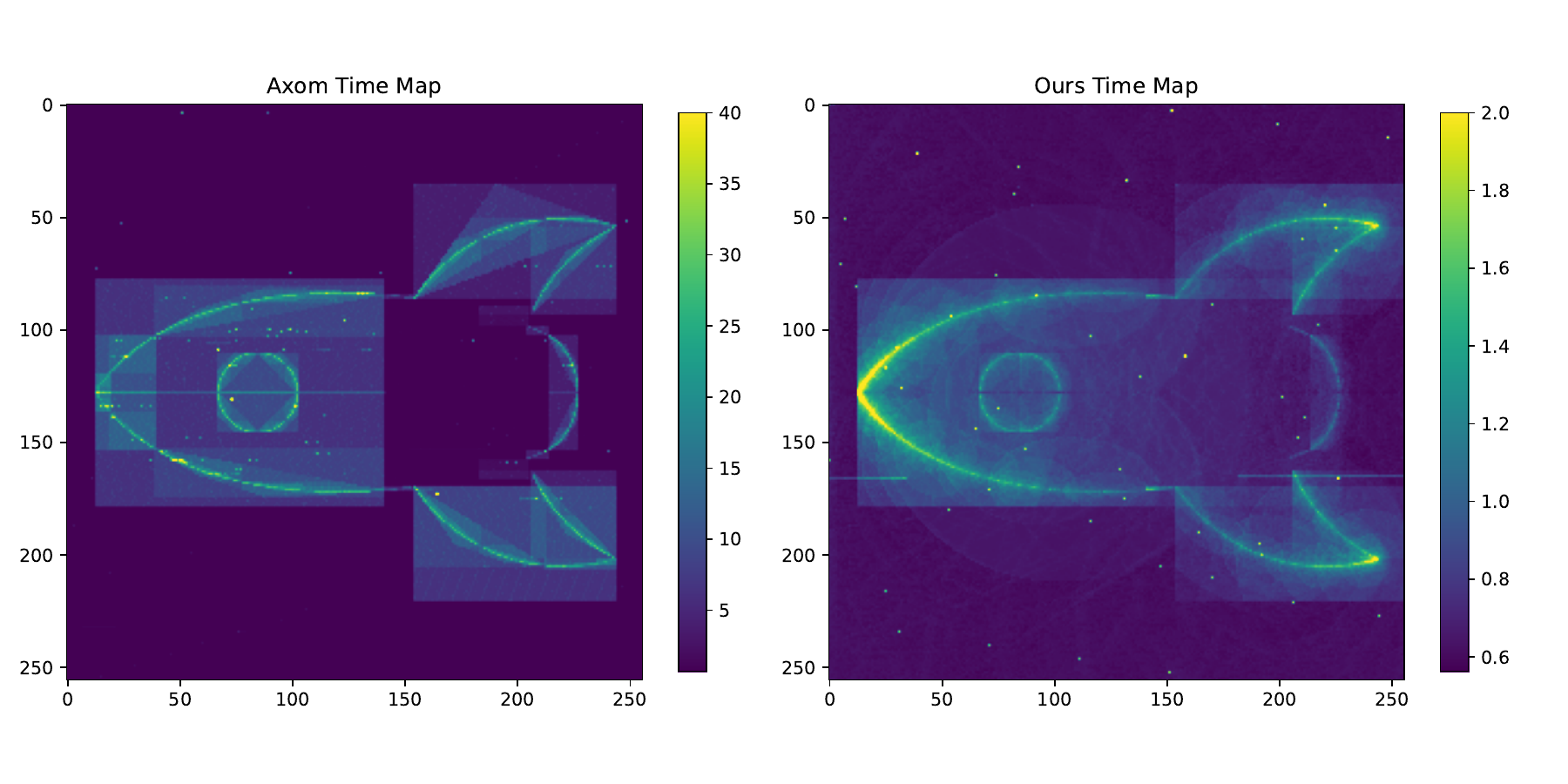}
    \caption{
    Spatial distribution of runtime cost for Axom and the proposed method on the parameter domain. 
    Different colormaps are used for clarity: the left ranges from 0–40\,ms and the right from 0–2\,ms. 
    The proposed method exhibits more uniform performance across the domain.
    }
    \label{fig:single_time_distribution}
\end{figure}

As an ablation study, Figure~\ref{fig:openclip} also reports performance without the path hierarchy. 
The hierarchy significantly accelerates winding number computation, particularly for complex shapes. 
In such cases, the benefit of linear evaluation becomes less pronounced, as all boundaries are cubic Bézier curves.

\begin{figure}[t]
\centering
\includegraphics[width=\linewidth]{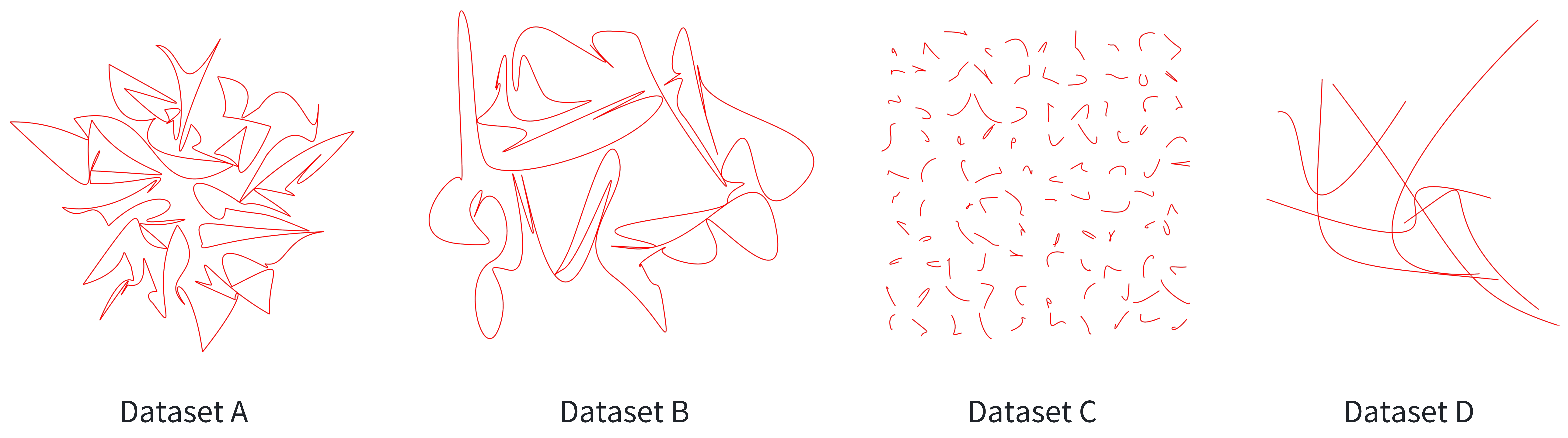}
\caption{
Examples from the Bézier Guarding dataset~\cite{mandad2020bezier}. 
Dataset~A: Closed $C^0$ boundaries. 
Dataset~B: Closed $C^1$ boundaries. 
Dataset~C: Isolated random curves uniformly distributed across the domain. 
Dataset~D: Randomly intersecting curves.
}
\label{fig:bg_dataset}
\end{figure}

We further analyze the asymptotic behavior of different algorithms. 
Using the BG Dataset~A~\cite{mandad2020bezier}, we elevate cubic curves to higher degrees without altering geometry. 
Test points are sampled directly from the boundary—the worst case for recursive methods such as ours and~\citet{spainhour_robust_2024}. 
Figure~\ref{fig:highorder} shows that our algorithm demonstrates empirically linear asymptotic behavior, while the control-point-based method exhibits faster growth in runtime.

\begin{figure}[t]
    \centering
    \caption{
    Runtime versus boundary degree for different algorithms. 
    The proposed method exhibits empirically linear scaling, while the control-point-based approach~\cite{spainhour_robust_2024} grows faster with increasing curve order.
    }
    \includegraphics[width=\linewidth]{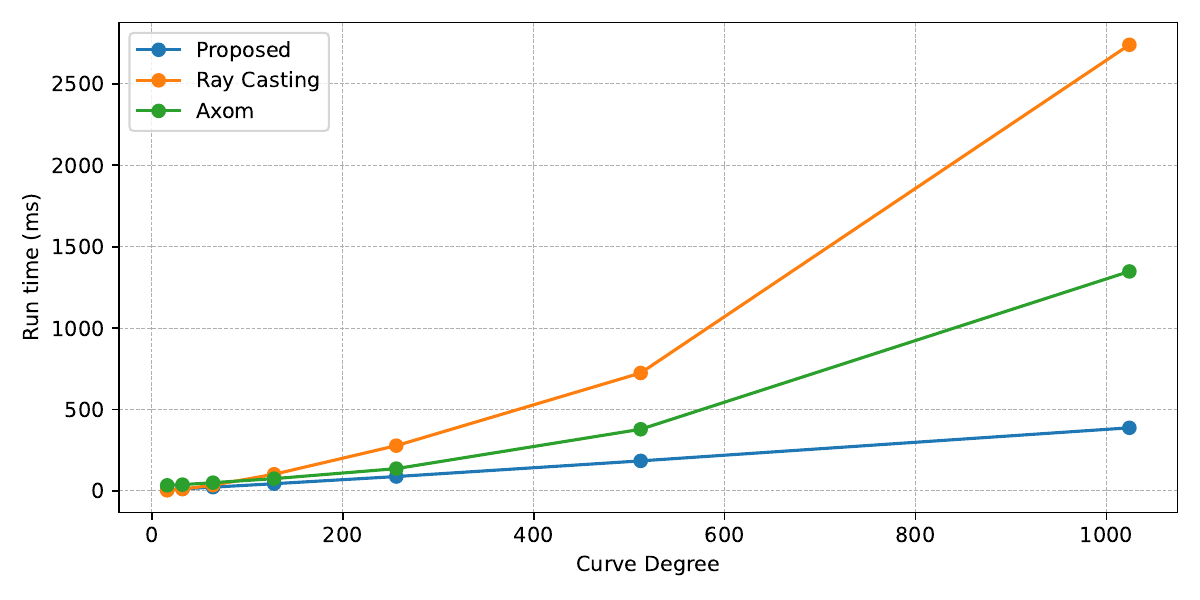}
    \label{fig:highorder}
\end{figure}

We also evaluate the effect of tolerance $\epsilon$ using BG Dataset~A. 
Our algorithm is tested under uniform sampling and on-boundary sampling. 
Table~\ref{tbl:tolerance} lists the runtime and the number of curve evaluations for each tolerance value. 
Uniform sampling shows negligible dependence on $\epsilon$, while boundary sampling exhibits an empirical logarithmic trend. 
Figure~\ref{fig:single_time_distribution} visualizes the global distribution of computational cost.

\begin{table}[h!]
    \centering
    \caption{Effect of tolerance $\epsilon$ on performance. Runtime (in seconds) and number of curve evaluations for classifying 1 million points on BG Dataset A, under uniform and on-boundary sampling. The results demonstrate the method's robustness and its logarithmic dependence on $\epsilon$ for boundary-adjacent points.}
    \begin{tabular}{cccccc}
    \toprule
    $\epsilon$            & $10^{-4}$ & $10^{-5}$ & $10^{-6}$ & $10^{-7}$ & $10^{-8}$ \\
    \midrule
    Boundary eval count   & 35.95 & 45.02 & 54.18 & 63.41 & 72.46 \\
    Boundary time (ms)    & 2.60  & 2.95  & 3.26  & 3.68  & 4.05  \\
    Uniform eval count    & 16.22 & 16.28 & 16.29 & 16.29 & 16.29 \\
    Uniform time (s)      & 0.086 & 0.090 & 0.091 & 0.090 & 0.086 \\
    \bottomrule
    \end{tabular}
    \label{tbl:tolerance}
\end{table}

\subsubsection{Robustness}

\begin{figure}[t]
    \centering
    \includegraphics[width=\linewidth]{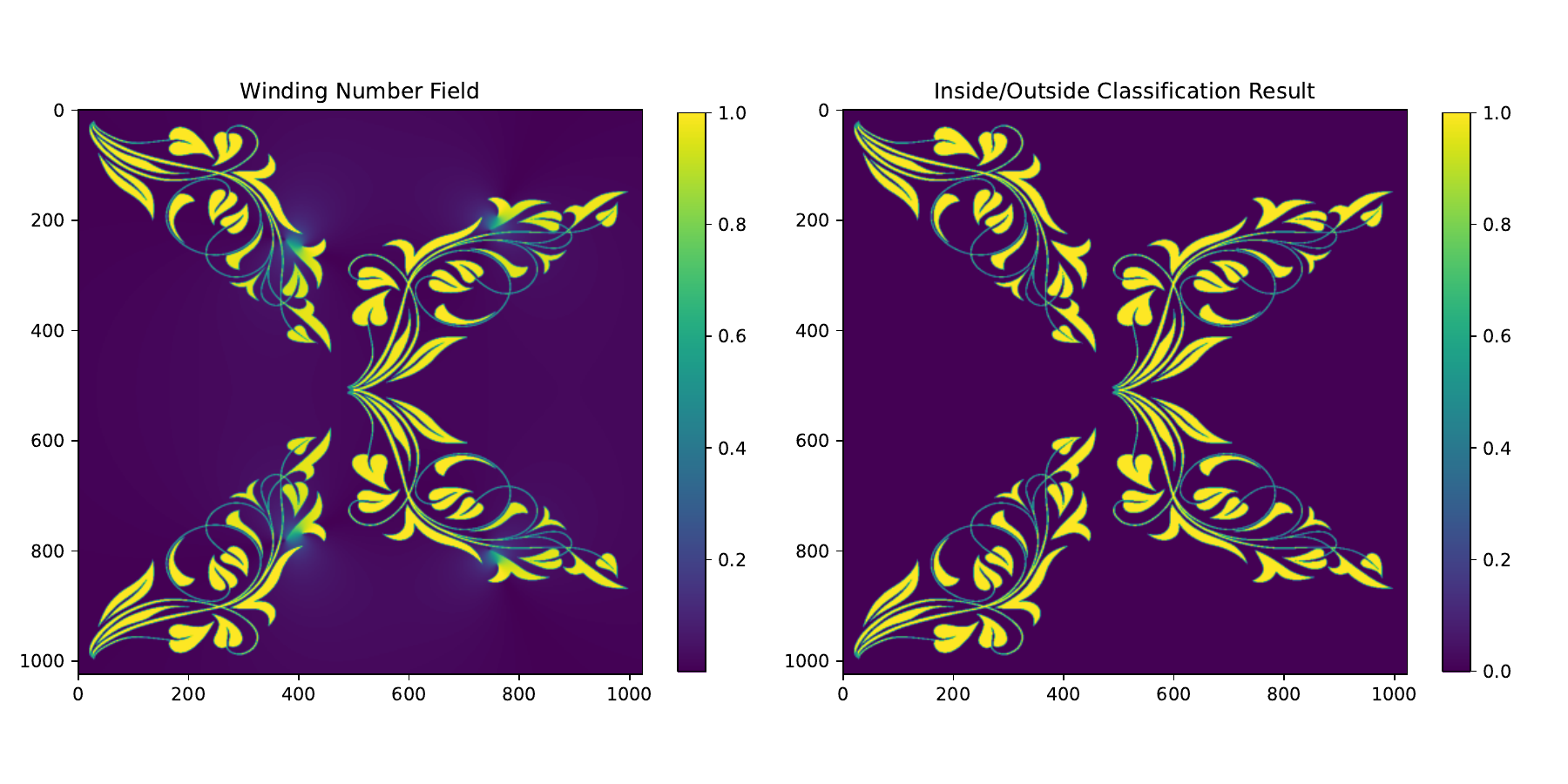}
    \caption{
    An extremely complex shape from OpenClipArt (ID\,=\,237213). 
    Some boundary curves are missing, yet the winding number classification automatically corrects these geometric errors.
    }
    \label{fig:open_boundary}
\end{figure}

From a robustness perspective, ray casting fails even on simple shapes with minor geometric inconsistencies. 
As shown in Figure~\ref{fig:robustness}, both methods perform correctly on clean, closed shapes. 
However, when random noise is added to boundary control points—introducing non-manifold or open boundaries—the ray-casting approach fails, while winding number methods remain stable. 
Figure~\ref{fig:open_boundary} further demonstrates that our method can handle extremely complex open-boundary shapes, automatically repairing gaps through the topological consistency of the winding number.

\begin{figure}[t]
    \centering
    \includegraphics[width=\linewidth]{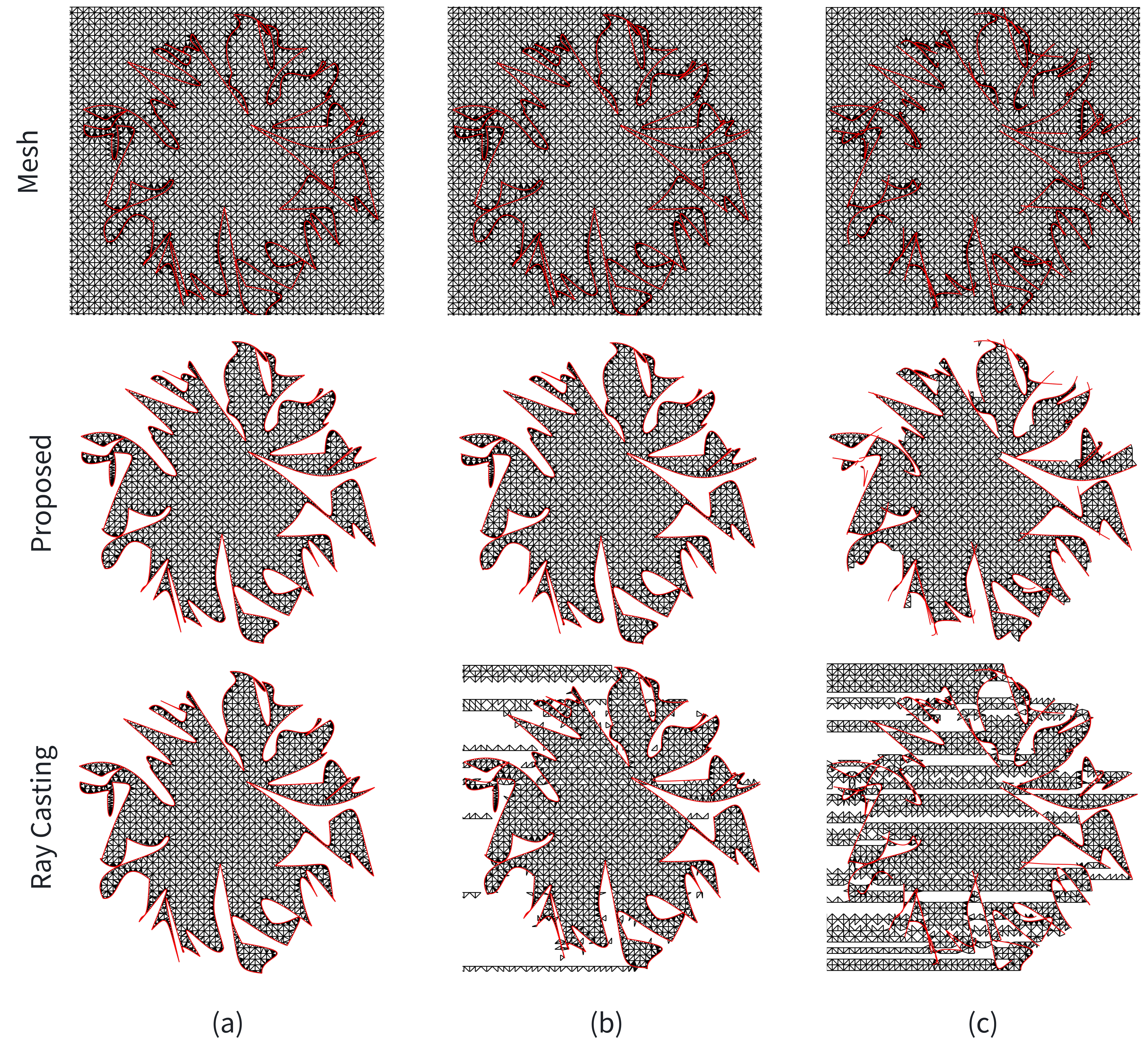}
    \caption{
    Robustness comparison using BG Dataset~A (ID\,=\,694). 
    The boundary is discretized and the domain triangulated via constrained Delaunay triangulation; containment queries are performed at triangle barycenters. 
    (a) Closed boundary: both methods succeed. 
    (b–c) Boundaries perturbed with Gaussian noise $\mathcal{N}(0,10^{-3})$ and $\mathcal{N}(0,10^{-2})$. 
    Ray casting fails under non-manifold or open boundaries, whereas the winding number remains stable.
    }
    \label{fig:robustness}
\end{figure}

\subsubsection{Correctness and Accuracy}

To assess accuracy and numerical stability, we applied the proposed containment-query algorithm to BG Dataset~A under varying tolerance values. 
Results are shown in Figure~\ref{fig:boundary_tolerance}. 
As $\epsilon$ decreases, the fraction of points classified as \texttt{on-boundary} quickly approaches zero. 
For all points confidently classified as inside or outside, the maximum deviation in winding number is $3.95 \times 10^{-9}$ relative to the ground truth computed by~\citet{liu2025closed}.

\begin{figure}[t]
    \centering
    \includegraphics[width=\linewidth]{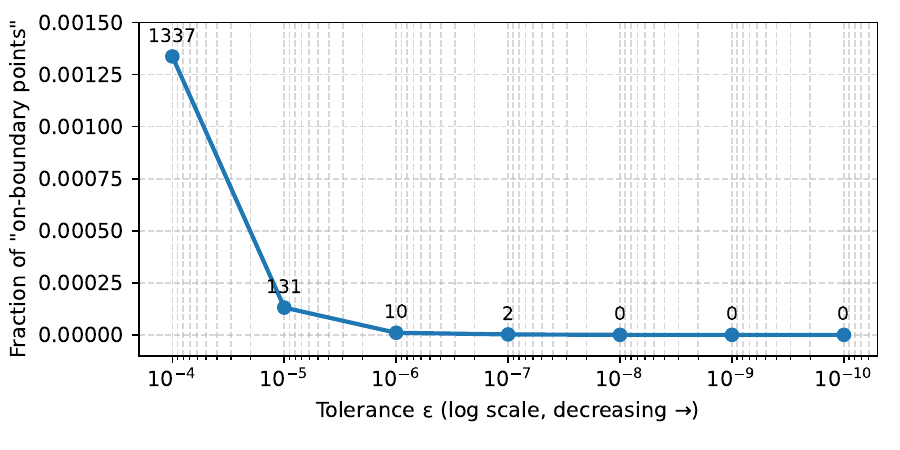}
    \caption{
    Fraction of points reported as \texttt{on-boundary} among $10^6$ uniformly sampled test points, for different tolerance values $\epsilon$. 
    Each marker is annotated with the number of boundary detections. 
    As $\epsilon$ decreases from $10^{-4}$ to $10^{-10}$, the ambiguous fraction rapidly vanishes—from 0.13\% to 0\%.
    }
    \label{fig:boundary_tolerance}
\end{figure}

\subsection{Periodic Domain}

In this section, we evaluate the proposed algorithm on periodic domains. 
Unlike the simply connected case, regions in a periodic domain may have non-closed boundaries. 
Although conventional winding-number methods can tolerate small geometric inconsistencies such as open boundaries, they still fail on many periodic configurations—examples of which were already shown in Figure~\ref{fig:open_boundary}.

Using the formulation presented in Section~\ref{sec:periodicity}, we can efficiently and correctly compute winding numbers in periodic domains. 
As illustrated in Figure~\ref{fig:bi_periodic_comparison}, directly applying conventional winding number evaluation for inside/outside classification fails for both non-closed $p$-curves and curves crossing domain boundaries. 
In such cases, directly computed winding number often misclassifies points on the trimmed surface as ``outside,'' and regions crossing periodic boundaries are lost. 
In contrast, our periodic formulation correctly handles both non-closed and boundary-crossing $p$-curves, preserving the complete trimmed surface.

\begin{figure*}[t]
\centering
\includegraphics[width=\linewidth]{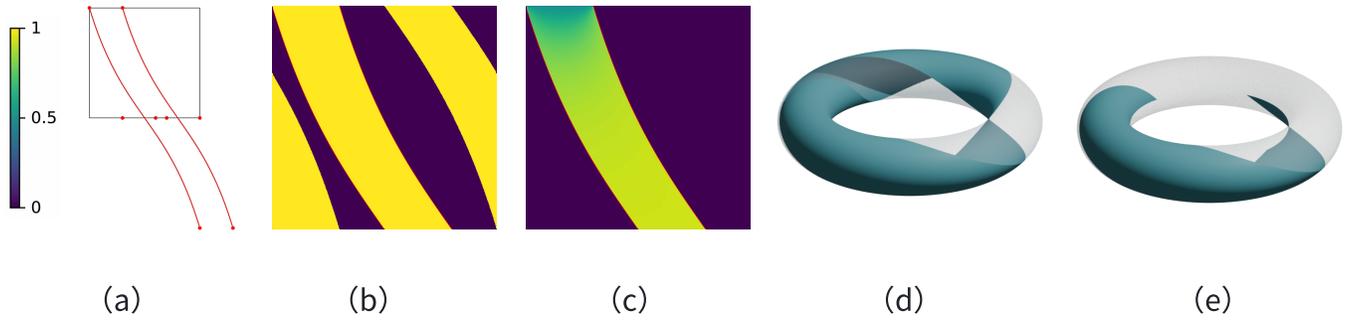}
\caption{
Comparison between our method and direct winding-number computation on a bi-periodic domain trimmed by double-separation curves. 
(a) Parameter domain and $p$-curve loop. 
(b) Winding number computed in the universal covering space. 
(c) Winding number computed in the original parameter domain. 
(d) Containment query result using the proposed method—showing complete trimmed region. 
(e) Result of direct computation in the physical domains, where part of the region is missing.
}
\label{fig:bi_periodic_comparison}
\end{figure*}

To further assess the correctness and performance of the periodic winding-number algorithm, we constructed a synthetic dataset consisting of two parts:
\begin{itemize}
    \item \textbf{Dataset A:} Non-contractible loops in a uni-periodic domain.
    \item \textbf{Dataset B:} Non-contractible loops in a bi-periodic domain with homology coefficients satisfying $2 \leq p, q \leq 5$.
\end{itemize}
Figures~\ref{fig:periodic_result_a} and~\ref{fig:periodic_result_b} visualize the computed winding numbers for both datasets. 
Compared with direct evaluation, our periodic formulation yields correct and topologically consistent results in all cases.



\subsection{Application}

\subsubsection{Processing Multiple Patches}

We demonstrate our method on boundary representation (B-Rep) models consisting of multiple surface patches. 
As illustrated in Figures~\ref{fig:containment_demo} and~\ref{fig:containment_demo2}, we construct B-Rep models by applying Boolean union operations between solids, resulting in multi-patch surfaces. 
Sample points are generated on the original geometric surfaces, and containment queries are then used to determine whether each point lies on the newly formed faces.

For a sample point $p$ on a geometric surface $S$, containment is determined by testing $p$ against all trimmed surfaces referencing $S$ as their base geometry. 
If $p$ lies within any trimming loop, it is labeled as \textit{inside}; otherwise, it is classified as \textit{outside}. 
Figure~\ref{fig:complex_models} further shows results on real-world CAD models. 
Even for complex configurations, our method robustly identifies points trimmed away by face boundaries.

\subsubsection{B-Rep Intersection}

We demonstrate the application of the proposed algorithm to \emph{boundary representation (B-Rep) intersection}, a fundamental modeling operation in CAD systems.  
The input consists of two B-Rep solids, and the goal is to compute the set of intersection curves between them.  
This process typically involves two stages:

\begin{itemize}
  \item \textbf{Geometric intersection:} computing intersection curves between the underlying untrimmed surfaces of the two solids, often using the marching method~\cite{barnhill1990marching}.
  \item \textbf{Topological intersection:} filtering the raw geometric intersections using the trimming boundaries of each surface to retain only the valid segments.
\end{itemize}

In the second stage, intersection curves are segmented by the trimming boundaries of the participating surfaces.  
To determine whether each resulting segment lies within the trimmed region, we evaluate a point containment query at its midpoint.

Figure~\ref{fig:intersection} illustrates an example where an ``eight''-shaped surface (genus two, represented by four torus patches) intersects with a wave surface.  
Stage~1 produces raw intersection loops on the untrimmed surfaces.  
Stage~2 splits these loops at the trimming boundaries (Fig.~\ref{fig:intersection}c,d).  
Containment queries then discard invalid segments (shown in red), yielding the final valid intersection curves.  
We also present a case with trimming curves wrapping in both parametric directions in Figure~\ref{fig:intersection2}.

\begin{figure*}[h]
  \centering
  \subfigbottomskip=2pt
  \subfigcapskip=-5pt
  \subfigure[Input models]{
      \includegraphics[width=0.23\linewidth]{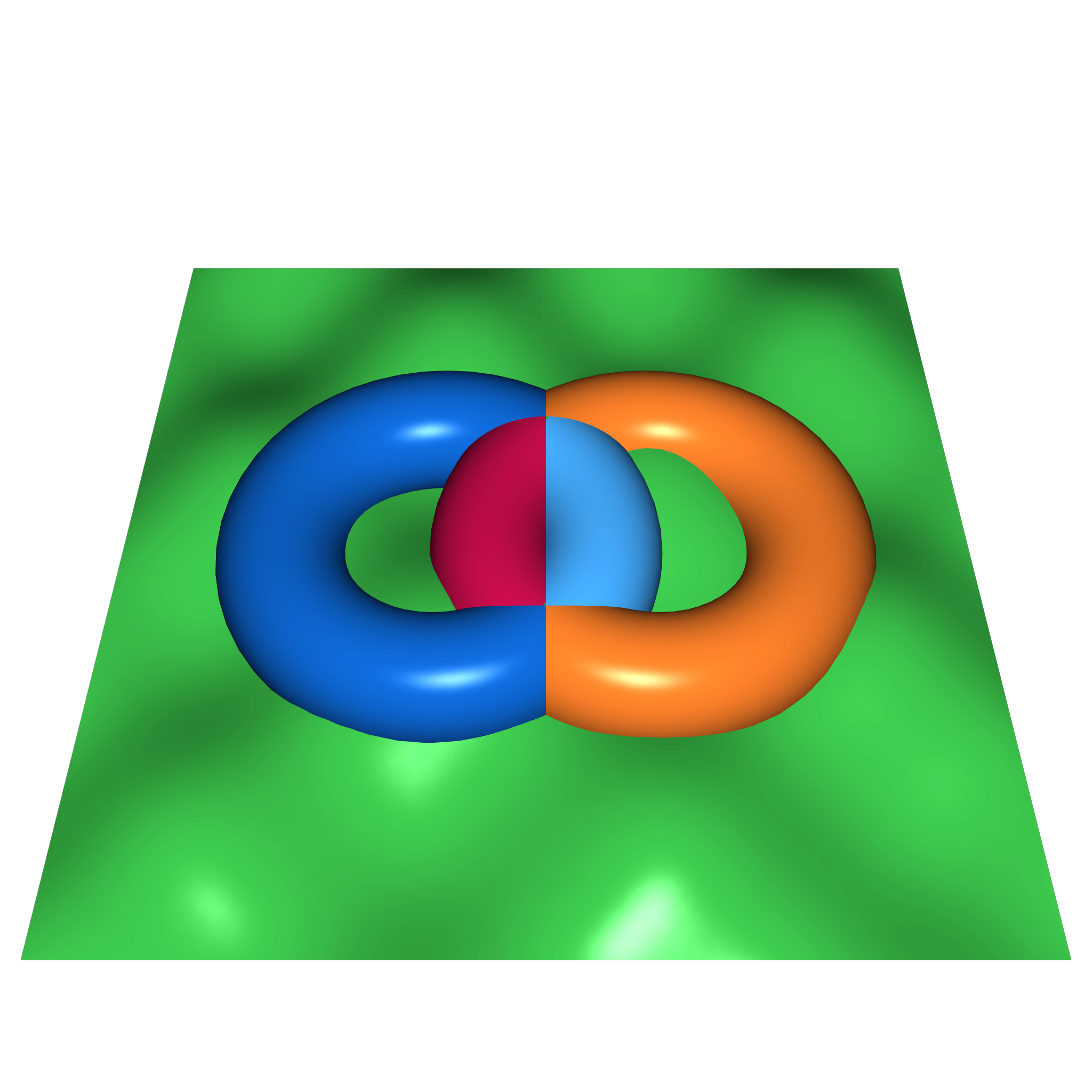}}
  \subfigure[Intersection curves]{
      \includegraphics[width=0.23\linewidth]{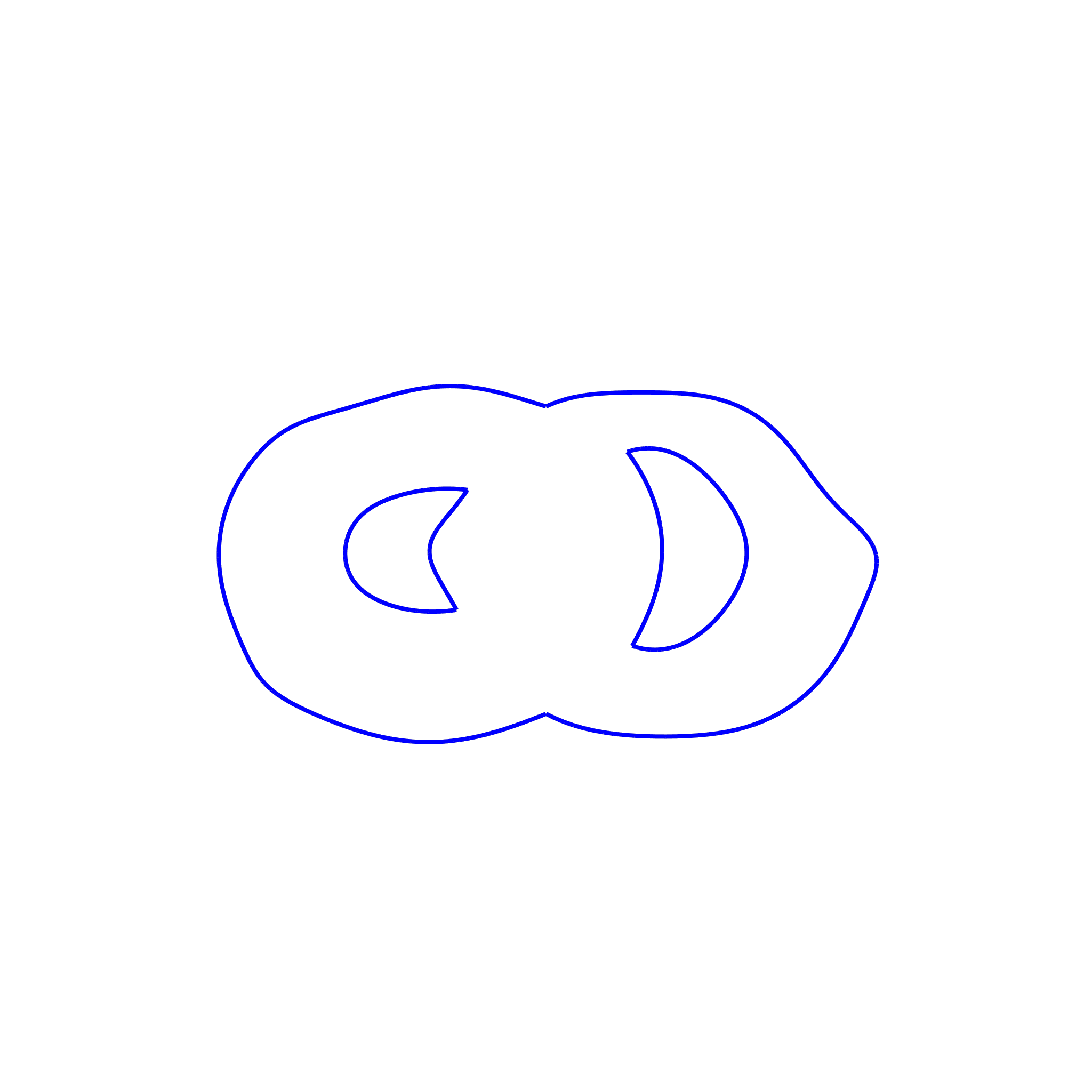}}
  \subfigure[Patch 1 with intersections]{
      \includegraphics[width=0.23\linewidth]{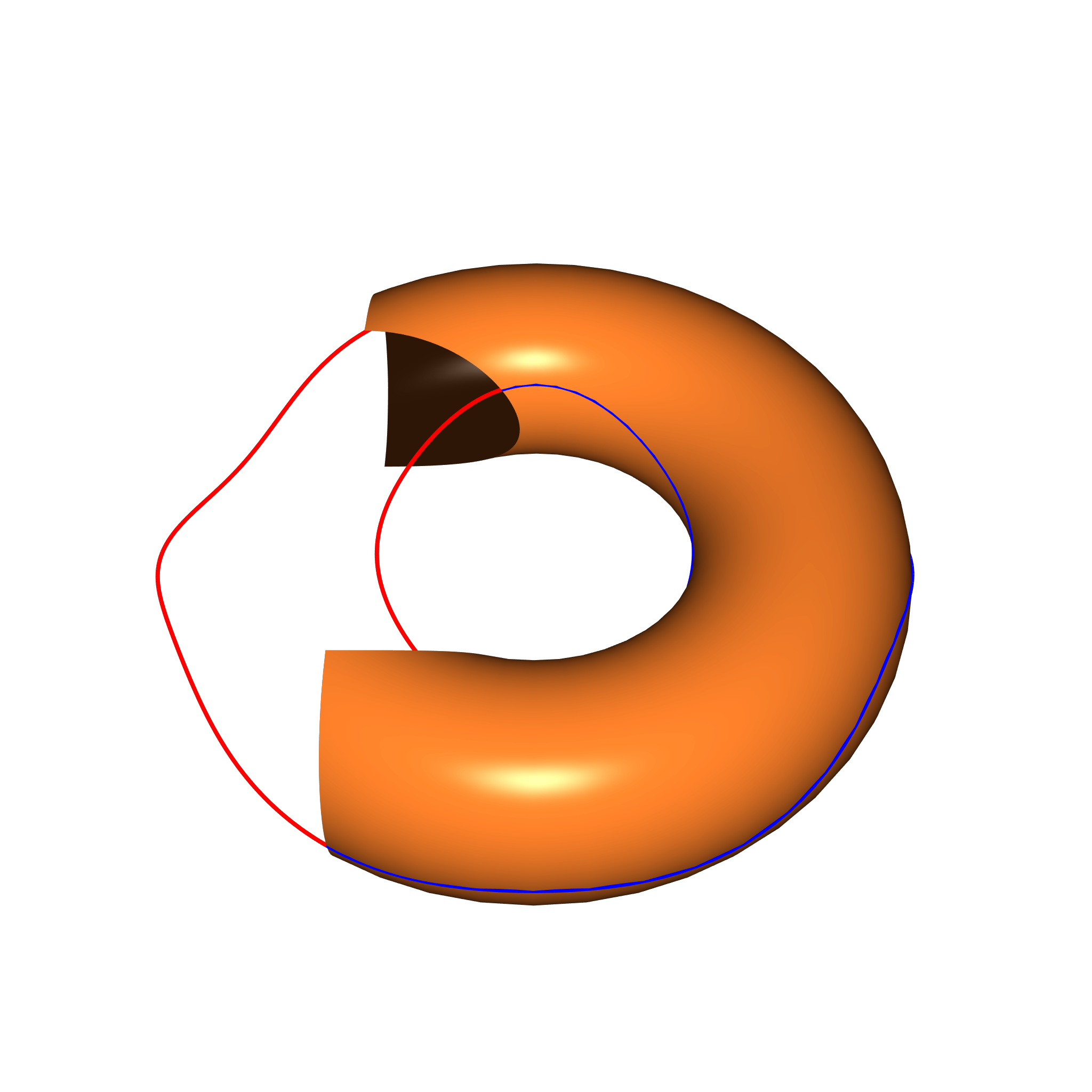}}
  \subfigure[Patch 1 after containment filtering]{
      \includegraphics[width=0.23\linewidth]{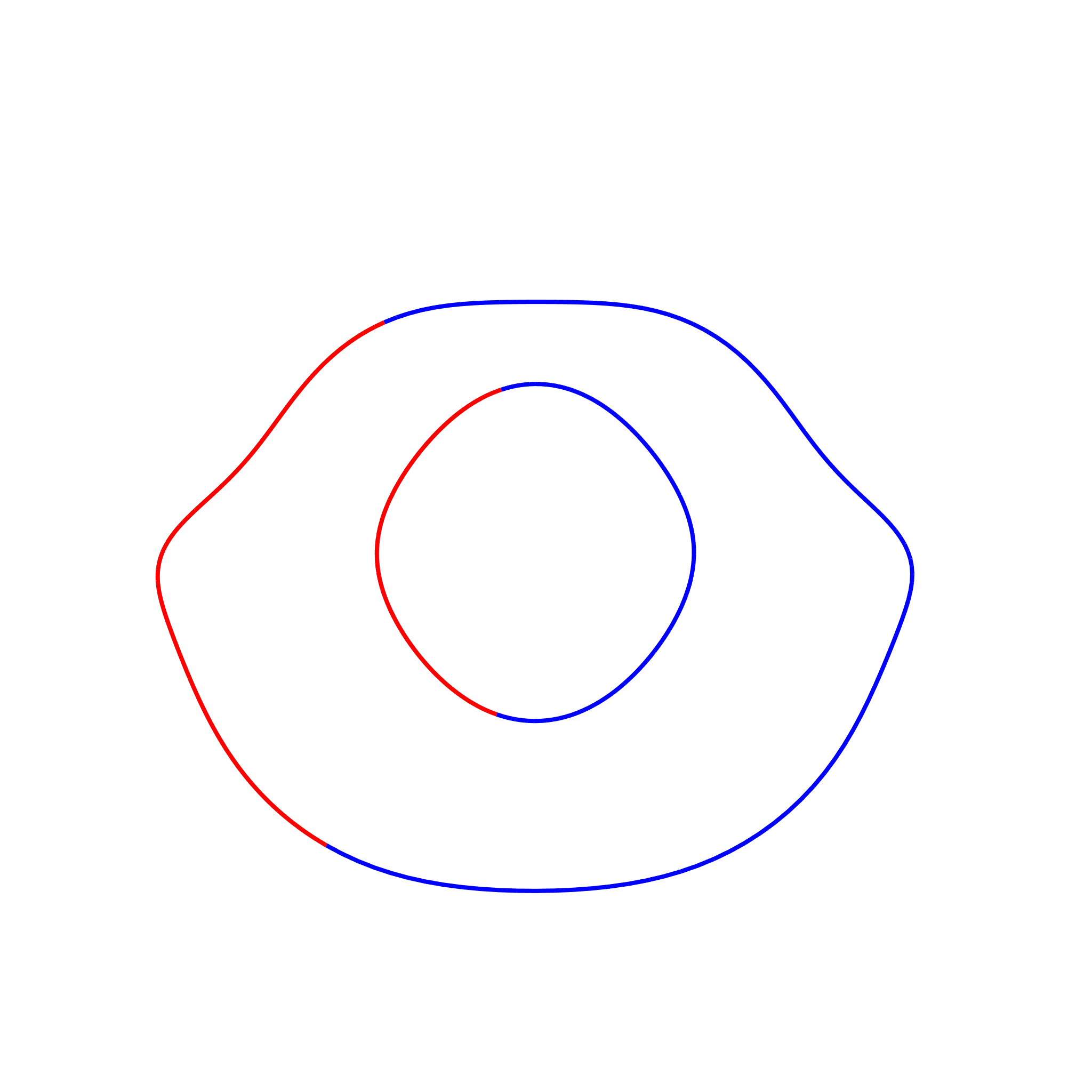}}
  \caption{Intersection between an eight-shaped genus-2 surface and a wave surface.  
  (a) Input models.  
  (b) Intersection curves.  
  (c,d) Raw intersection curves restricted to one surface patch; red segments are removed by containment queries.}
  \label{fig:intersection}
\end{figure*}

\begin{figure*}[h]
  \centering
  \subfigbottomskip=2pt
  \subfigcapskip=-5pt
  \subfigure[Periodic surface]{
      \includegraphics[width=0.23\linewidth]{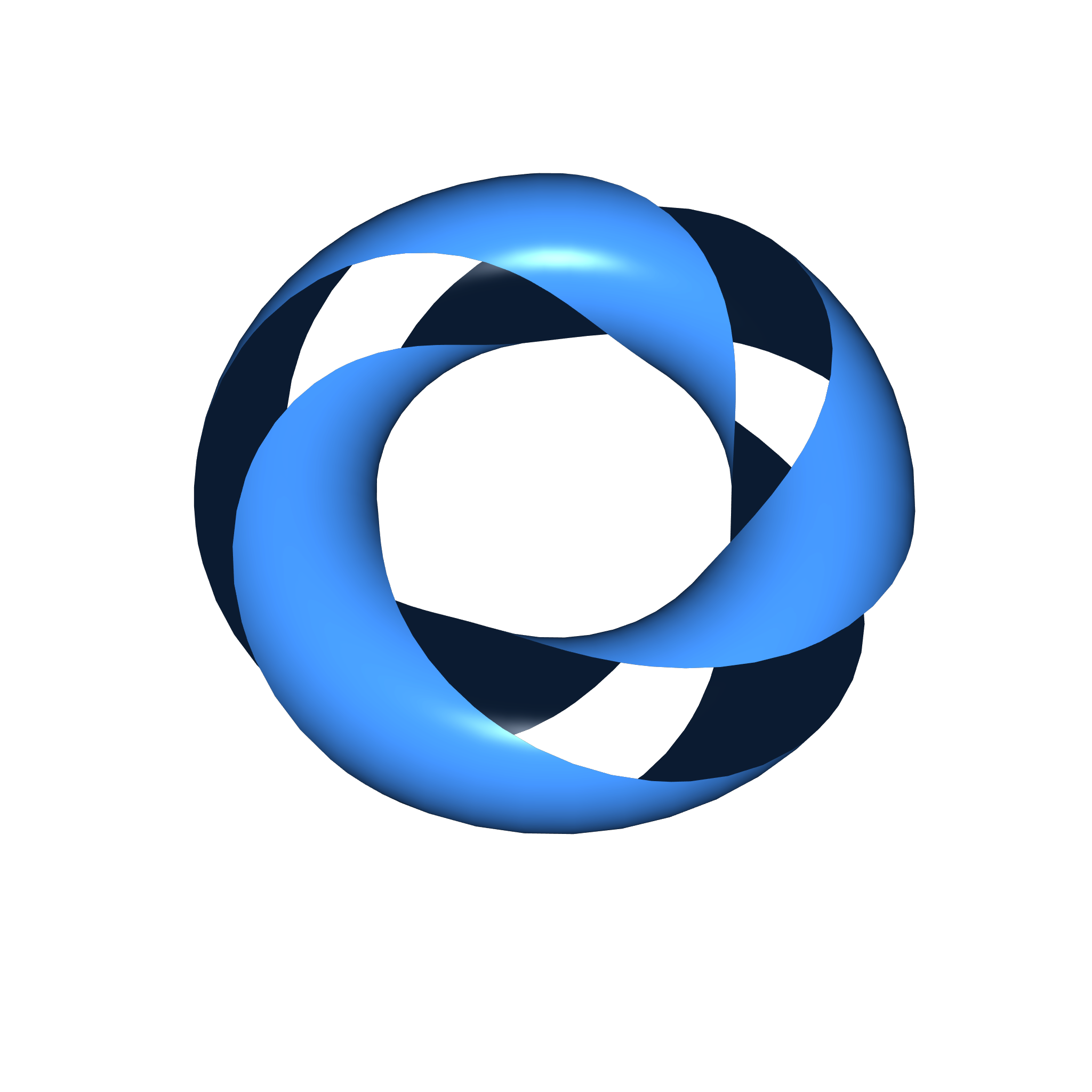}}
  \subfigure[Waving surface]{
      \includegraphics[width=0.23\linewidth]{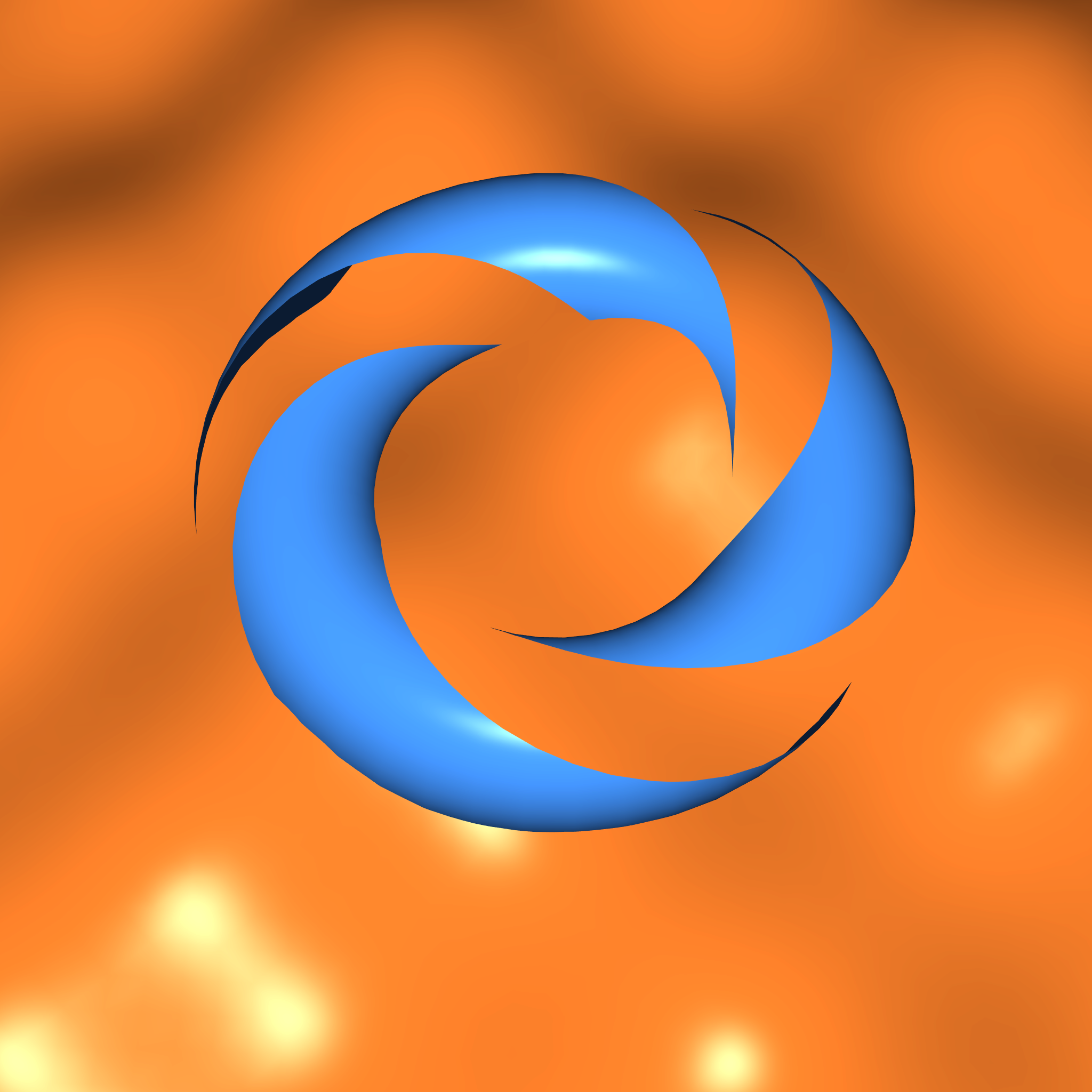}}
  \subfigure[Containment filtering]{
      \includegraphics[width=0.23\linewidth]{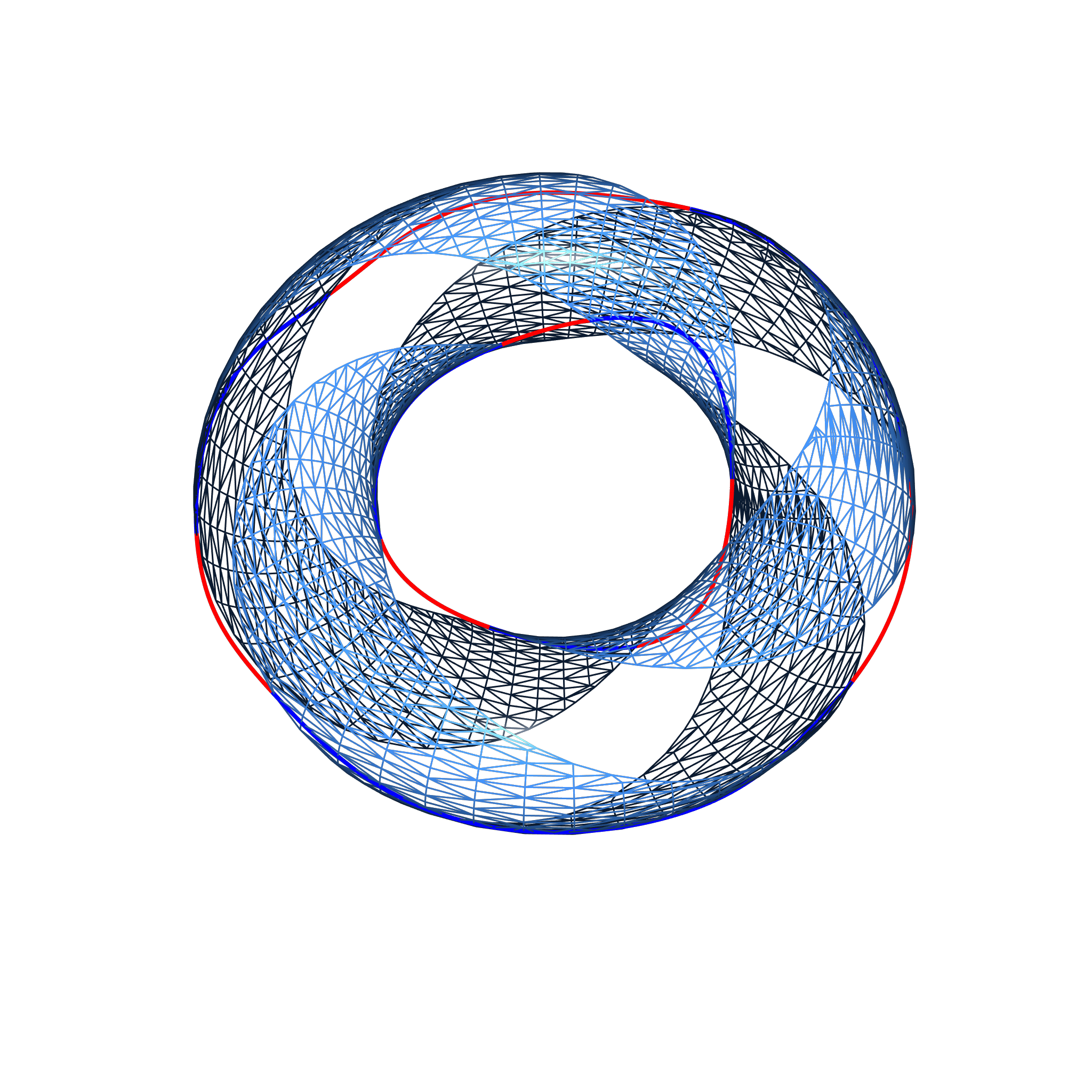}}
  \subfigure[Final intersection results]{
      \includegraphics[width=0.23\linewidth]{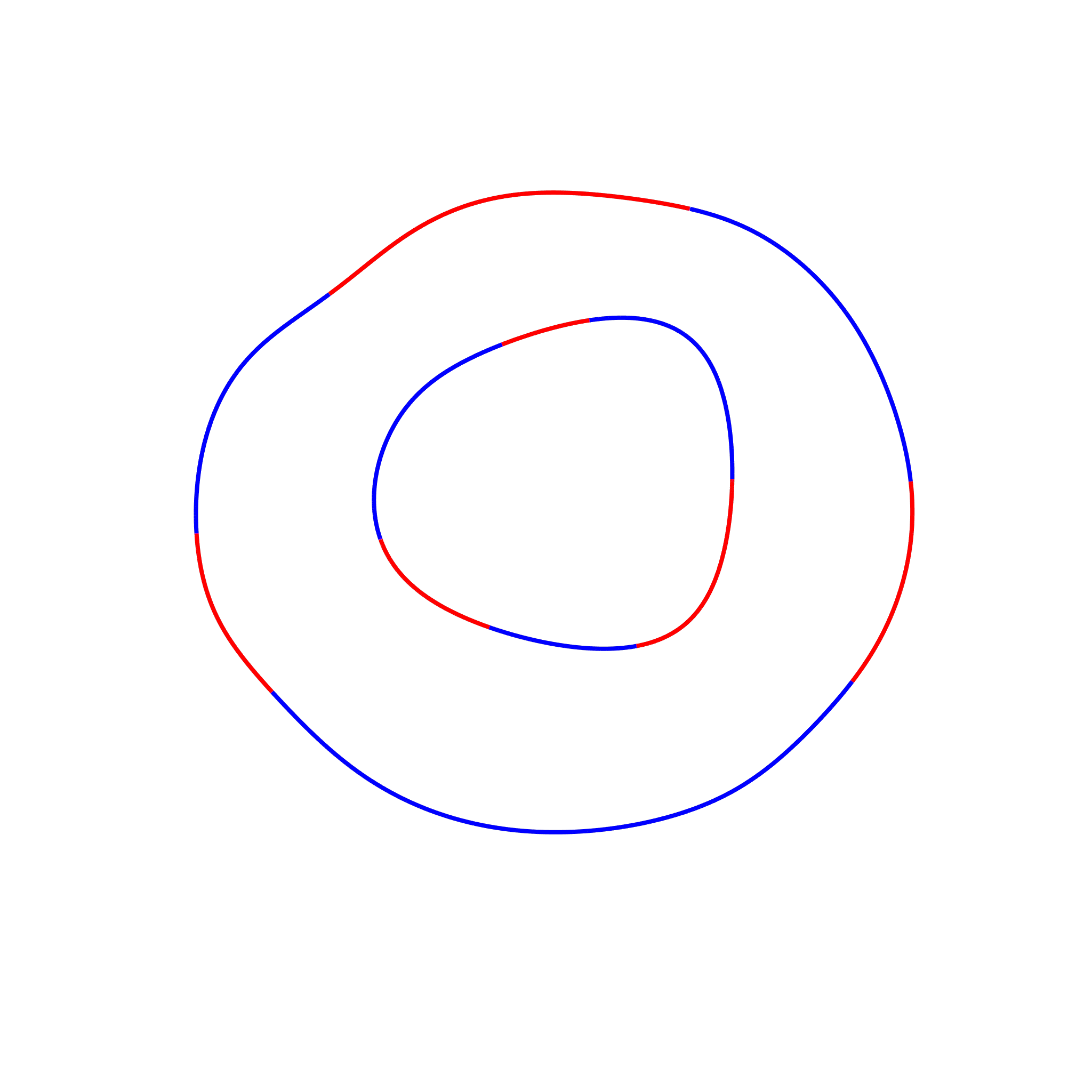}}
  \caption{Intersection between a bi-periodic surface with homology coefficients $(p, q) = (2, 3)$ and a waving surface.  
  (a) Input periodic surface.  
  (b) Relative position with respect to the waving surface.  
  (c,d) Containment queries on the periodic surface correctly remove intersections outside the trimmed region.}
  \label{fig:intersection2}
\end{figure*}

\subsubsection{Tessellation of Trimmed Surfaces}
\label{sec:application}

As a practical application, we use the proposed algorithm to perform tessellation of trimmed periodic surfaces. 
Efficient rendering of trimmed surfaces requires a tessellation stage. 
A common approach is to discretize $p$-curves, sample a uniform grid in the parametric domain, and then apply a triangulation strategy such as constrained Delaunay triangulation (CDT). 
However, triangles lying outside the trimmed region must subsequently be removed—a step that is often error-prone due to geometric inconsistencies such as open or non-manifold boundaries. 
Periodic domains present additional challenges, as naïve inside/outside checks fail across periodic boundaries.

With our algorithm, tessellation becomes both efficient and robust, even in the presence of disconnected or imperfect boundaries. 
Figure~\ref{fig:tessellation} shows an example of a surface trimmed by open and non-manifold boundary curves. 
Figure~\ref{fig:tessellation_demo} demonstrates robust tessellation on a torus-like periodic surface: despite disconnected $p$-curves, our method successfully produces a consistent triangulation of the trimmed region.

\section{Limitations and Future Work}
\label{sec:future}

\textit{Curve Representation}. In this paper, we demonstrated an efficient method to compute the winding number given a point and a parametric curve. The parametric curve is represented by the \bezier curve or rational \bezier curve. If the p-curves are given in other forms, a conversion is required, which needs some overhead. For general B-spline curves or even NURBS curves, the upper bound of derivatives can also be efficiently evaluated, but the linear evaluation method still needs to be developed in future work.

\textit{Differentiable Trimming}. Recently, inverse/differentiable rendering has evolved rapidly. The method for mesh, voxel, and point cloud has been established. Interestingly, \cite{worchel2023differentiable} developed a method to perform differentiable rendering for parametric surfaces via differentiable triangulation. An important topic is whether we can develop differentiable trimming techniques for parametric surfaces.

\section{Conclusion}
\label{sec:conclusion}

In this paper, we presented a novel method for containment queries on trimmed surfaces. By exploiting the bound of \bezier pieces, we designed an accelerating structure for 2D winding number computation. The simplicity of the proposed ellipse bound enables us to use more efficient methods to evaluate the boundary curves without loss of numerical robustness. For the ubiquitous periodic geometric surfaces, we leverage universal covering space to perform inside/outside classification, and simplification methods are used to avoid the infinite computation of winding numbers. Experiment shows our method is efficient, outperforming the existing ray casting method and the winding number method proposed in recent work \cite{spainhour_robust_2024}. From the aspect of robustness, our method can tolerate geometric errors, spur edges, and a variety of corner cases in the periodic domains.

\bibliographystyle{ACM-Reference-Format}
\bibliography{sample-base}

\FloatBarrier
\begin{figure}[h]
    \centering
    \subfigbottomskip=2pt
    \subfigcapskip=-5pt
    \subfigure[Original model]{
        \includegraphics[width=0.4\linewidth]{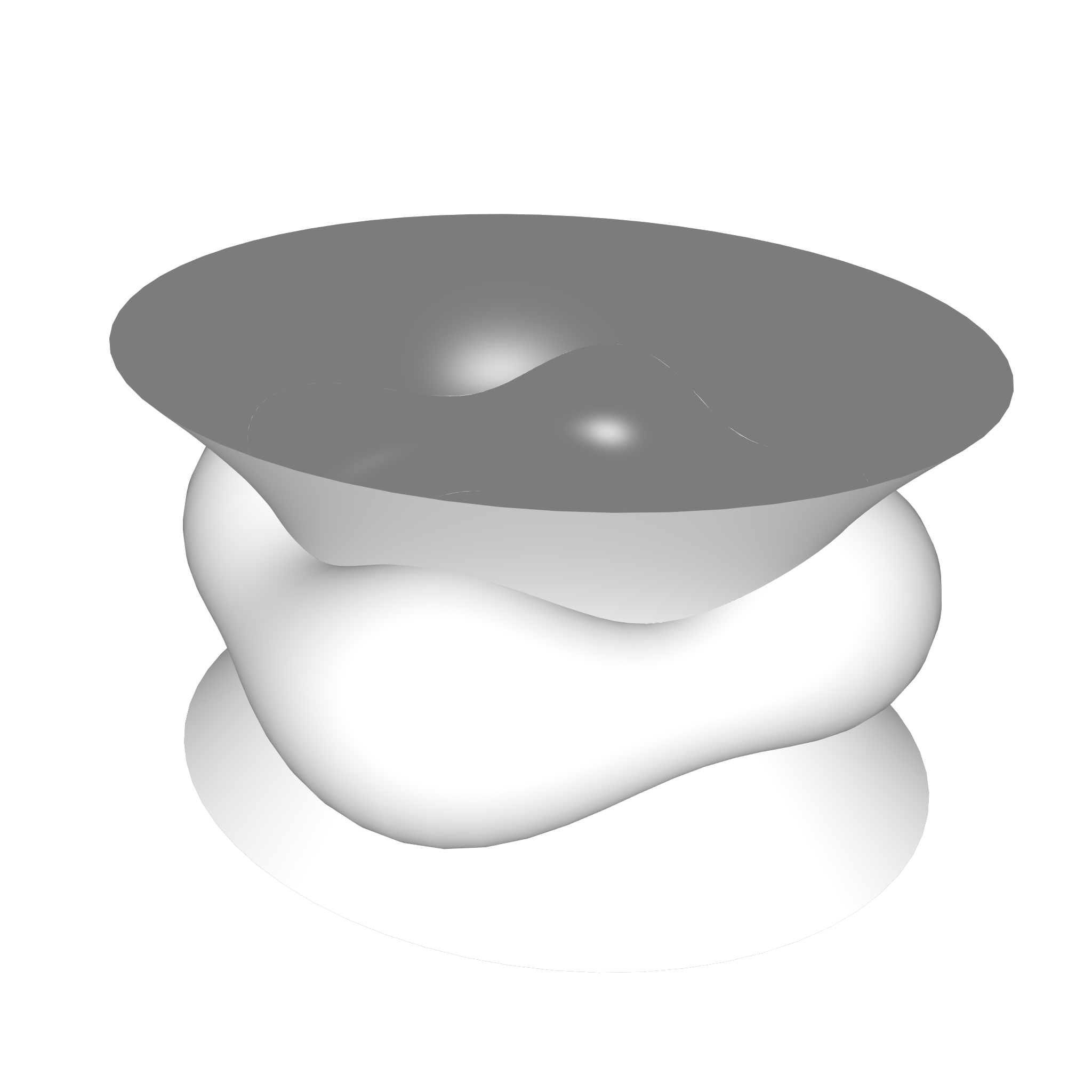}}
    \subfigure[Containtment Result]{
        \includegraphics[width=0.4\linewidth]{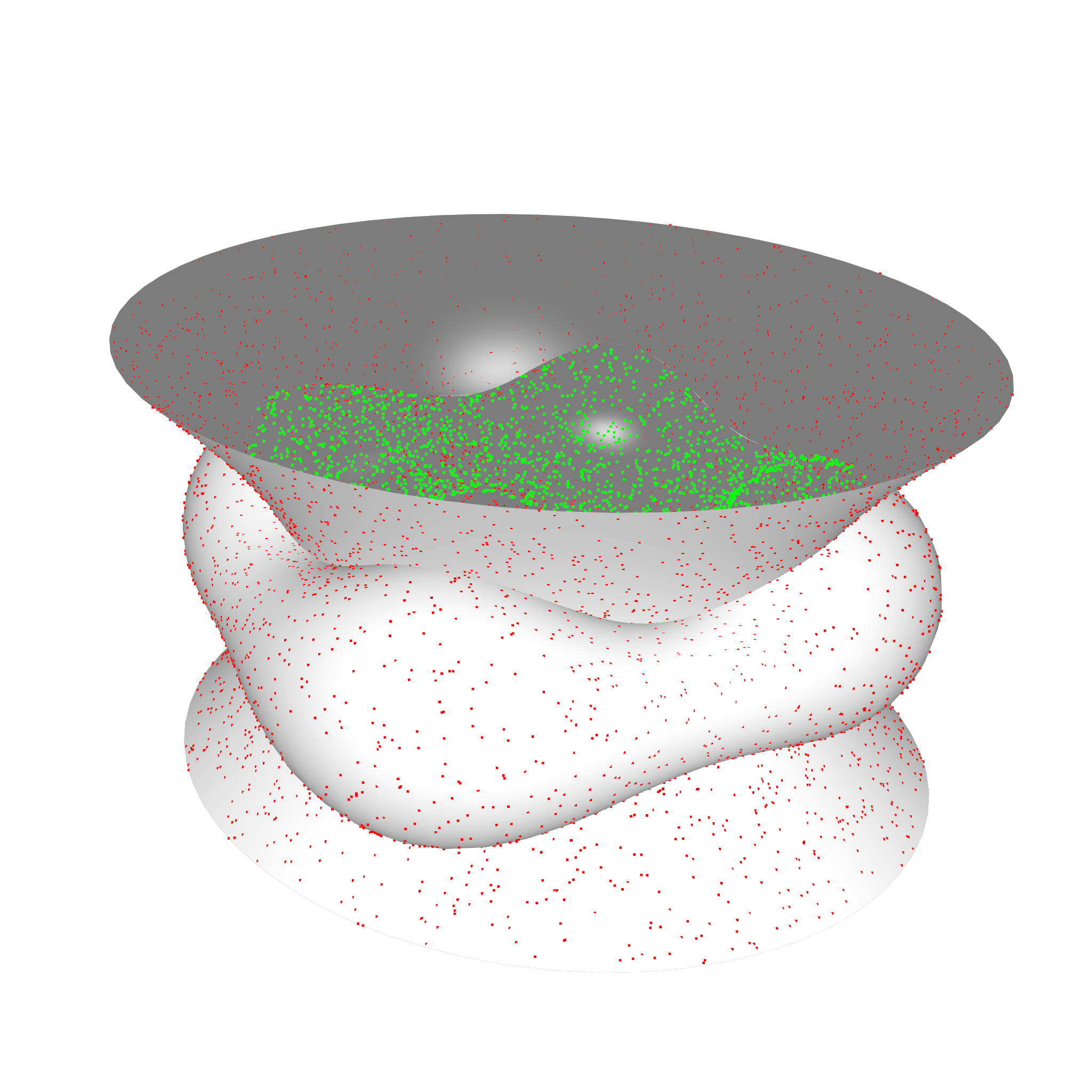}} \\
    \subfigure[Result on surface 1]{
        \includegraphics[width=0.4\linewidth]{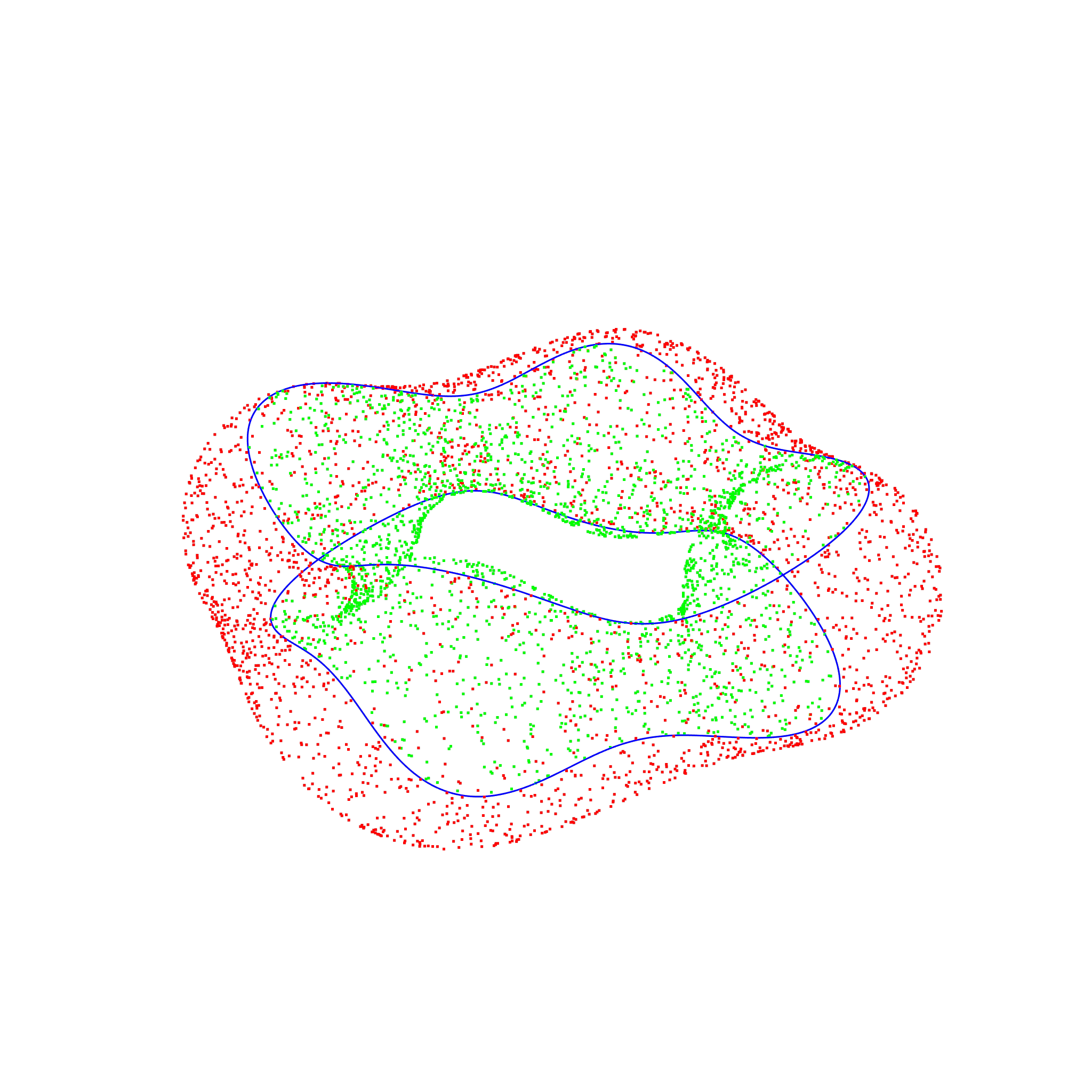}}
    \subfigure[Result on surface 2]{
        \includegraphics[width=0.4\linewidth]{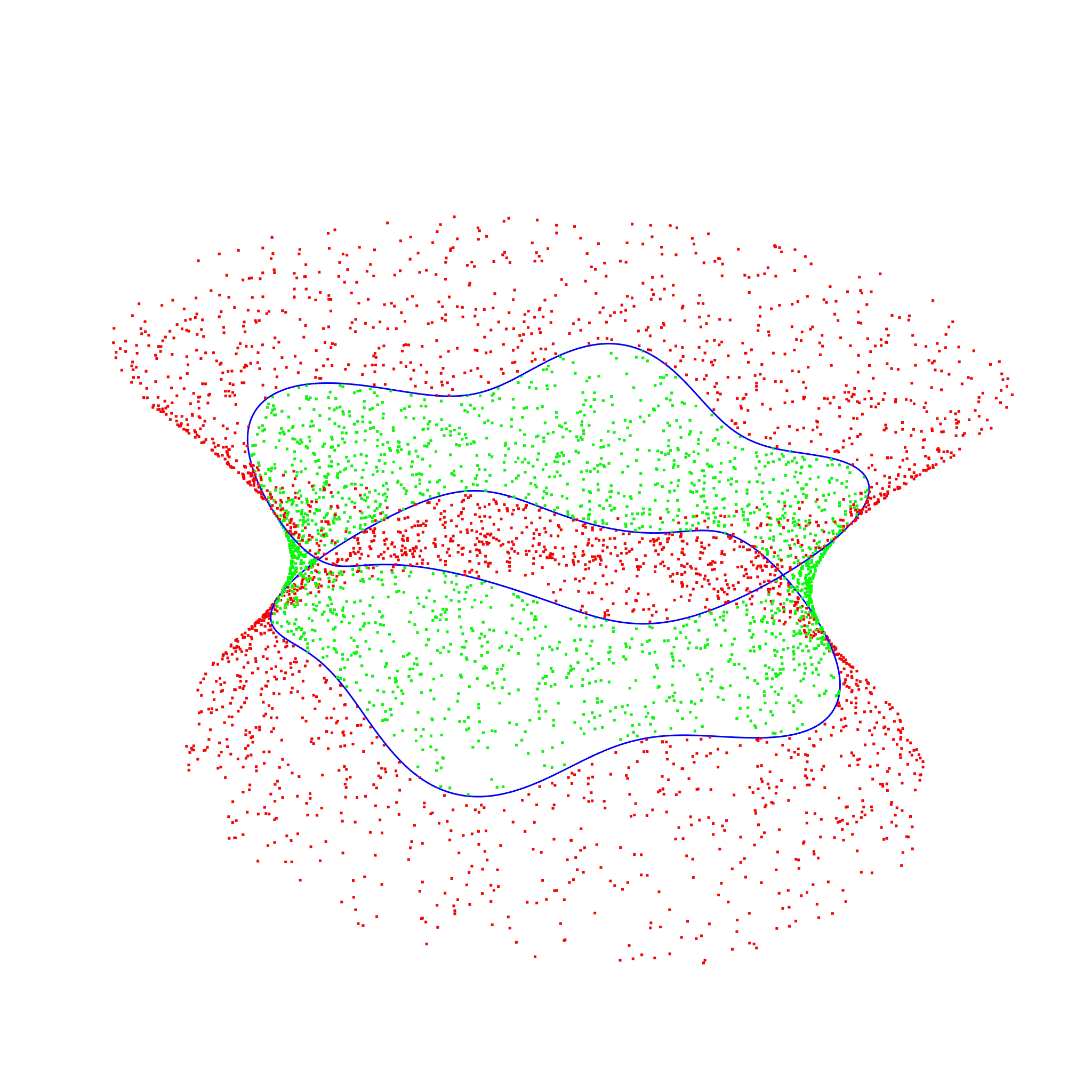}}
    \caption{Containtment queries on a B-Rep model with multiple patches. The surface is generated by union between a hyperbolid surface and defromed torus. Blue curves are intersection curves between two surface. Samples were generated on the geometric surfaces.}
    \label{fig:containment_demo}
\end{figure}

\begin{figure}[h]
    \centering
    \subfigbottomskip=2pt
    \subfigcapskip=-5pt
    \subfigure[Original model]{
        \includegraphics[width=0.4\linewidth]{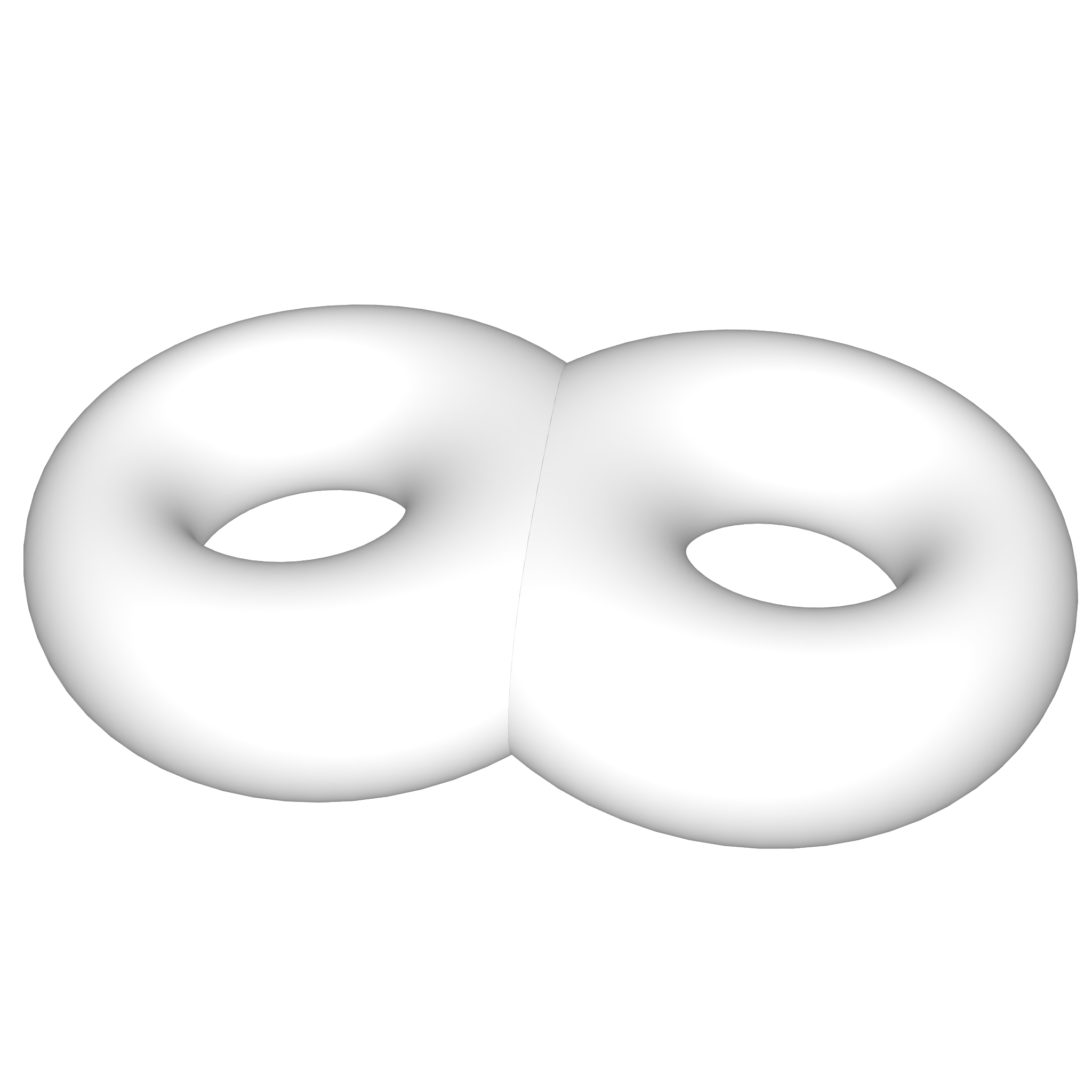}}
    \subfigure[Containtment Result]{
        \includegraphics[width=0.4\linewidth]{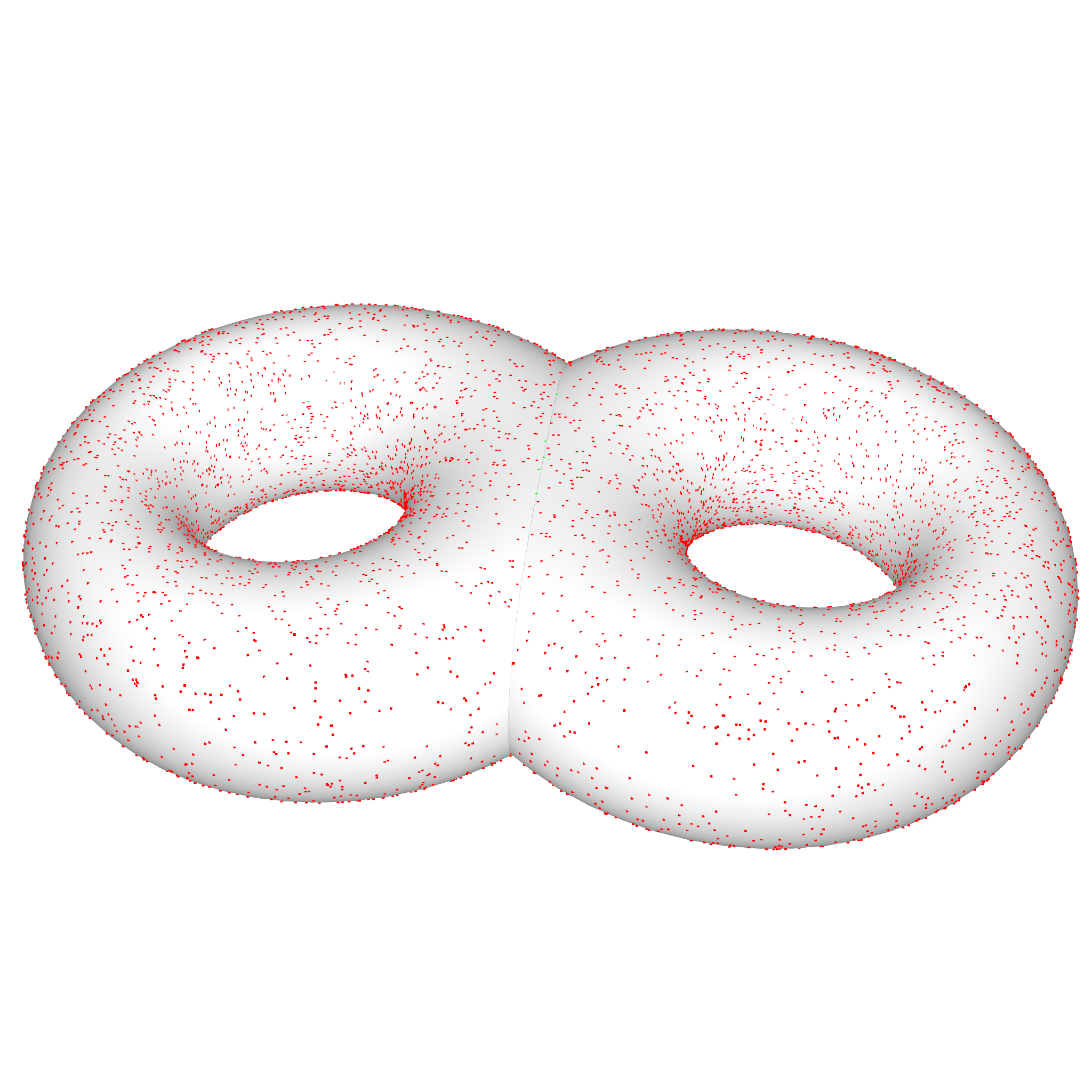}} \\
    \subfigure[Result on surface 1]{
        \includegraphics[width=0.4\linewidth]{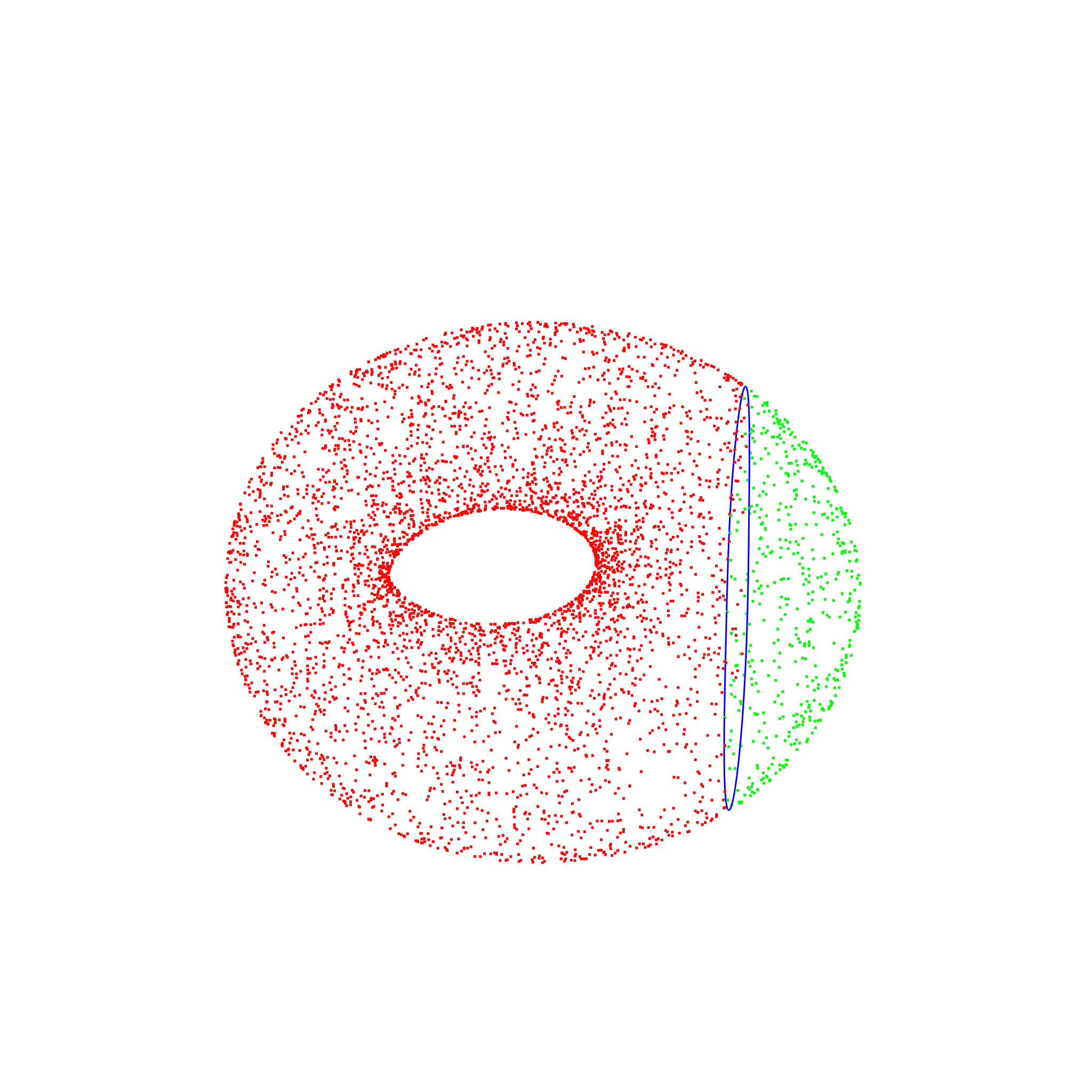}}
    \subfigure[Result on surface 2]{
        \includegraphics[width=0.4\linewidth]{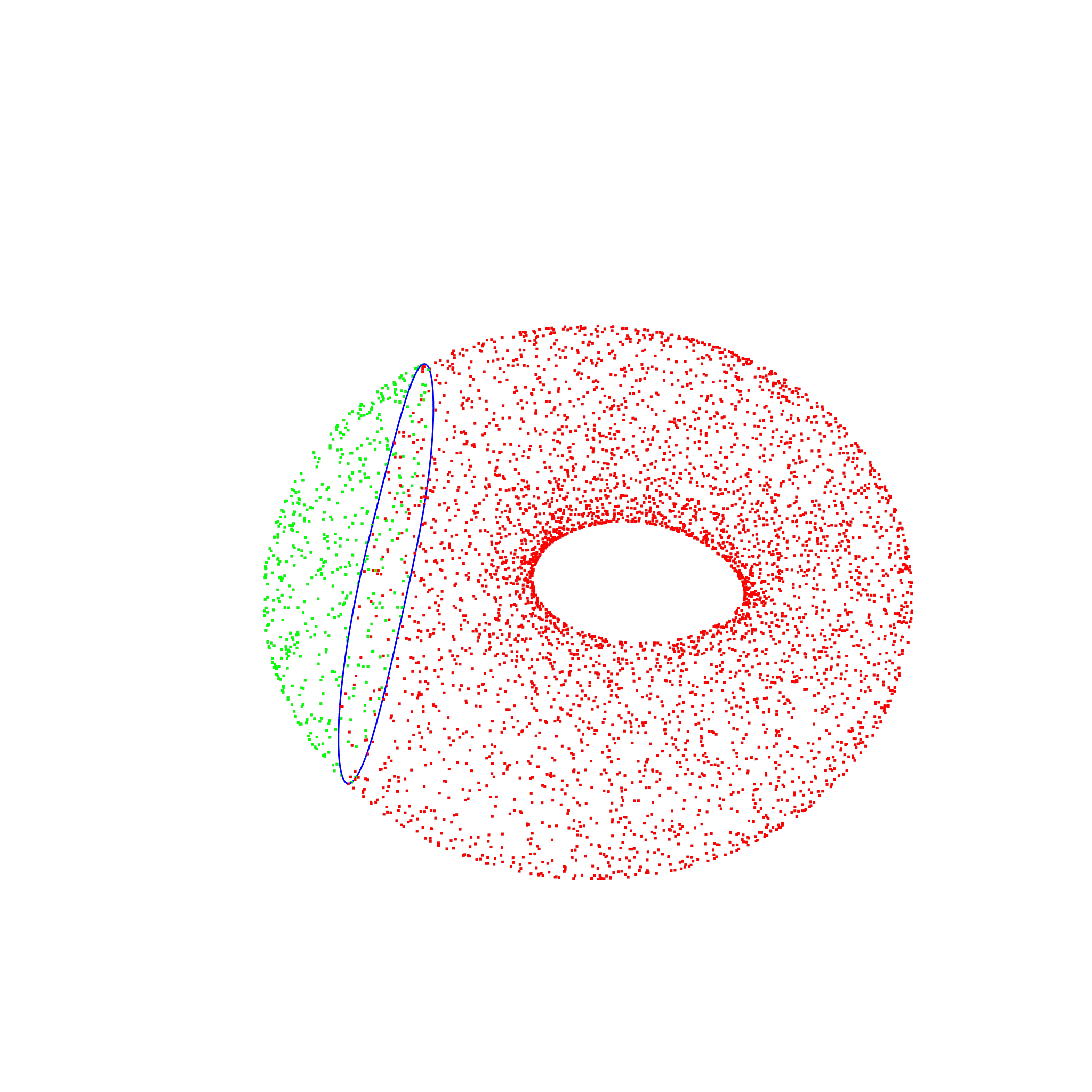}}
    \caption{Containtment queries on a genus 2 B-Rep model defined by the union of two torus like surfaces. Blue curves are intersection curves between two surface. Samples were generated on the geometric surfaces.}
    \label{fig:containment_demo2}
\end{figure}

\begin{figure}[h]
    \centering
    \subfigbottomskip=2pt
    \subfigcapskip=-5pt
    \subfigure[Model 1]{
        \includegraphics[width=0.48\linewidth]{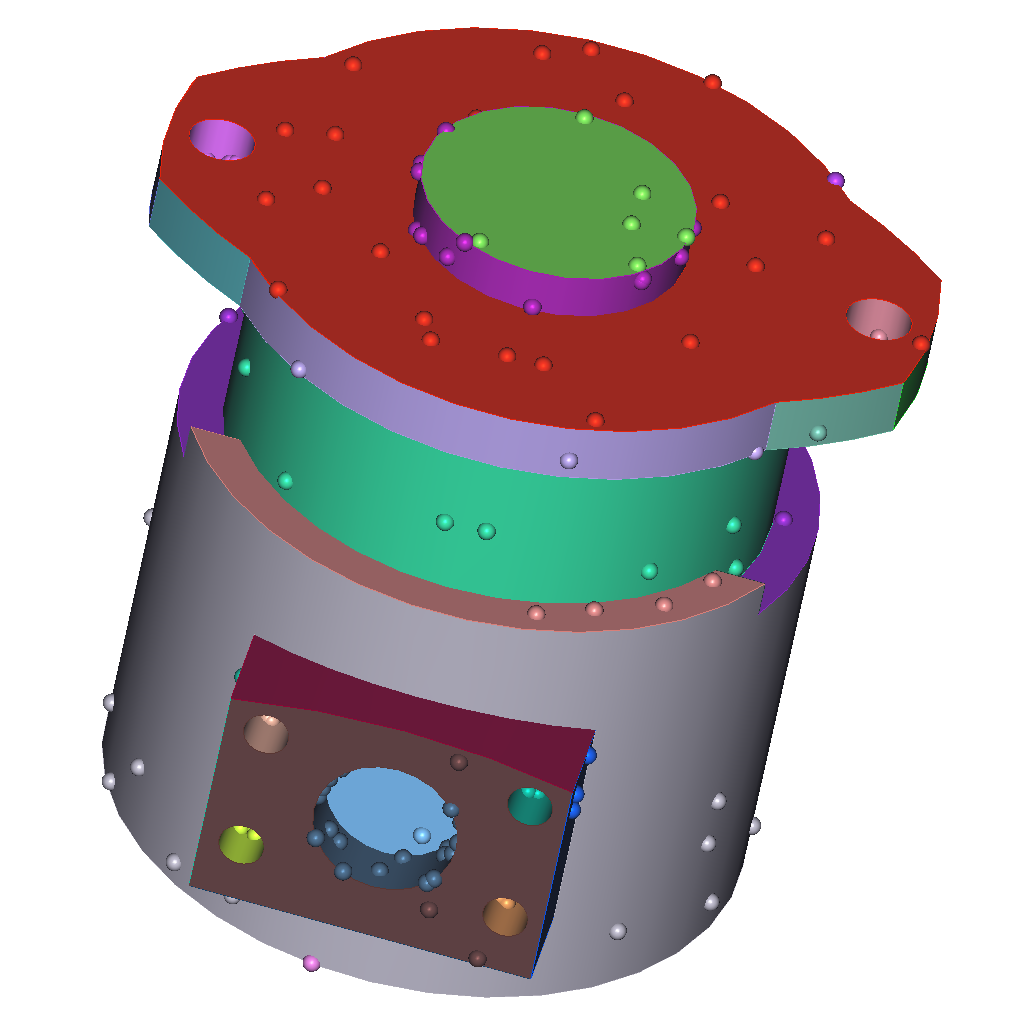}}
    \subfigure[Model 2]{
        \includegraphics[width=0.48\linewidth]{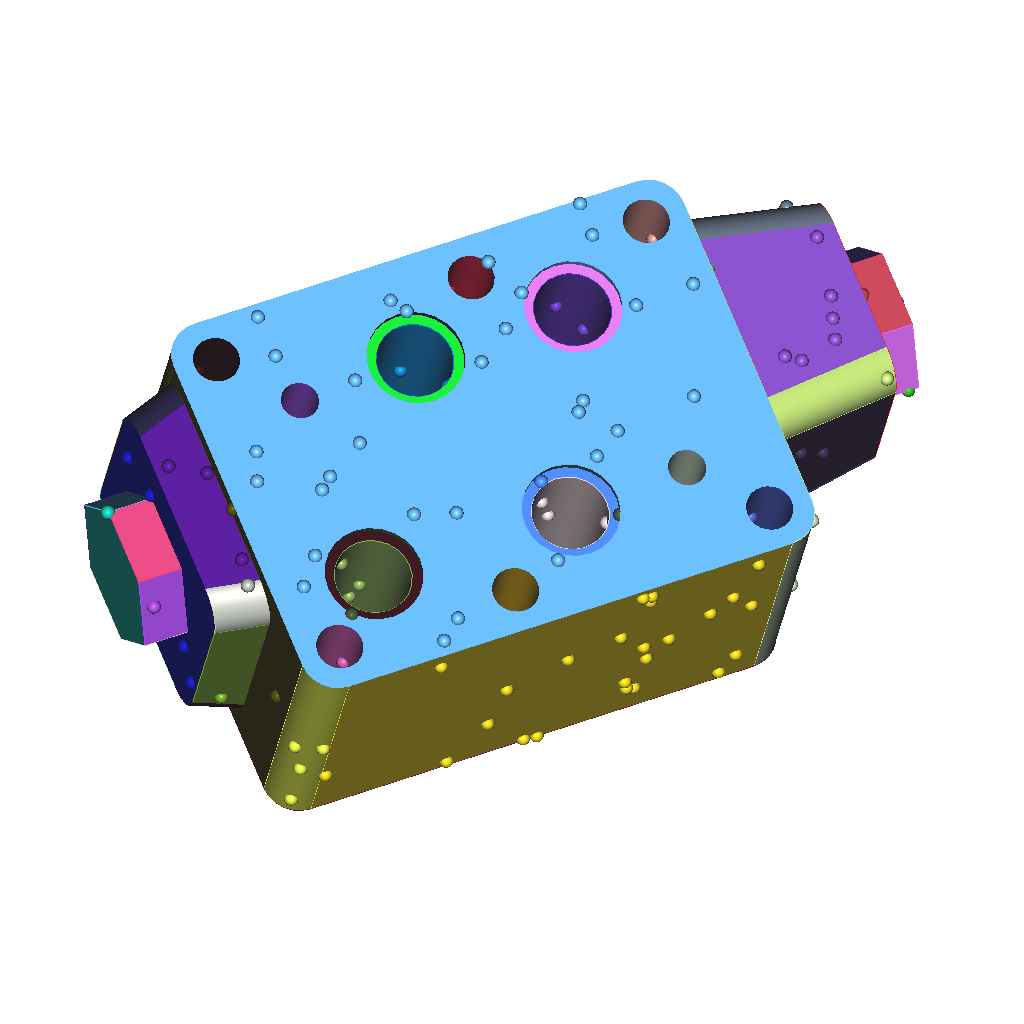}}
    \caption{Real world models. We randomly generate some points on the underlying geometric surfaces and use proposed algorithm to select the points on the boundary of B-Rep model.}
    \label{fig:complex_models}
\end{figure}

\begin{figure}[h]
    \centering
    \subfigbottomskip=2pt
    \subfigcapskip=-5pt
    \subfigure[Parameter Domain]{
        \includegraphics[width=0.3\linewidth]{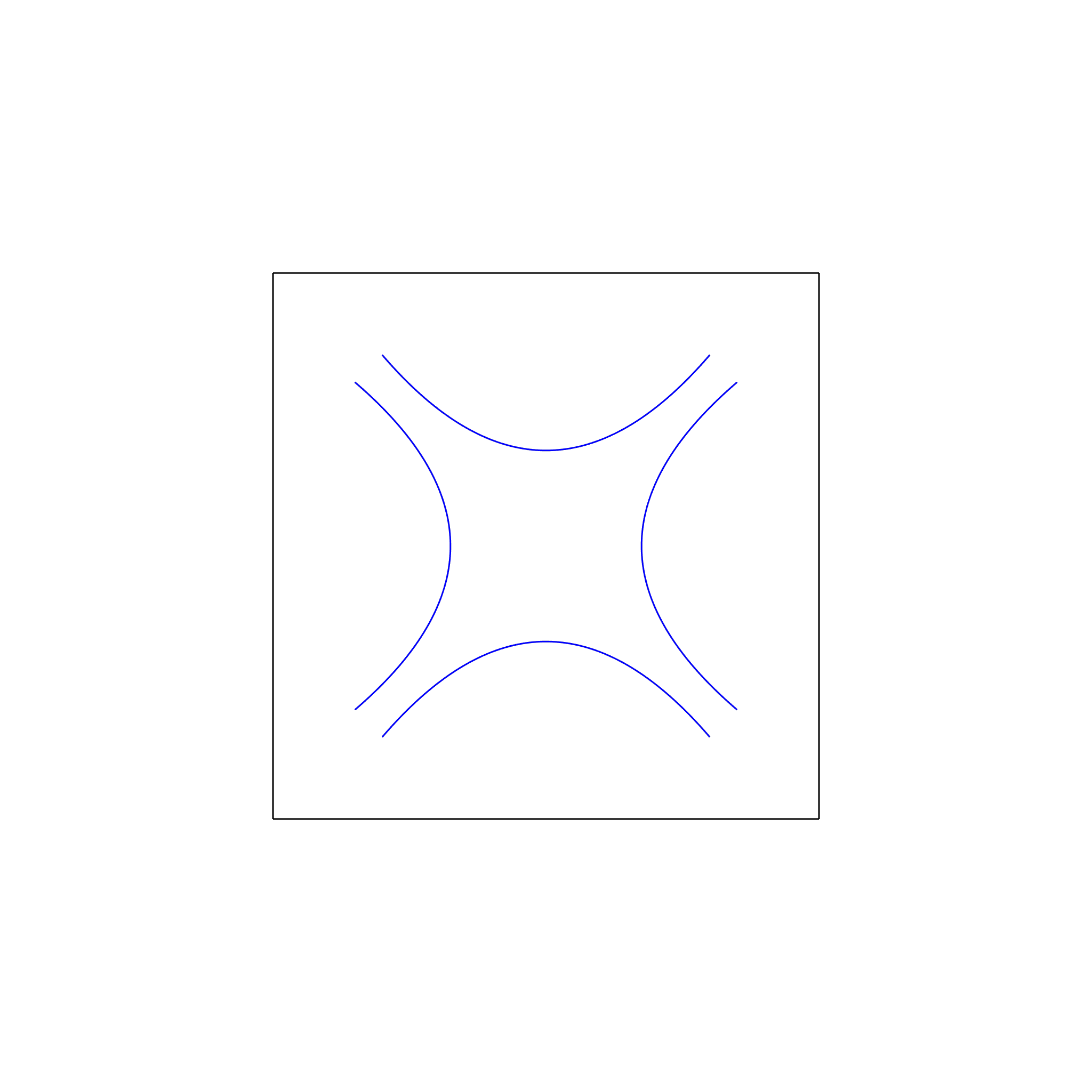}}
    \subfigure[Wireframe]{
        \includegraphics[width=0.3\linewidth]{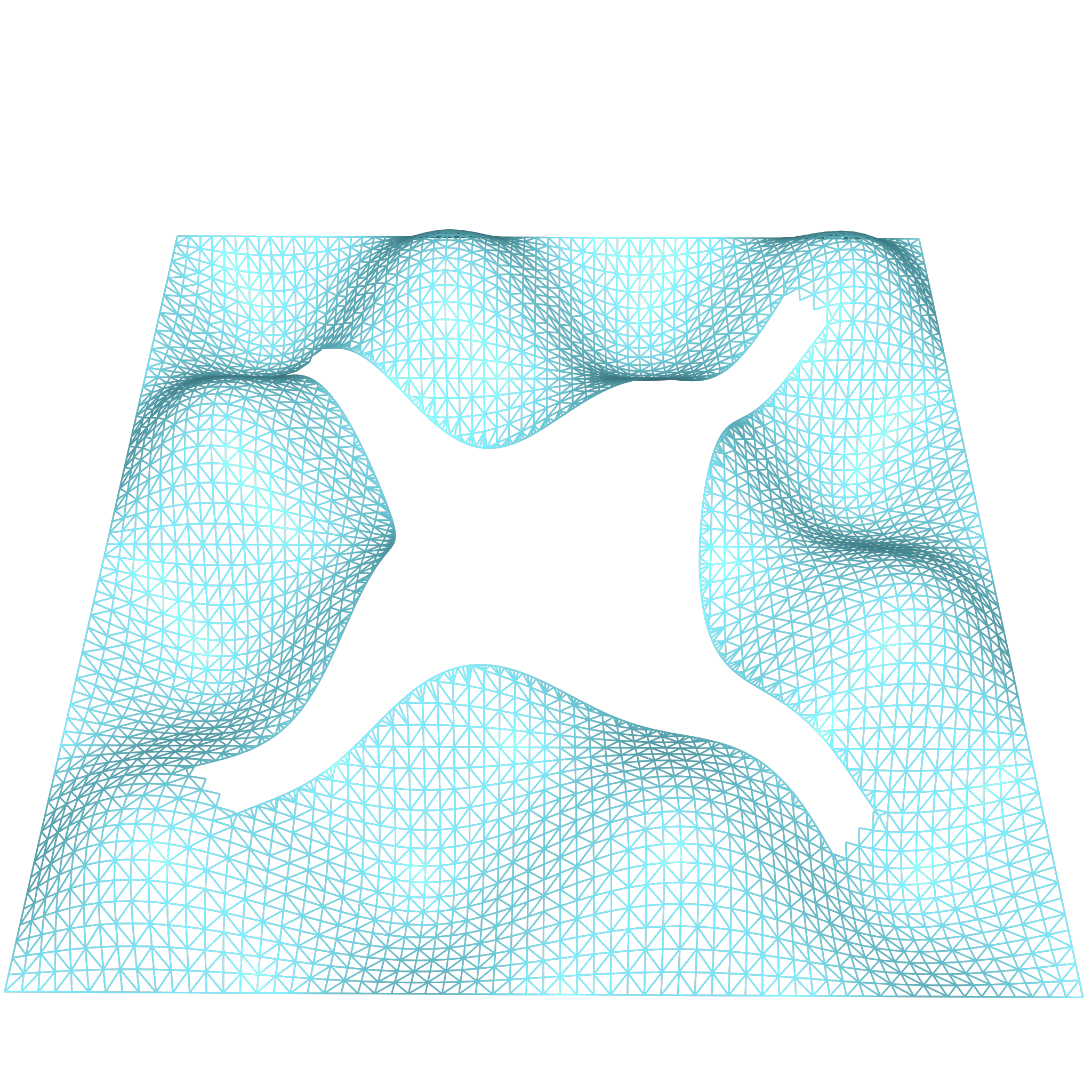}}
    \subfigure[Physical Space]{
        \includegraphics[width=0.3\linewidth]{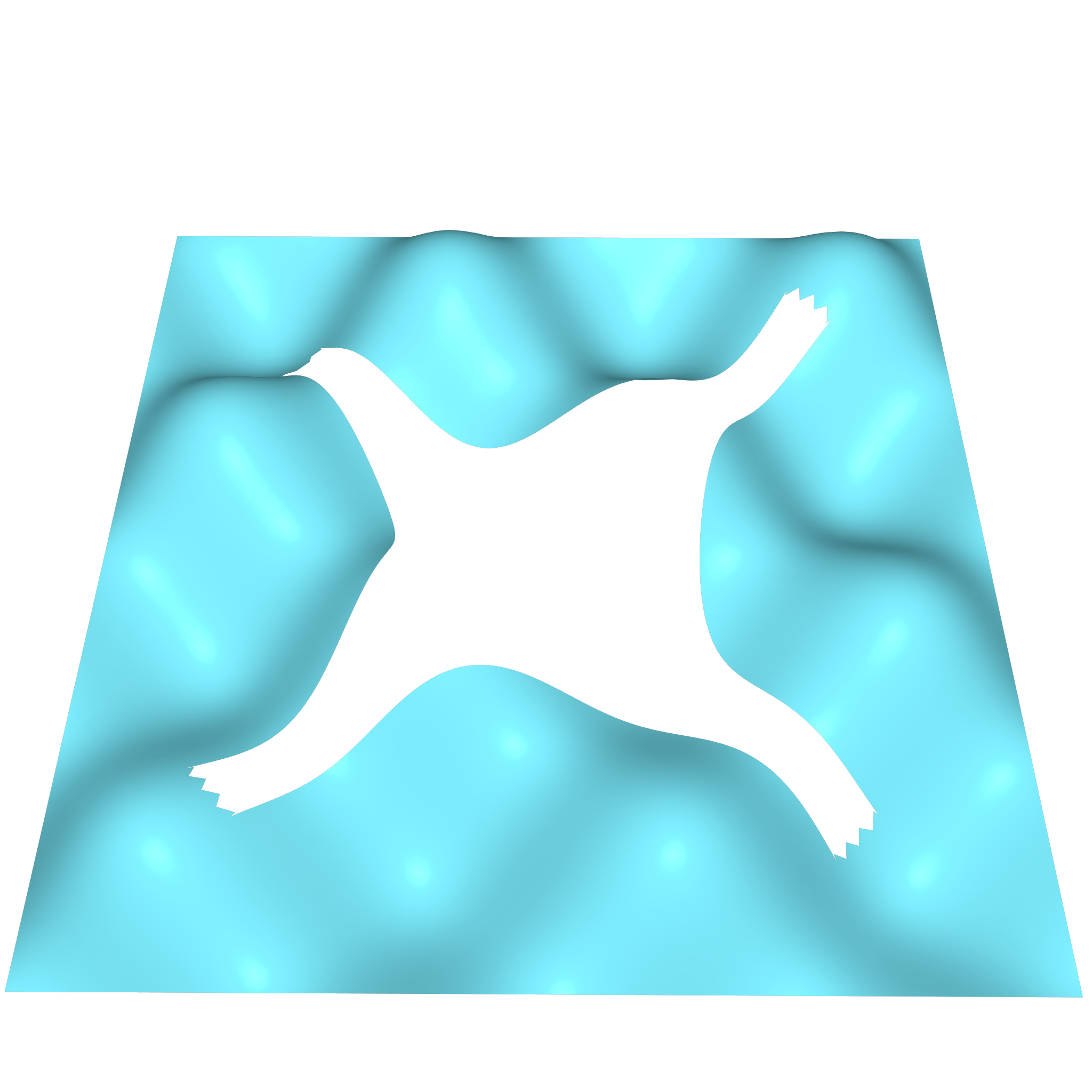}} \\
    \subfigure[Parameter Domain]{
        \includegraphics[width=0.3\linewidth]{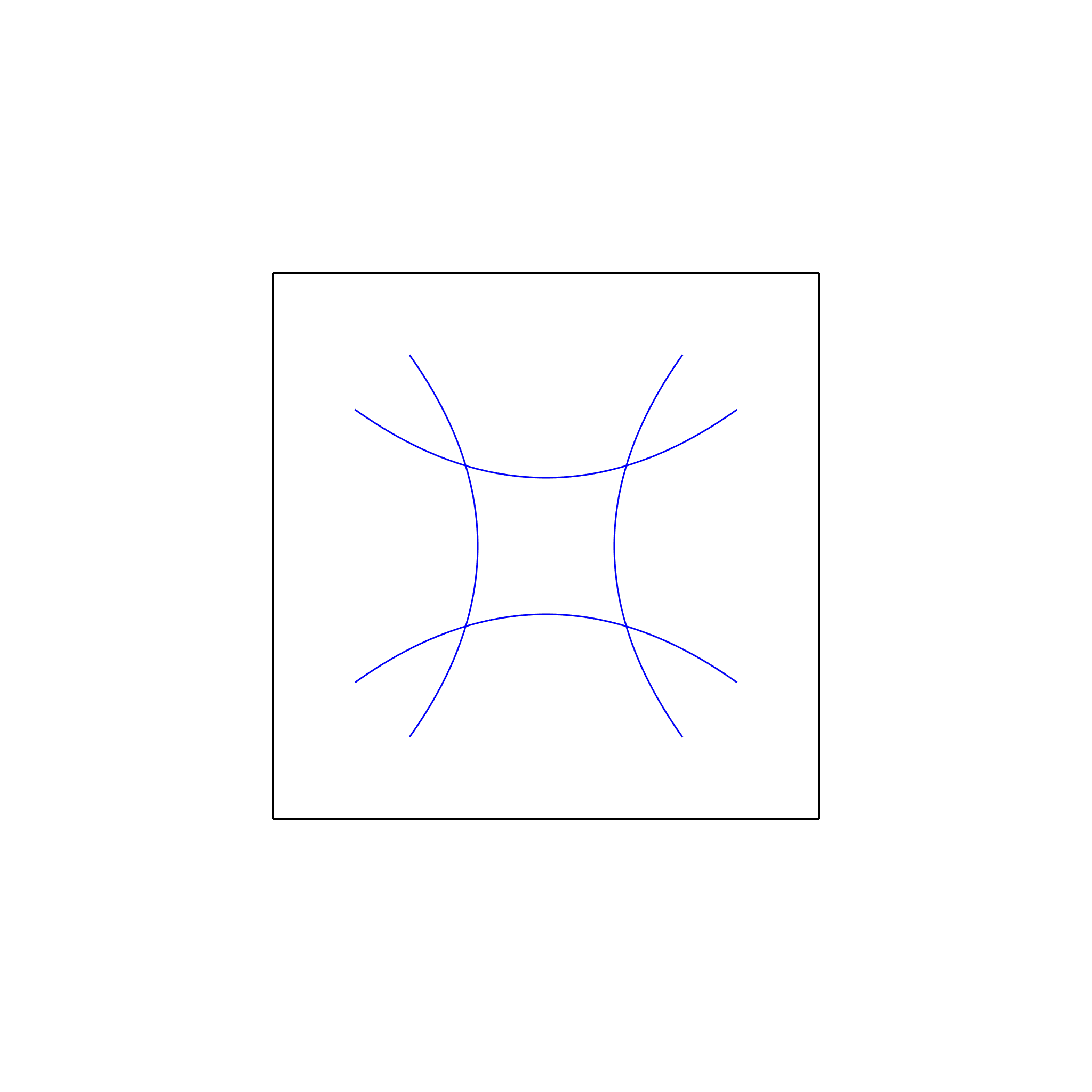}}
    \subfigure[Wireframe]{
        \includegraphics[width=0.3\linewidth]{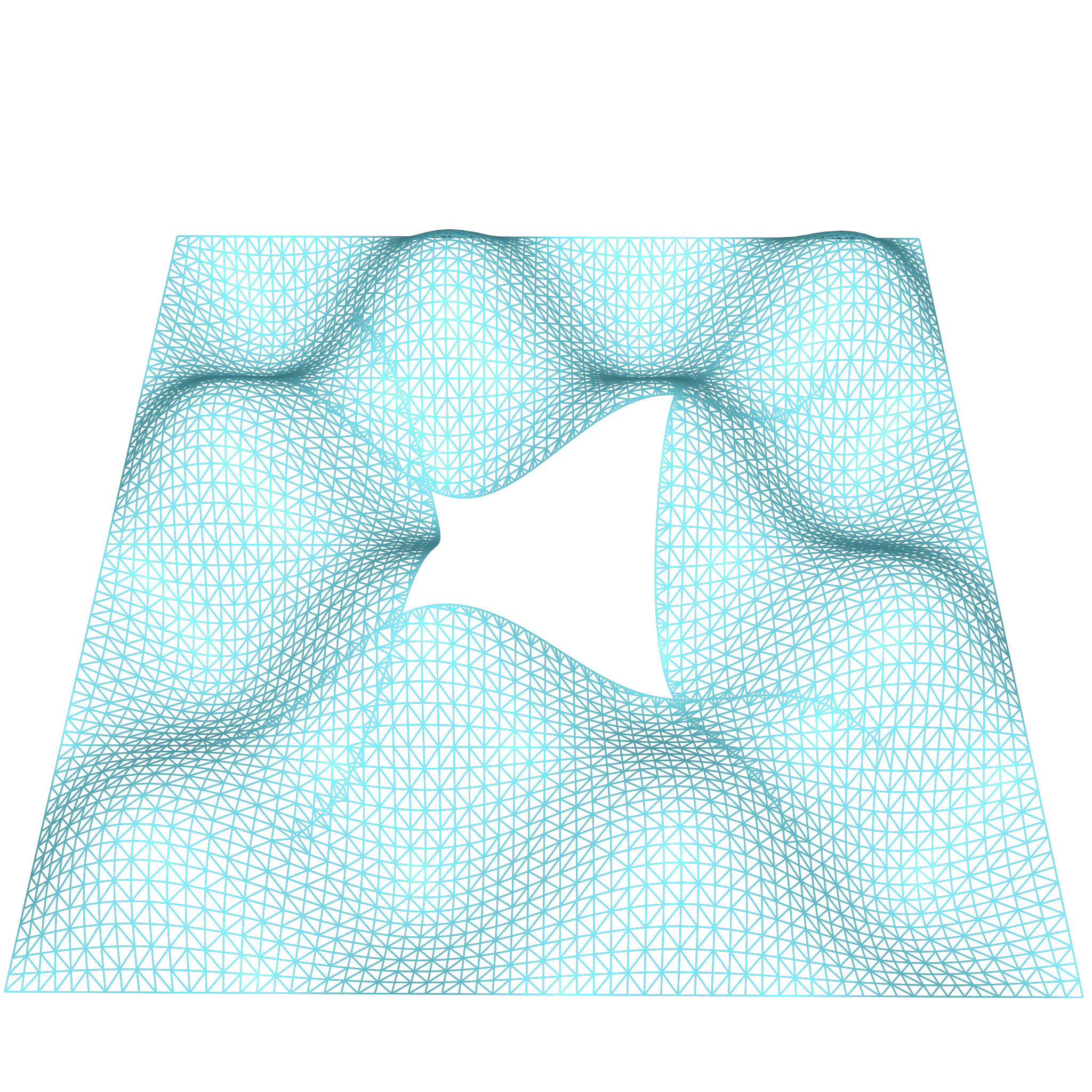}}
    \subfigure[Physical Space]{
        \includegraphics[width=0.3\linewidth]{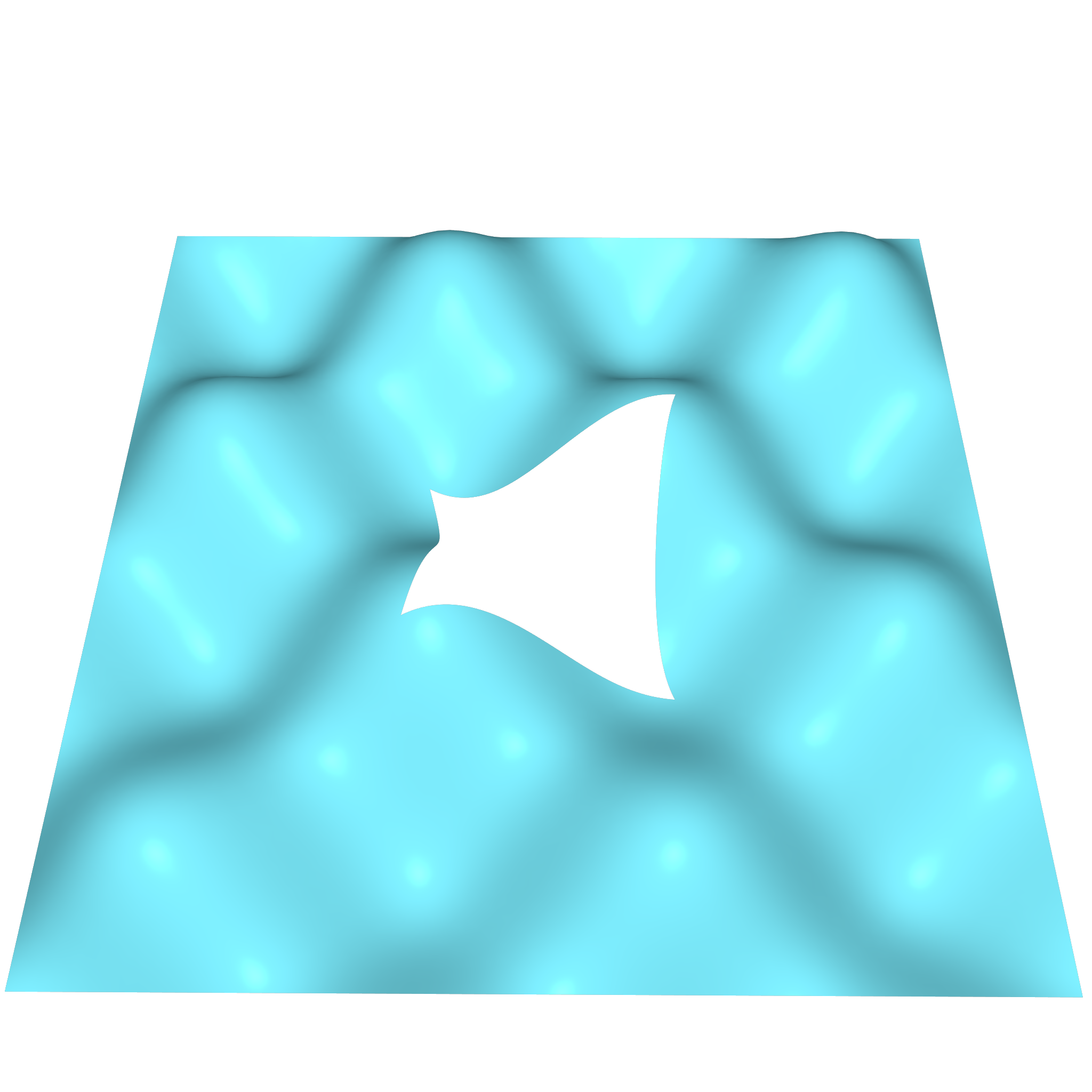}}
    \caption{Illustration of robust periodic trimmed surface tessellation. First row: tessellation of surface with open boundary curves. Second row: tessellation of surface with non-manifold boundary curves.}
    \label{fig:tessellation}
\end{figure}

\begin{figure}[h]
    \centering
    \subfigbottomskip=2pt
    \subfigcapskip=-5pt
    \subfigure[Covering Space]{
        \includegraphics[width=0.48\linewidth]{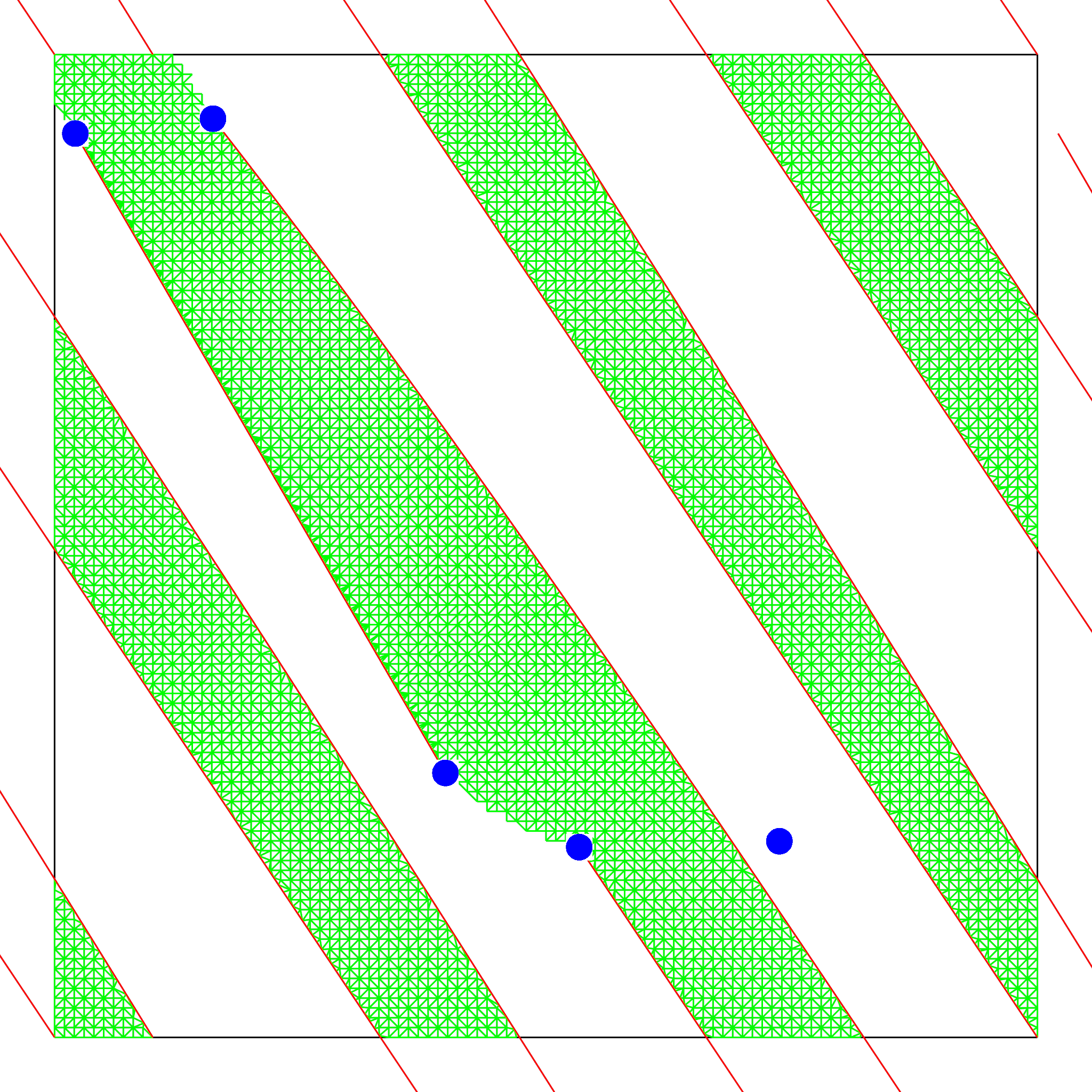}}
    \subfigure[Physical Space]{
        \includegraphics[width=0.48\linewidth]{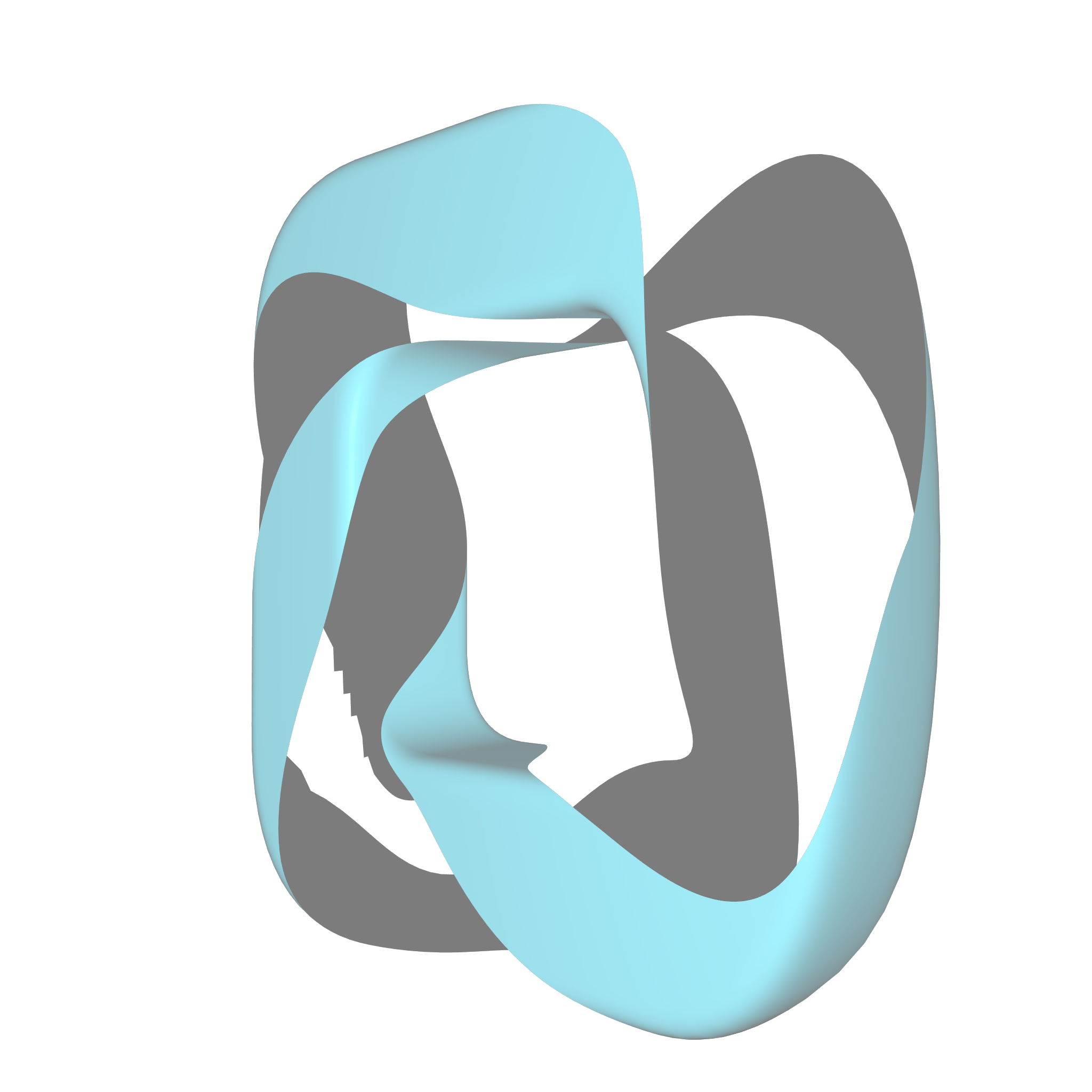}}
    \caption{Illustration of robust periodic trimmed surface tessellation. The black rectangle is boundary of parameter domain. Blue points are control points of the boundary curves. Even the boundary curves are not end-to-end continuous, our method can generate reasonable tessellation result.}
    \label{fig:tessellation_demo}
\end{figure}

\begin{figure*}[h]
    \centering
    \subfigbottomskip=2pt
    \subfigcapskip=-5pt
    \subfigure[4]{
    \includegraphics[width=0.24\linewidth]{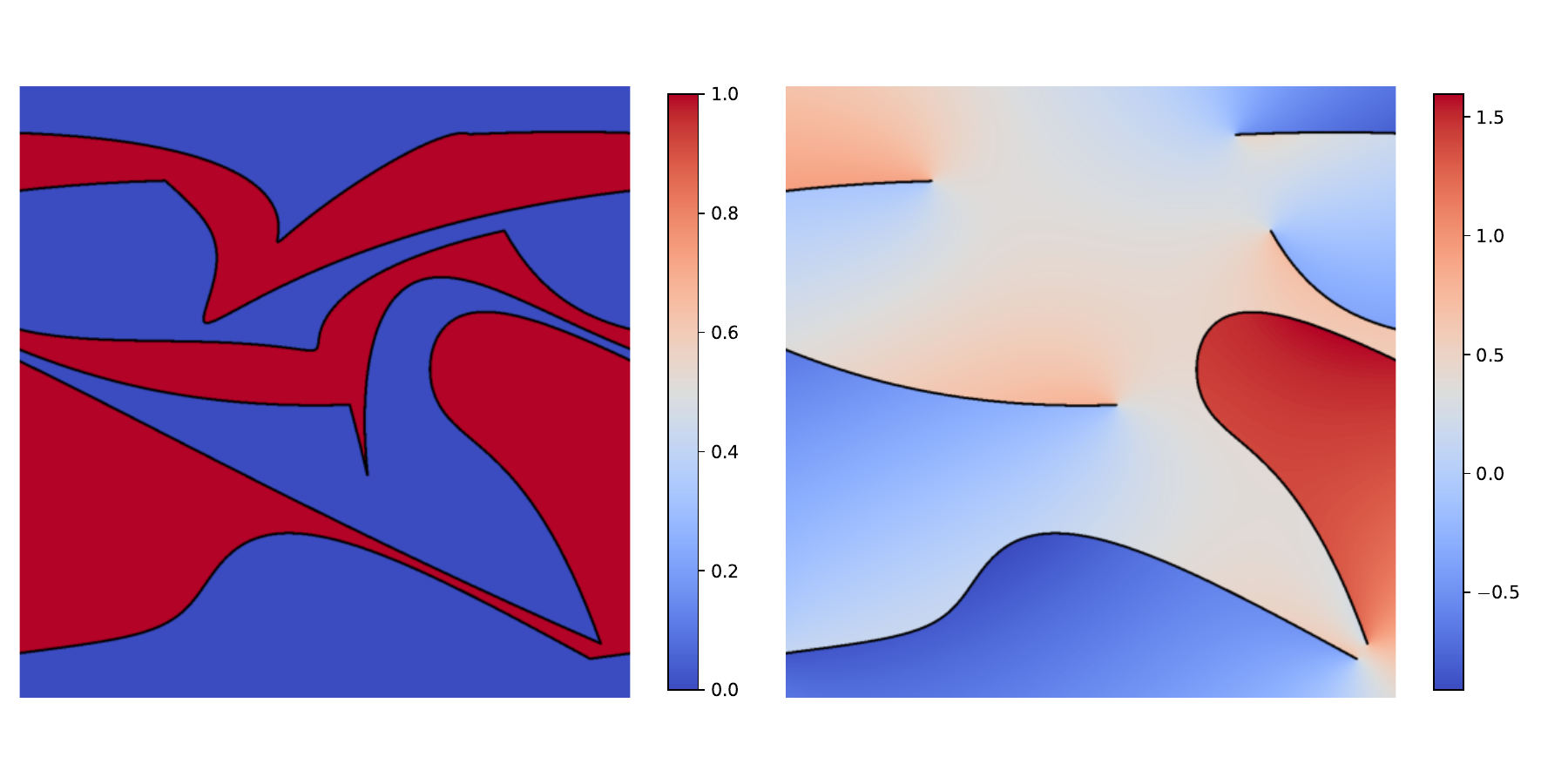}}
    \subfigure[24]{
    \includegraphics[width=0.24\linewidth]{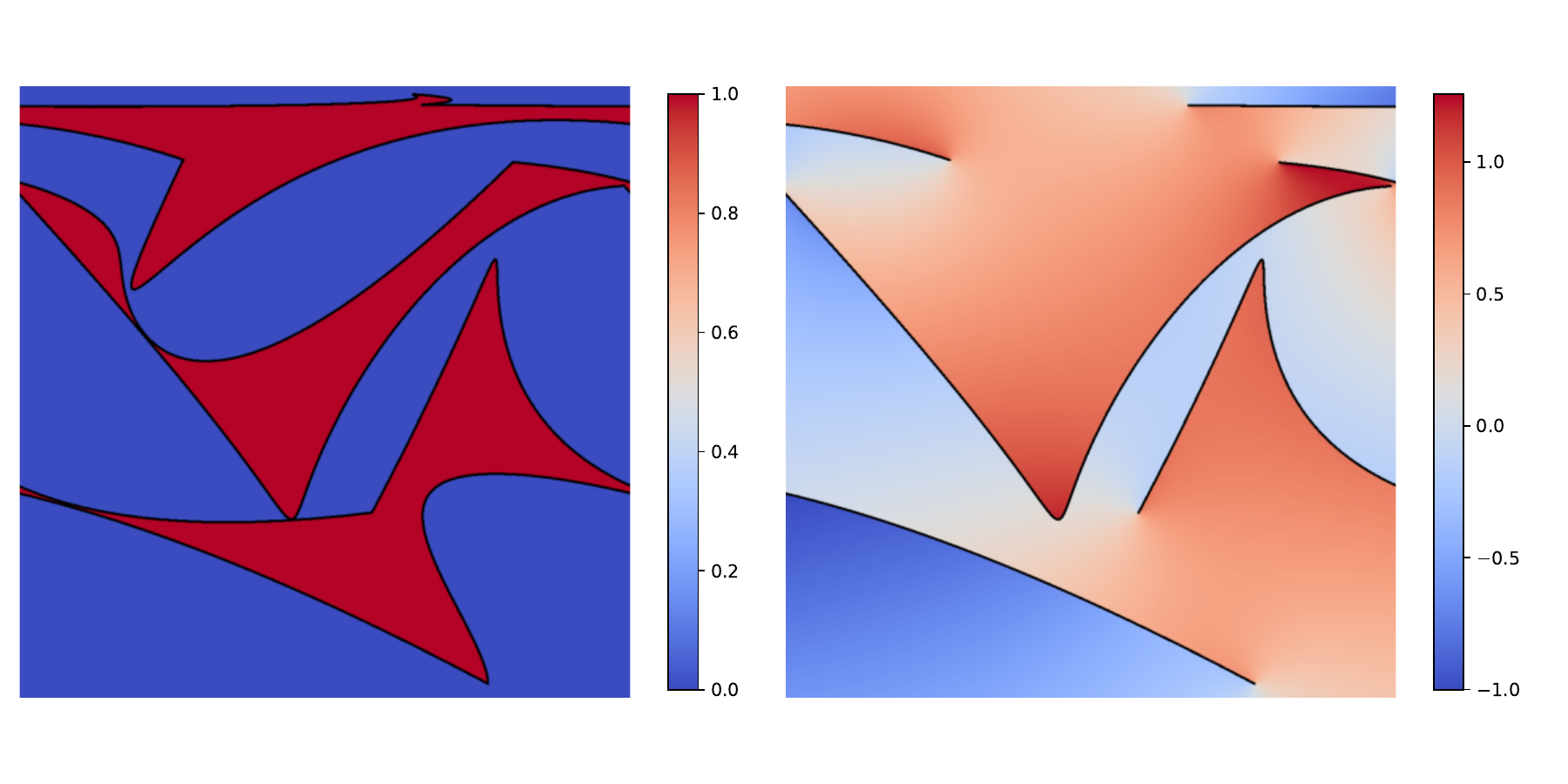}}
    \subfigure[37]{
    \includegraphics[width=0.24\linewidth]{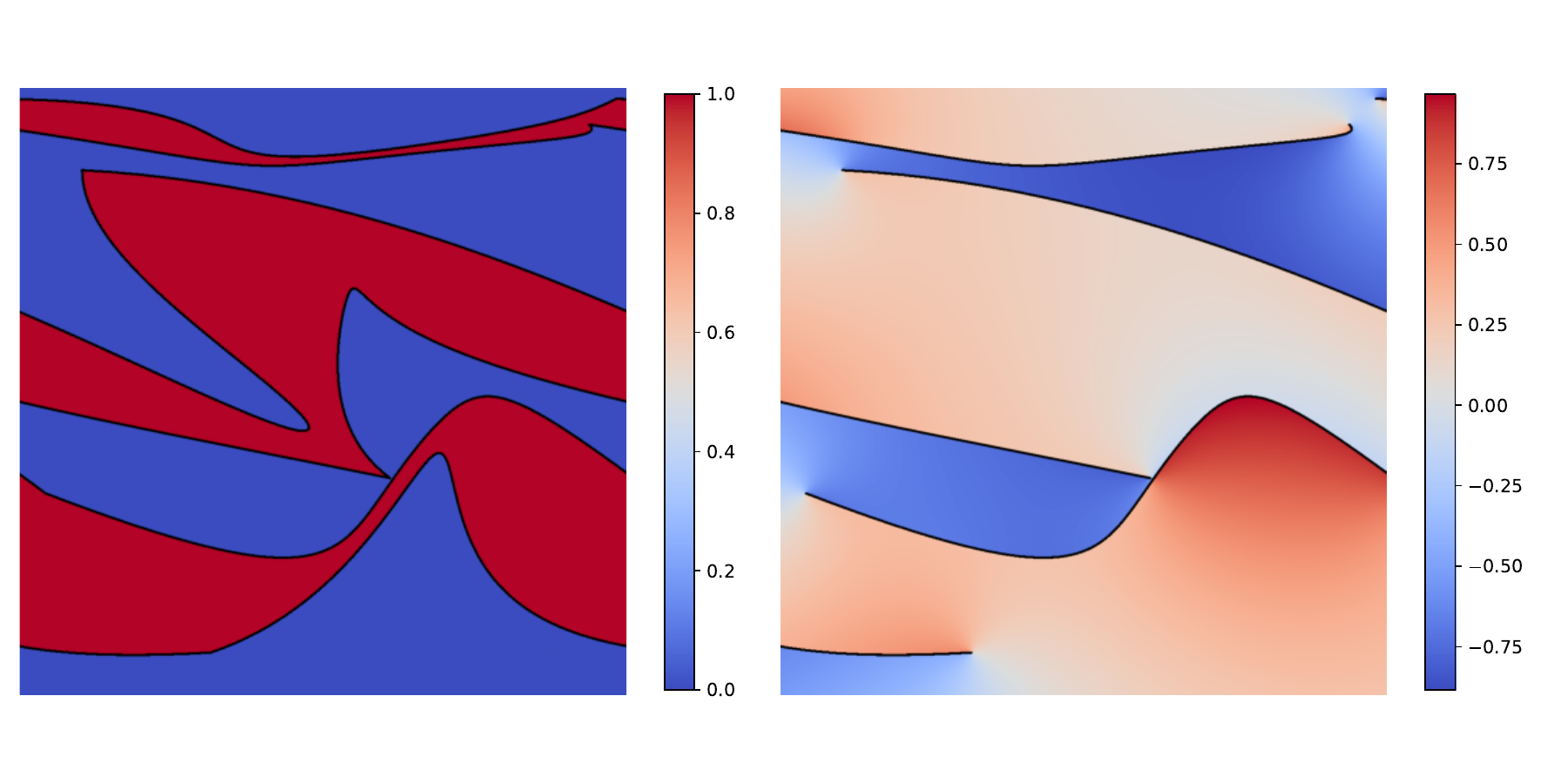}}
    \subfigure[57]{
    \includegraphics[width=0.24\linewidth]{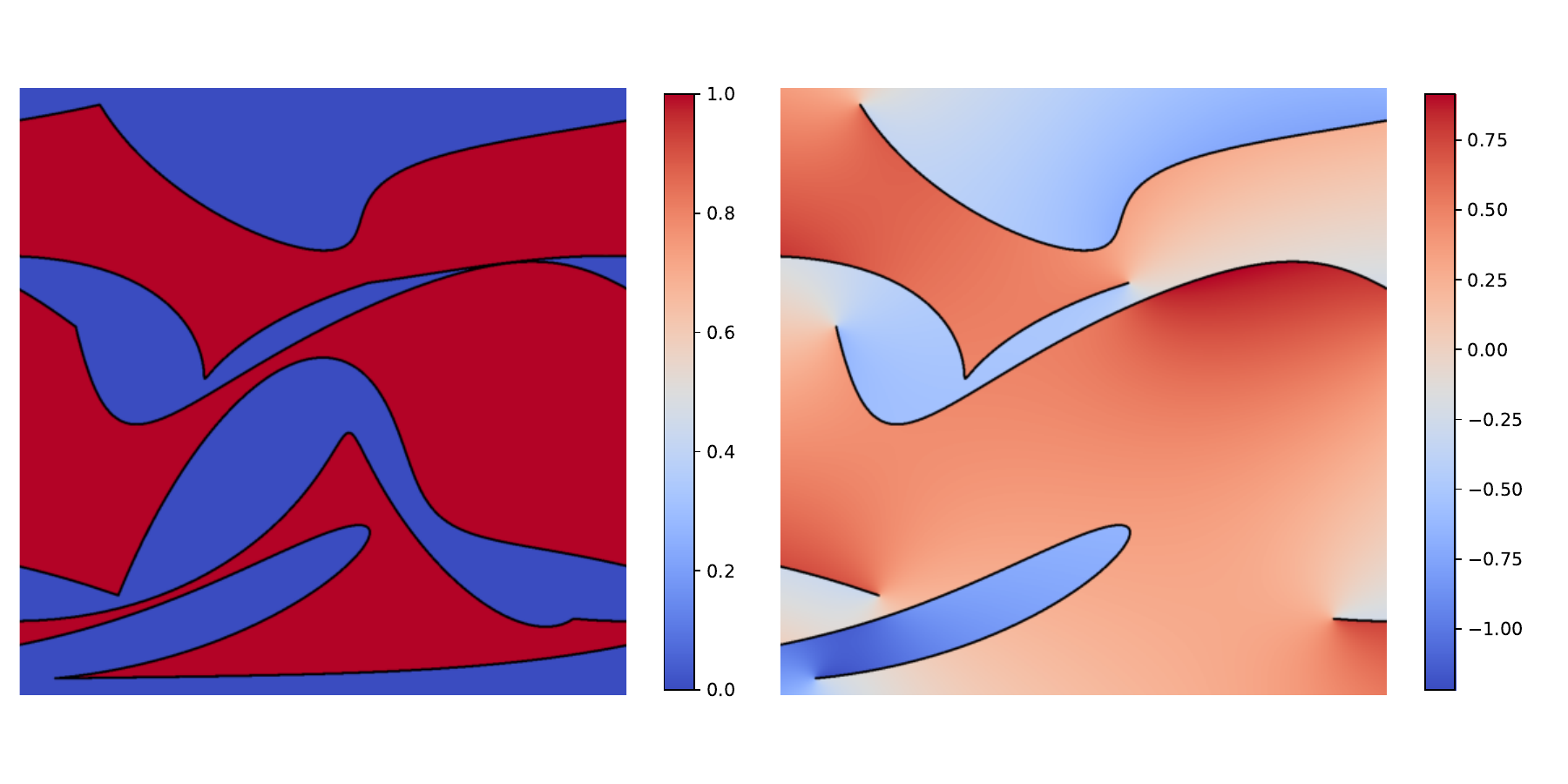}} \\
    \subfigure[64]{
    \includegraphics[width=0.24\linewidth]{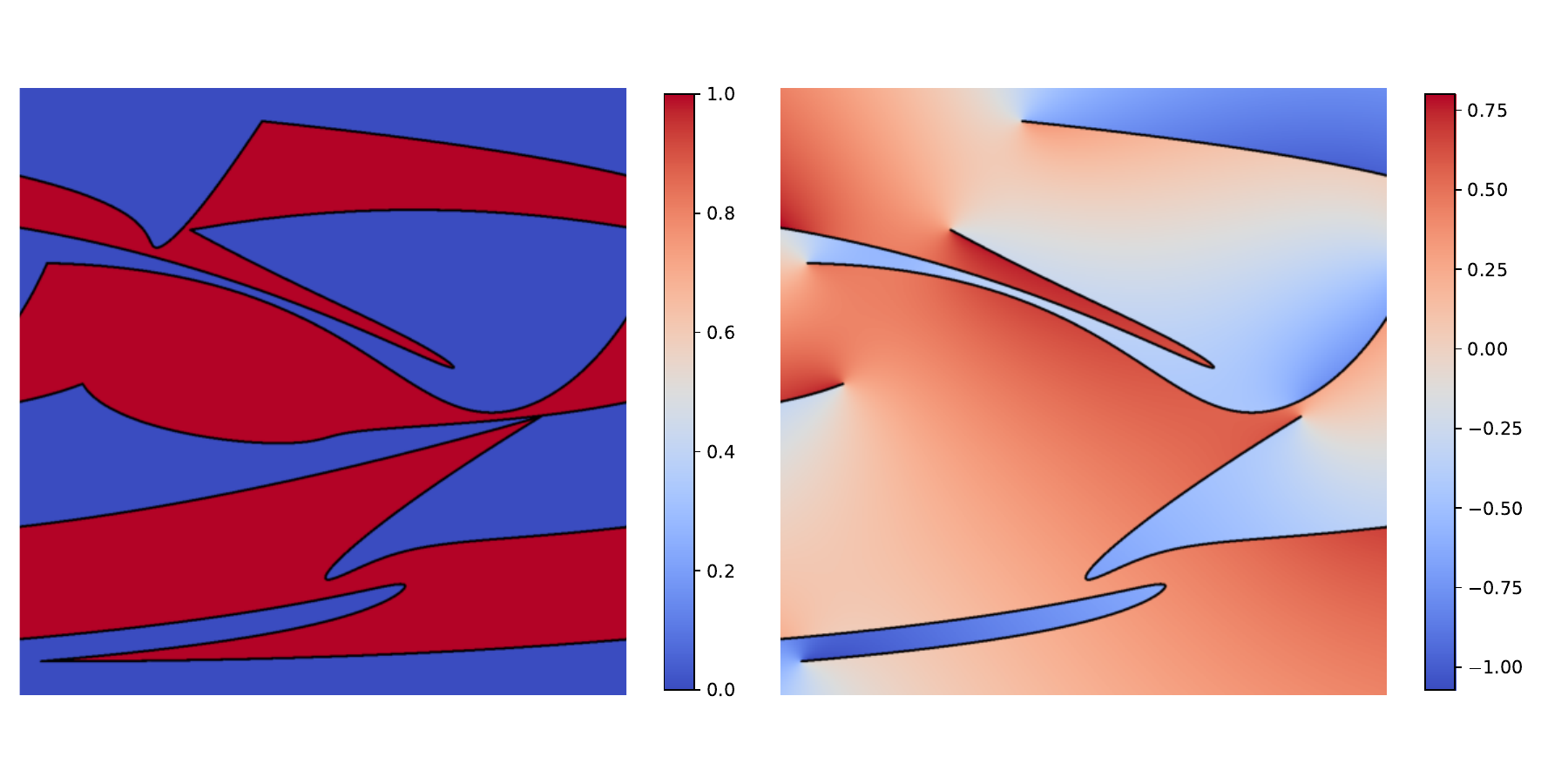}}
    \subfigure[48]{
    \includegraphics[width=0.24\linewidth]{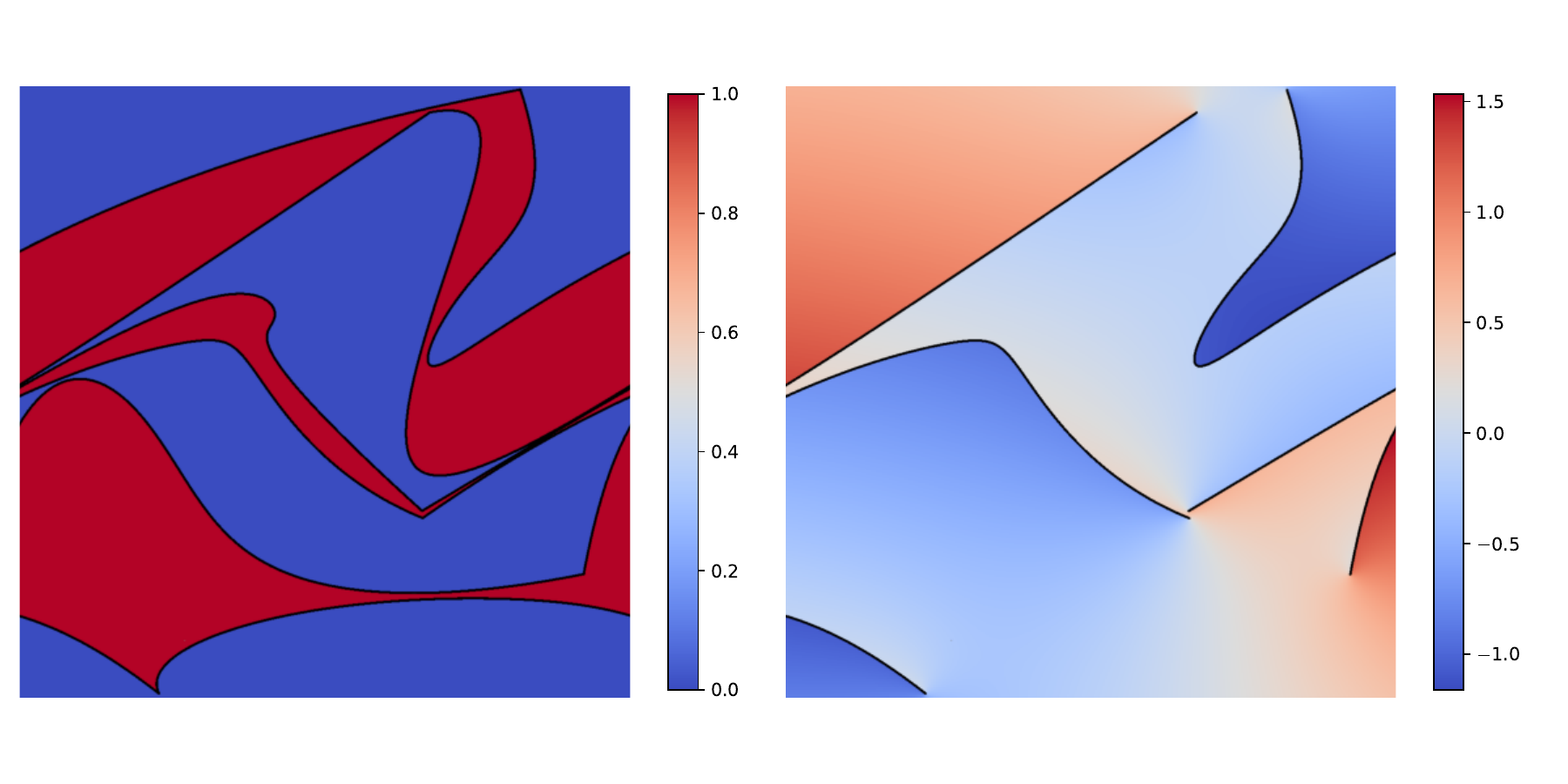}}
    \subfigure[81]{
    \includegraphics[width=0.24\linewidth]{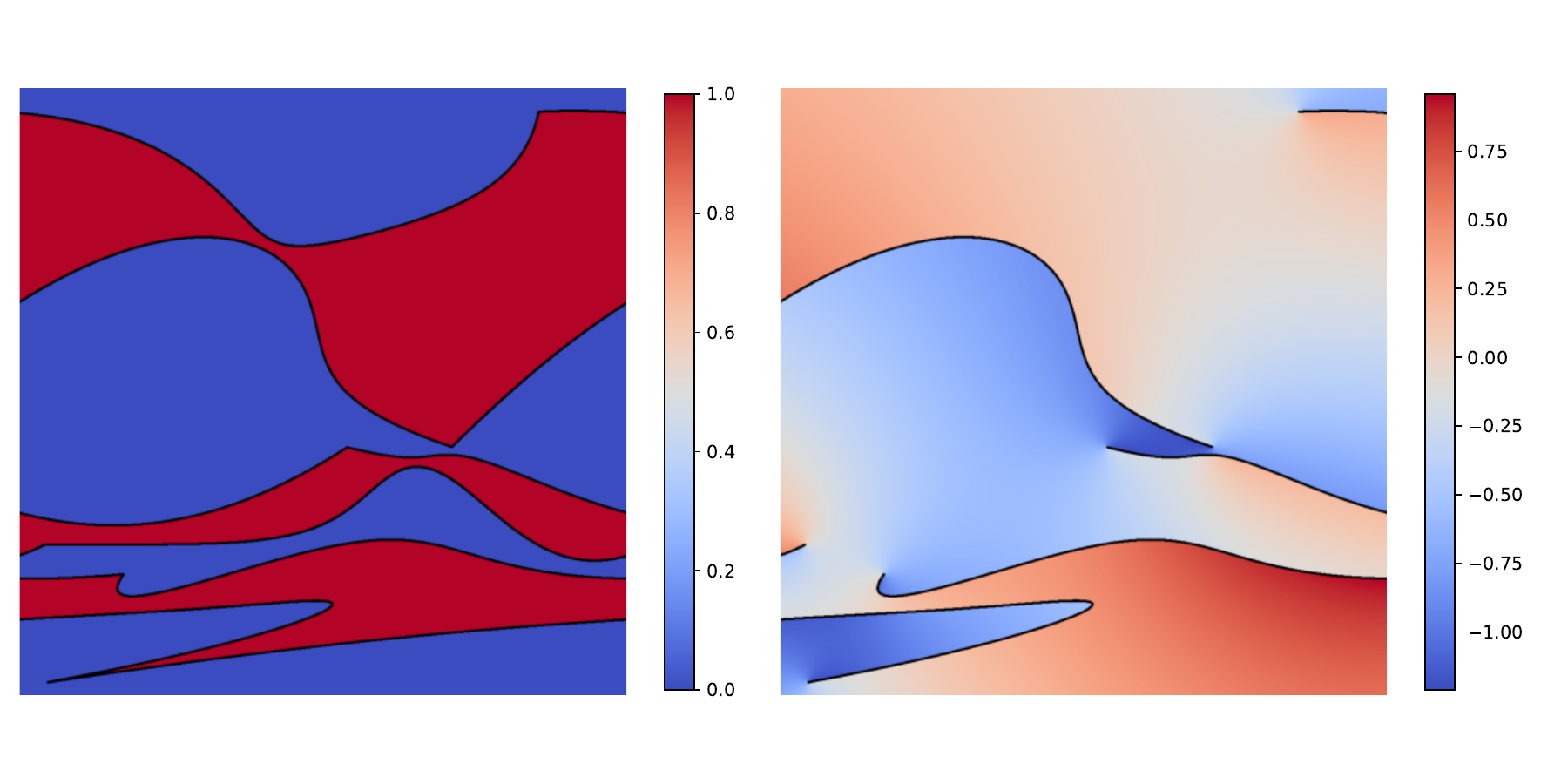}}
    \subfigure[105]{
    \includegraphics[width=0.24\linewidth]{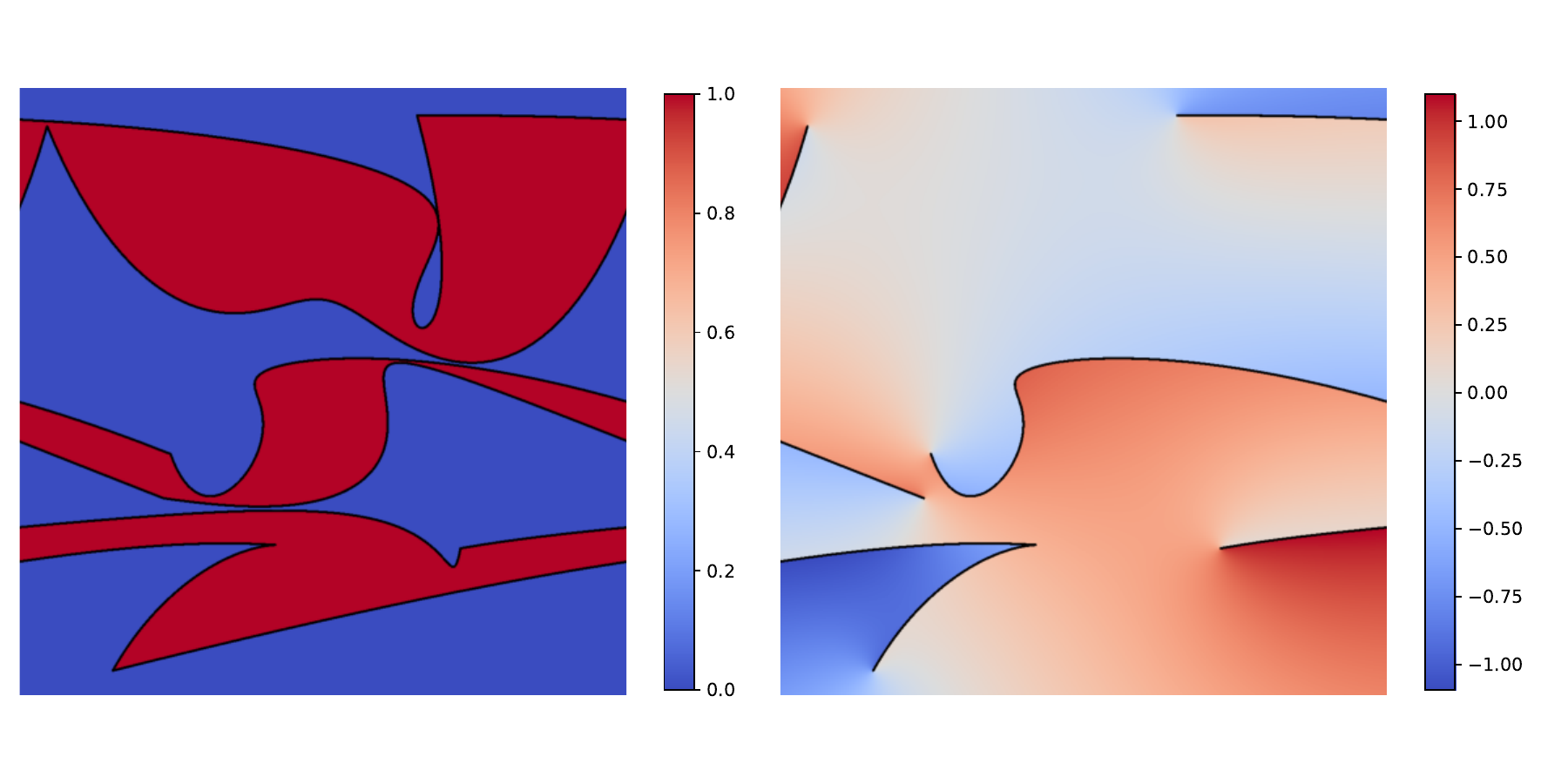}}
    \subfigure[117]{
    \includegraphics[width=0.24\linewidth]{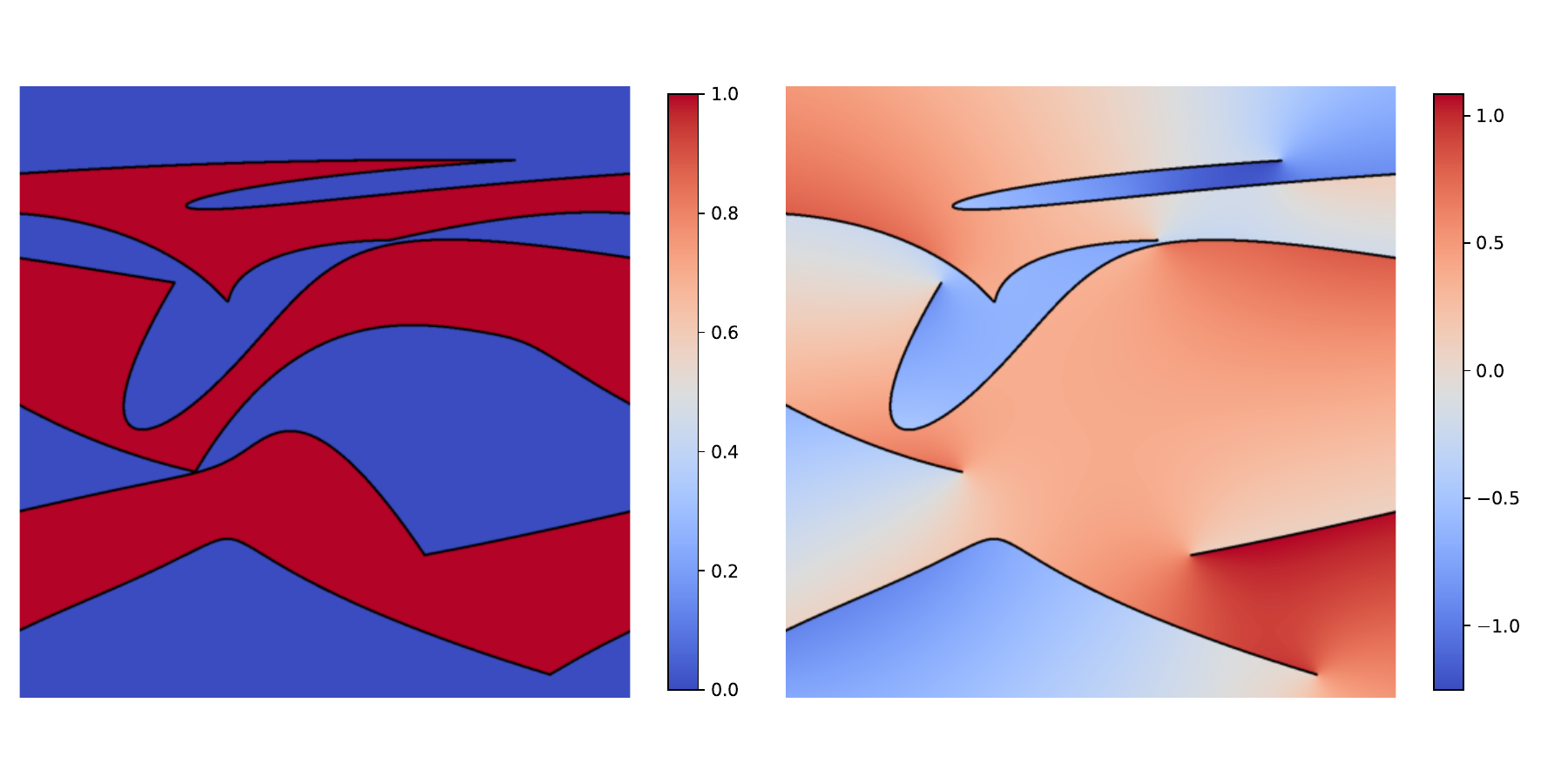}}
    \subfigure[136]{
    \includegraphics[width=0.24\linewidth]{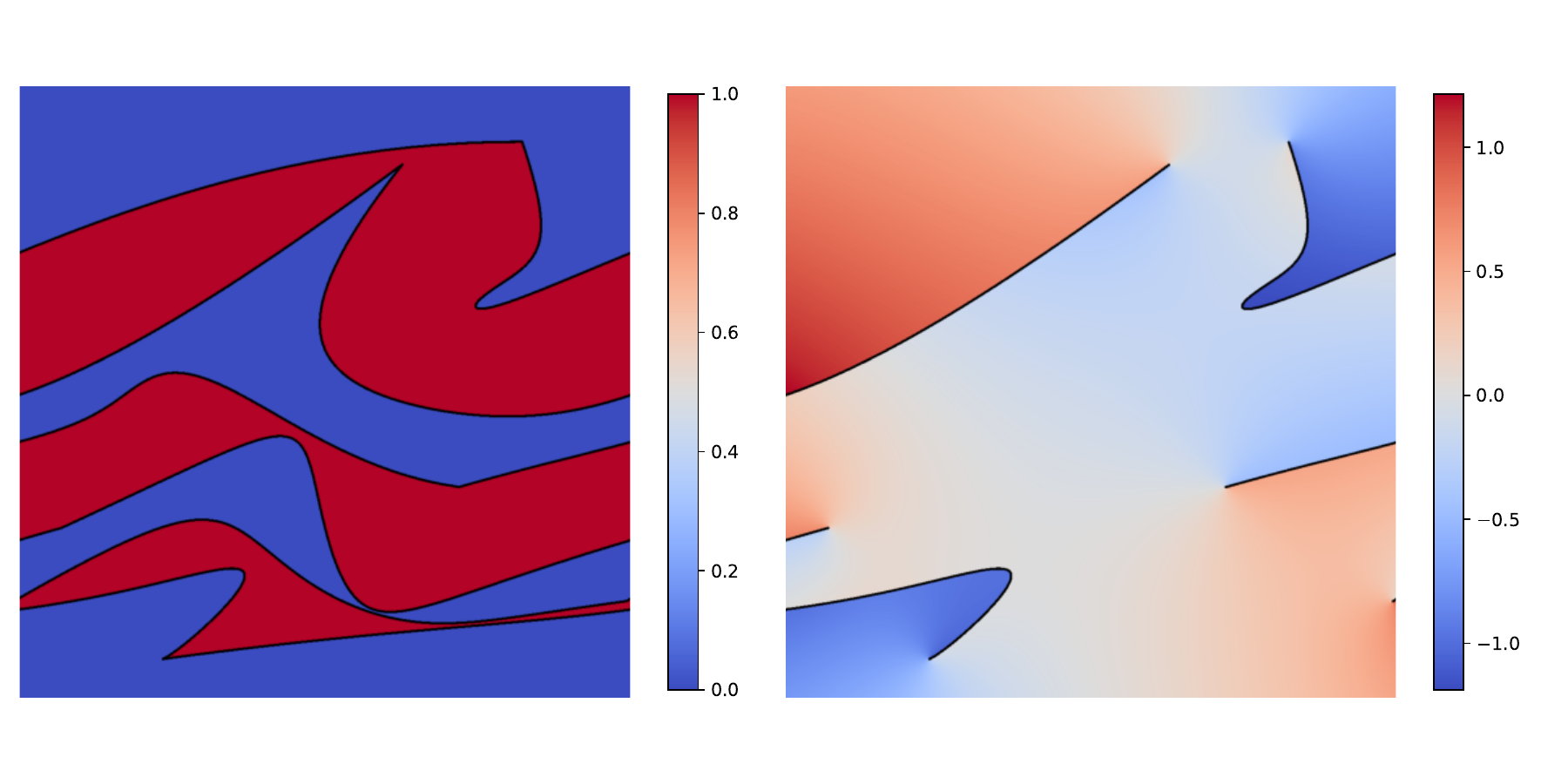}}
    \subfigure[156]{
    \includegraphics[width=0.24\linewidth]{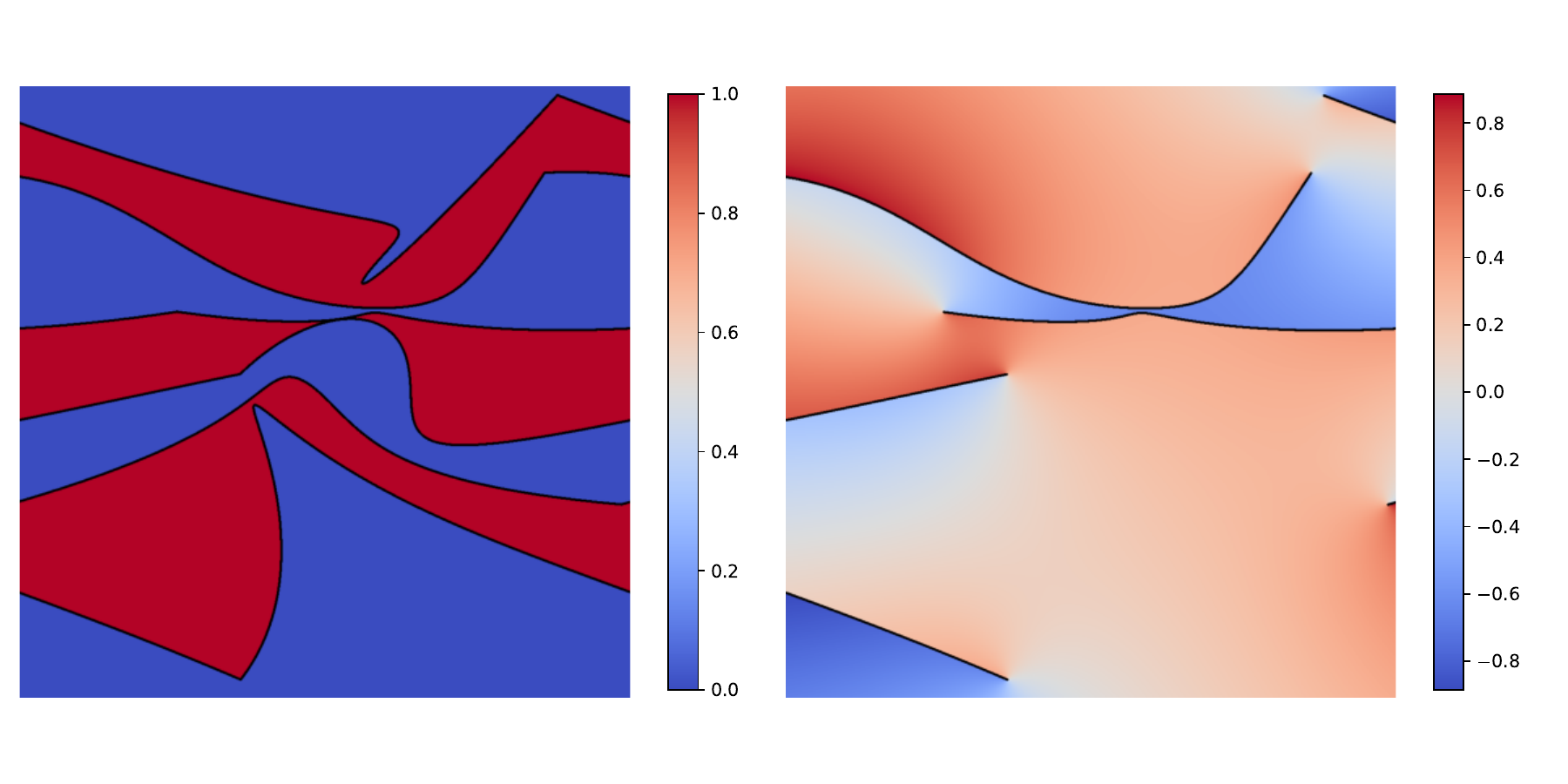}}
    \subfigure[173]{
    \includegraphics[width=0.24\linewidth]{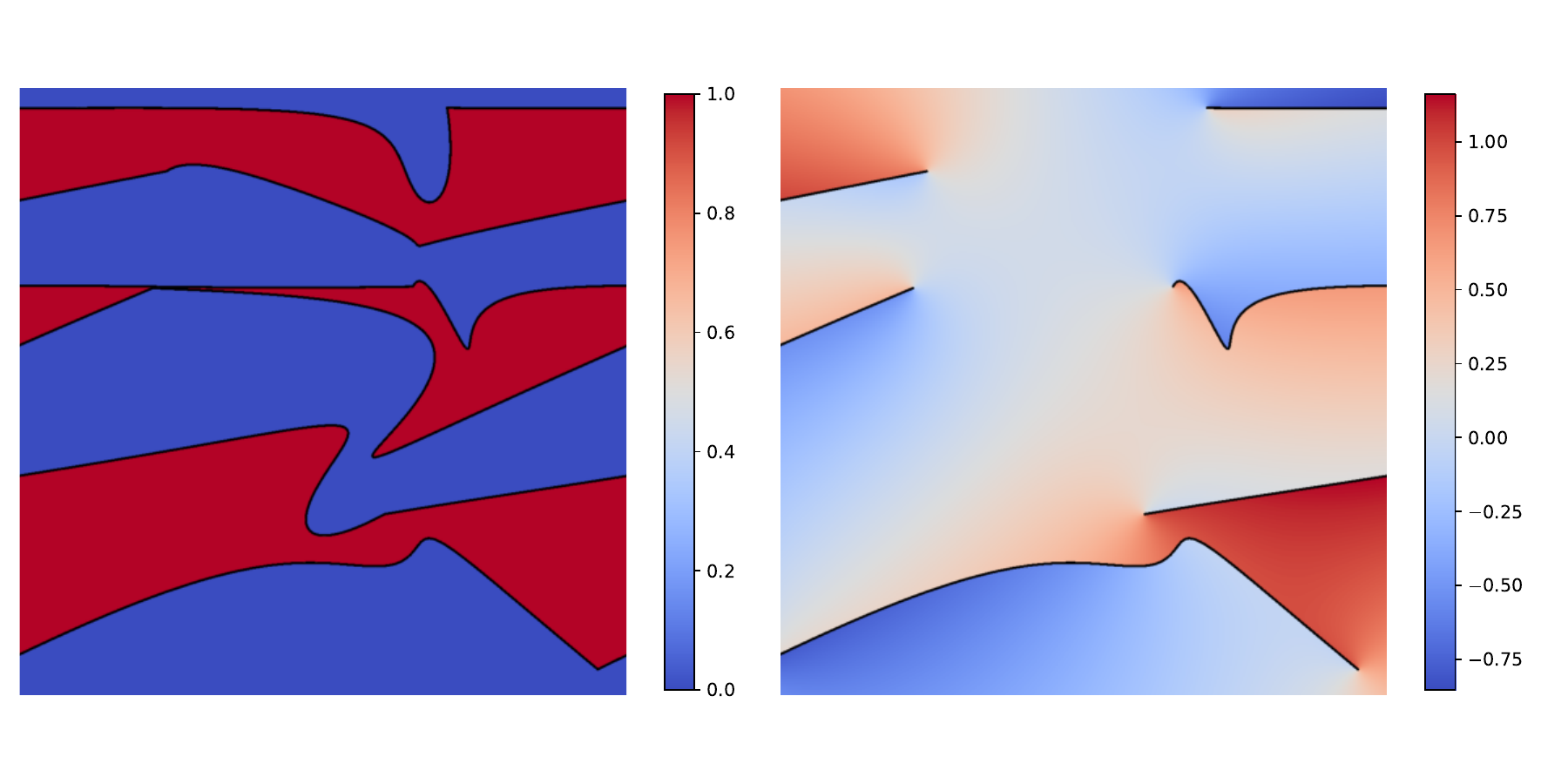}} \\
    \caption{Samples from the periodic dataset A. For each subfigure, the left part is the winding number computed by the periodic method in covering space and the right is the simply computed winding number in parametric domain.}
    \label{fig:periodic_result_a}
\end{figure*}

\begin{figure*}[h]
    \centering
    \subfigbottomskip=2pt
    \subfigcapskip=-5pt
    \subfigure[21]{
    \includegraphics[width=0.24\linewidth]{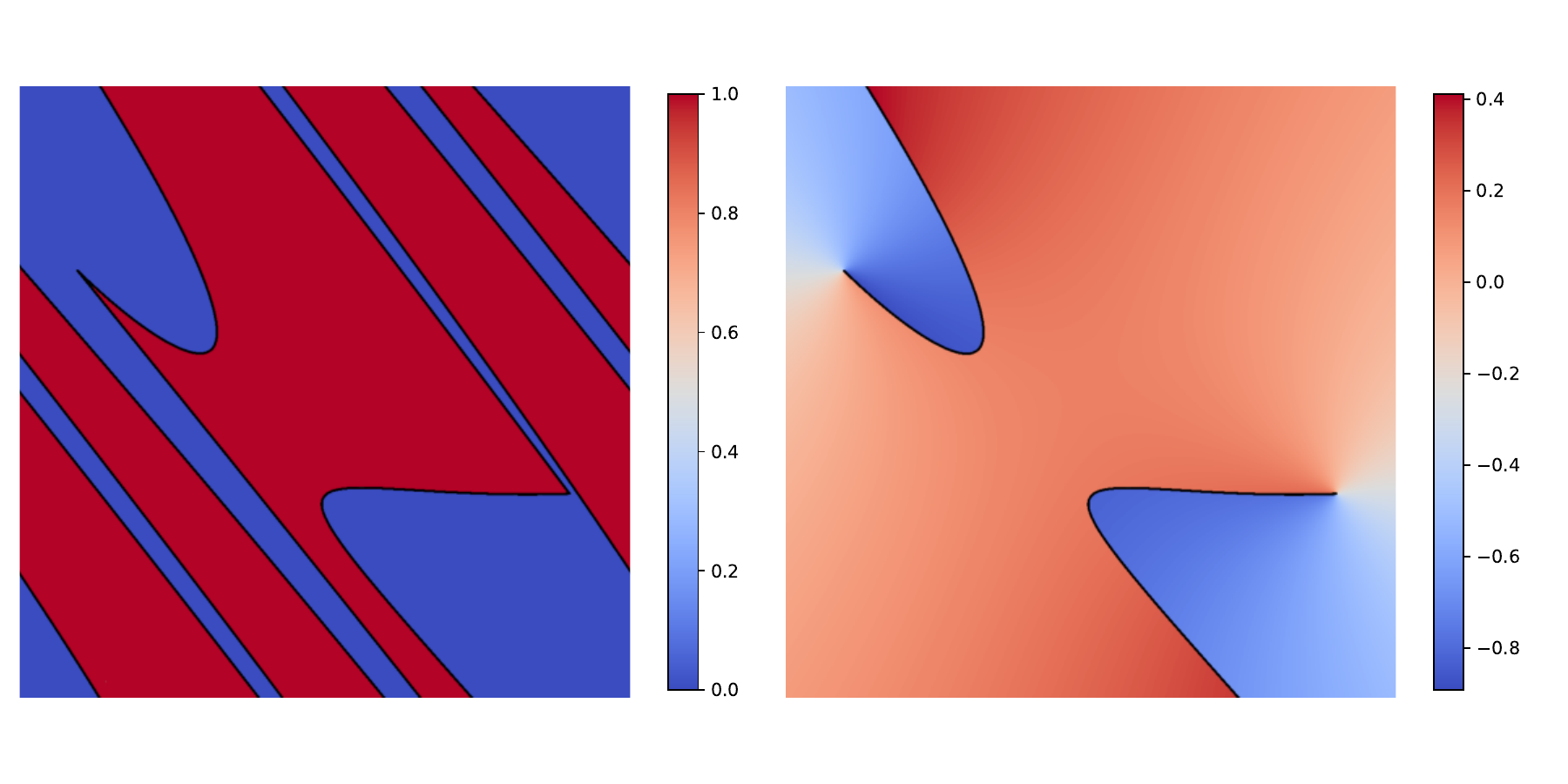}}
    \subfigure[35]{
    \includegraphics[width=0.24\linewidth]{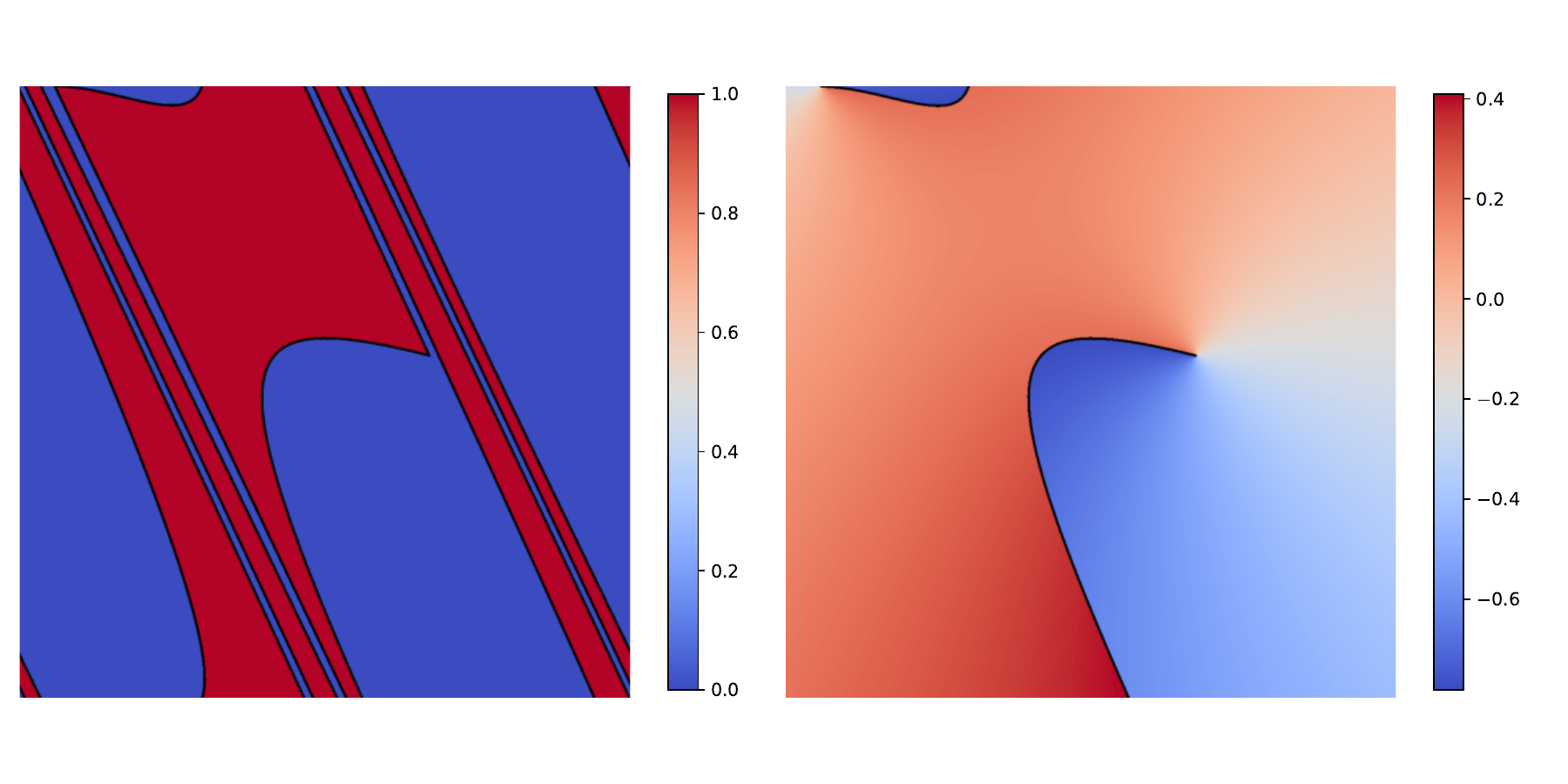}}
    \subfigure[63]{
    \includegraphics[width=0.24\linewidth]{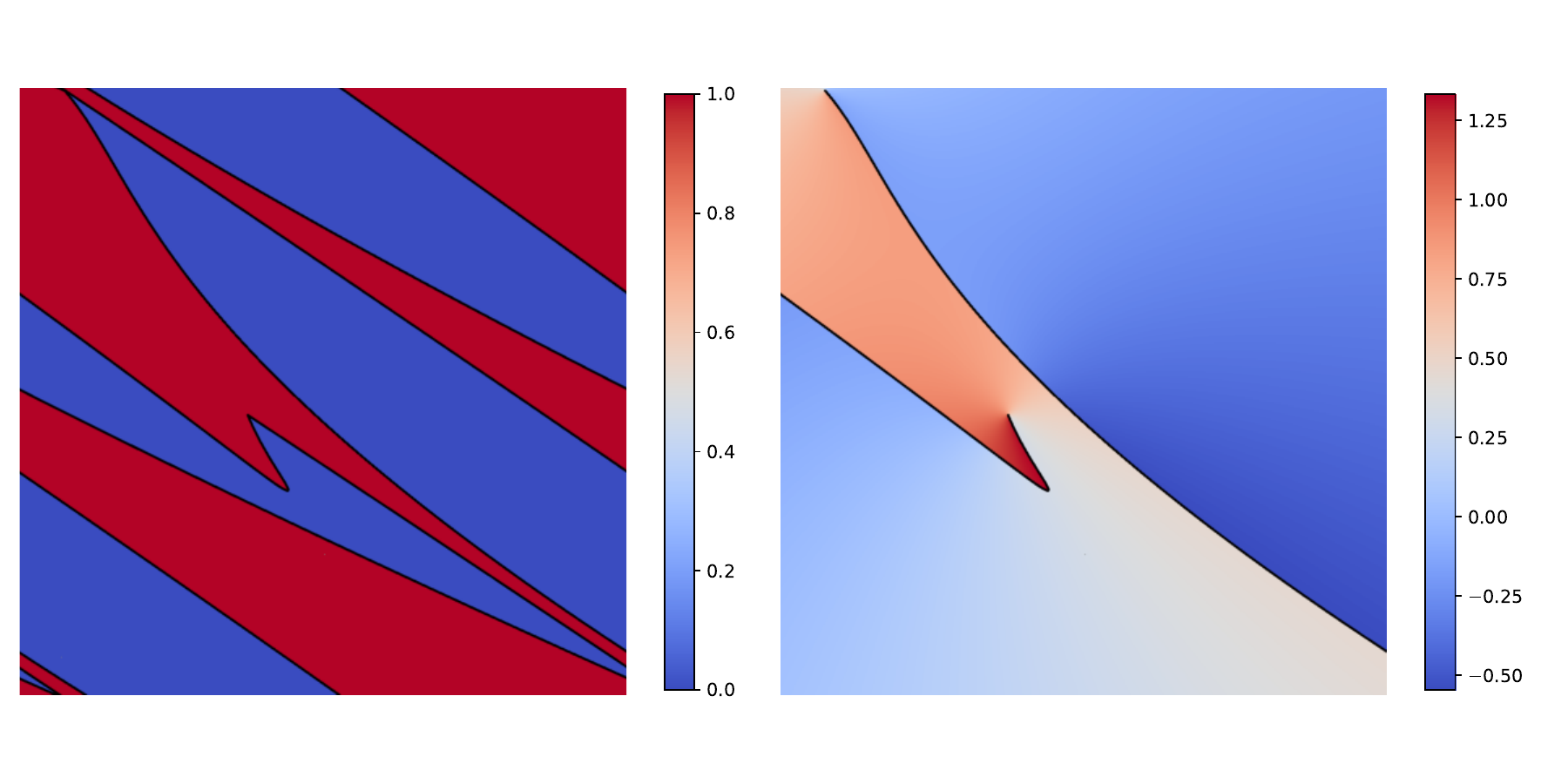}}
    \subfigure[86]{
    \includegraphics[width=0.24\linewidth]{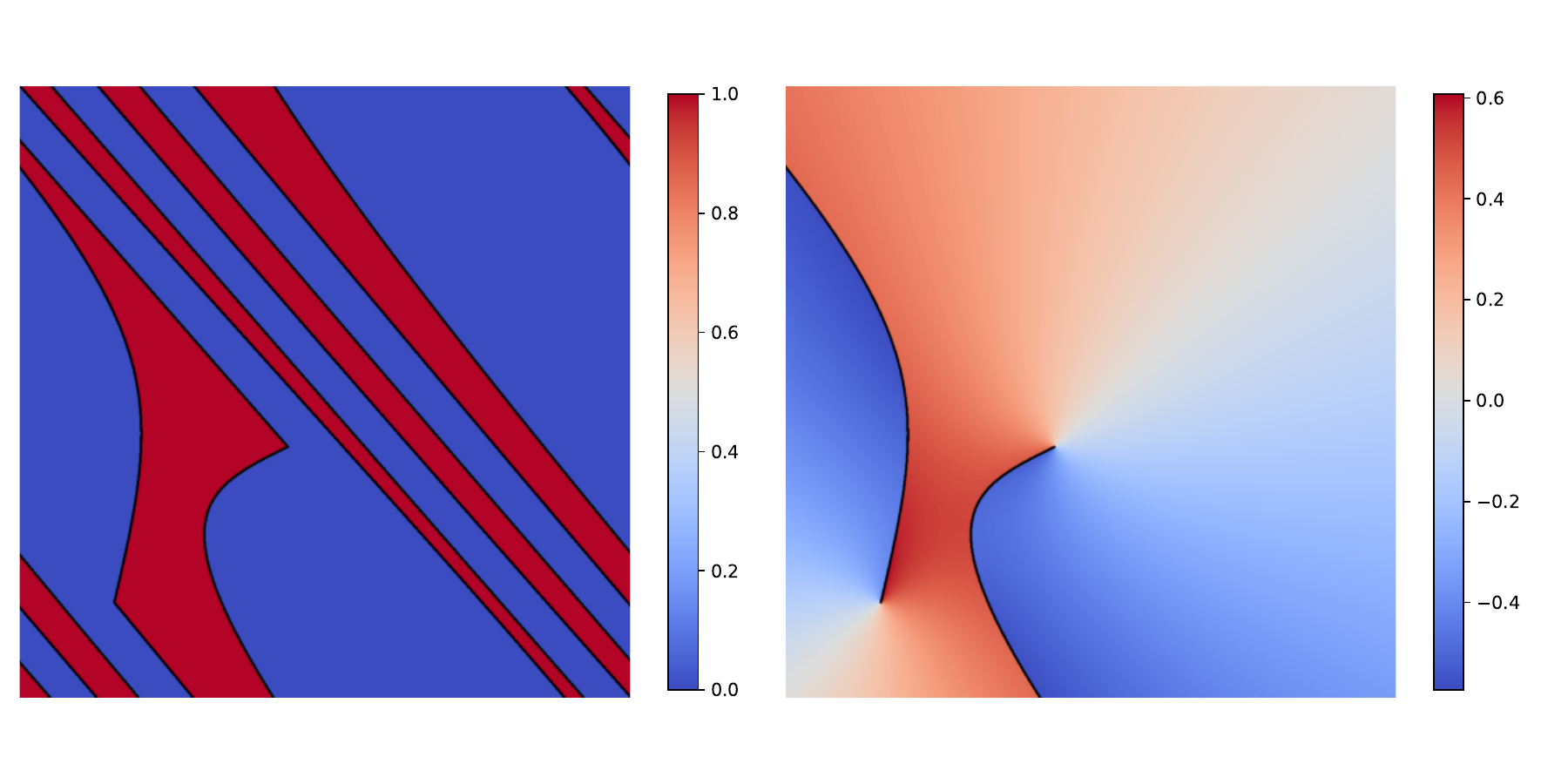}} \\
    \subfigure[93]{
    \includegraphics[width=0.24\linewidth]{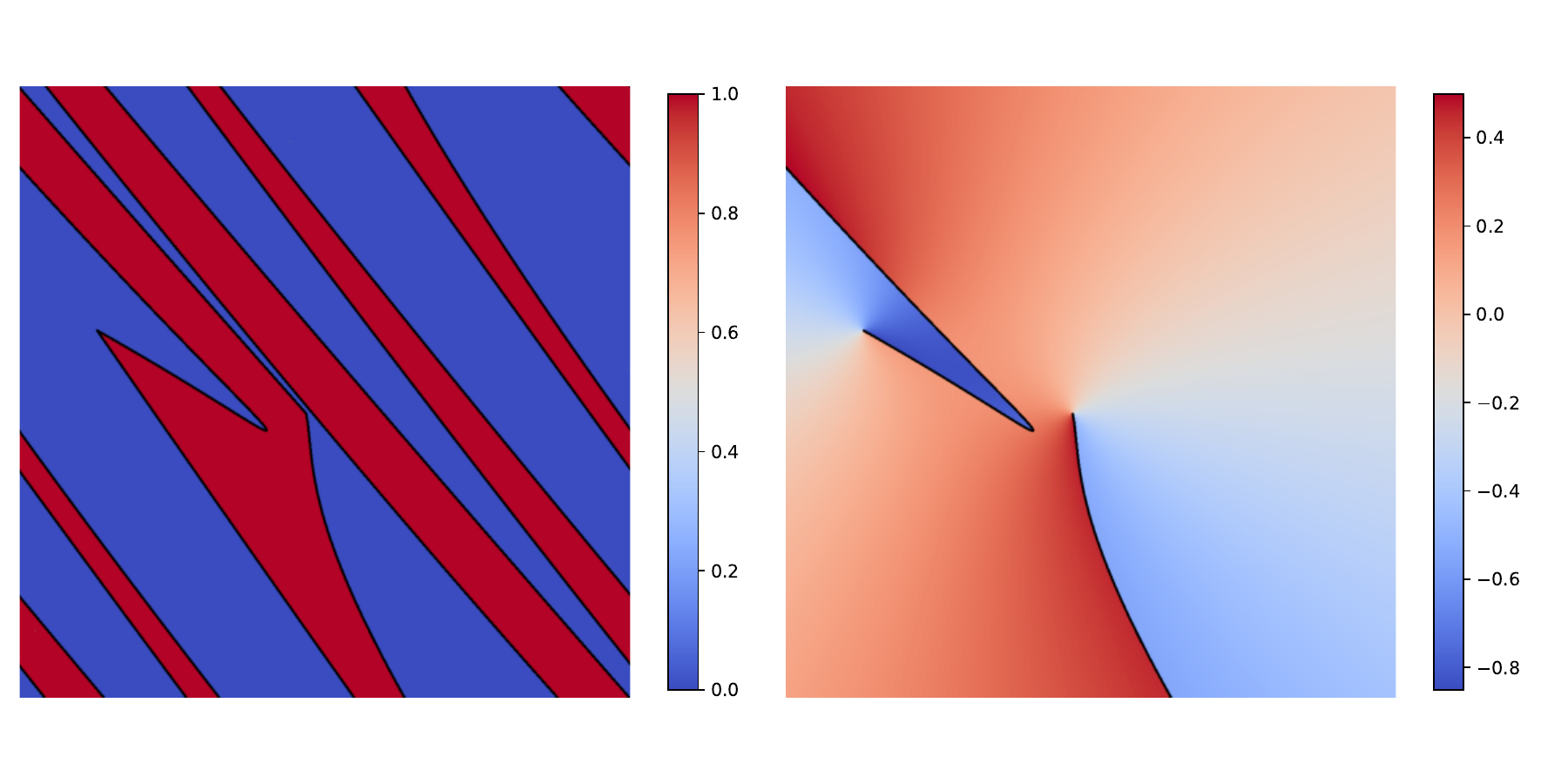}}
    \subfigure[107]{
    \includegraphics[width=0.24\linewidth]{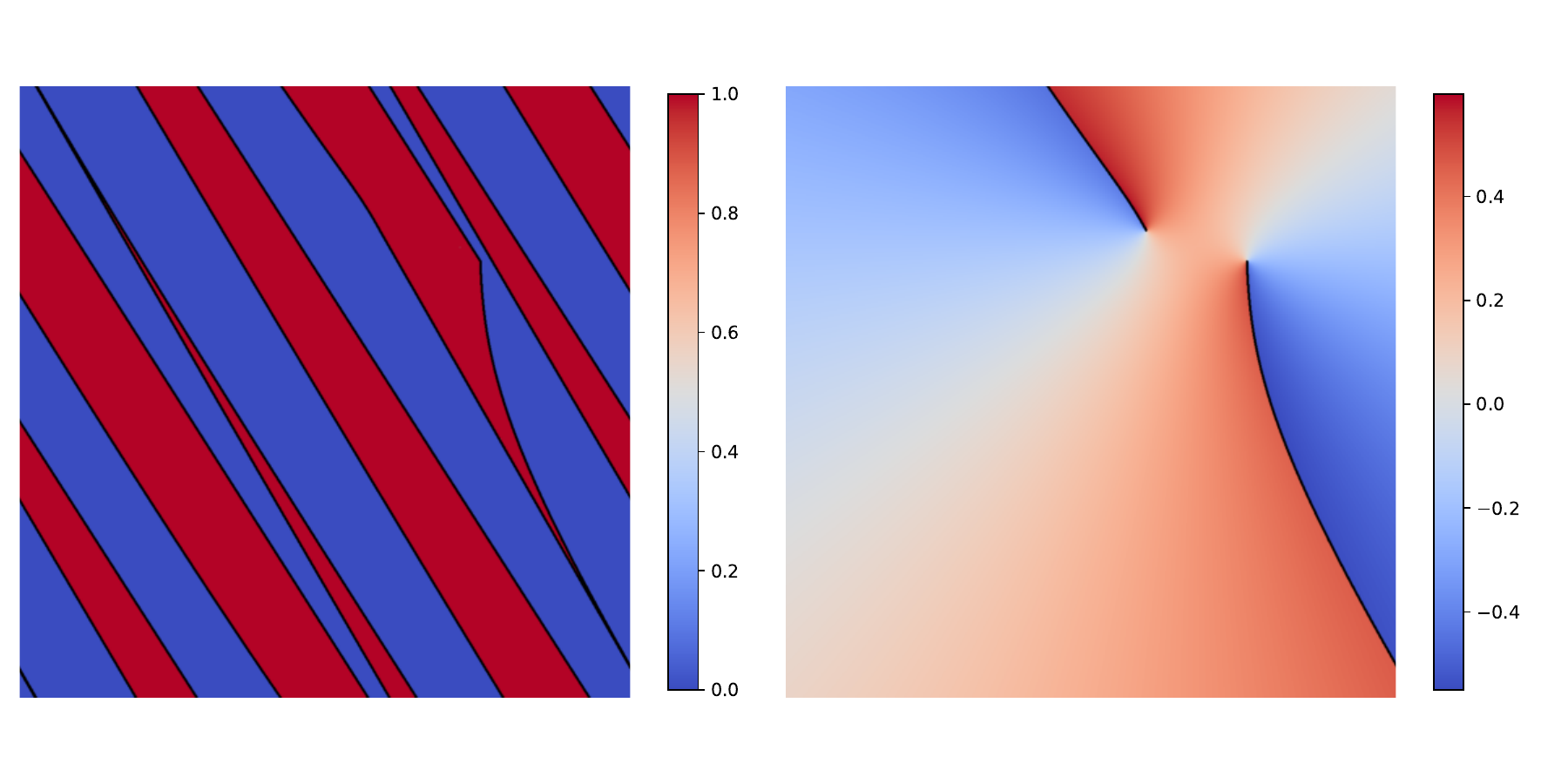}}
    \subfigure[109]{
    \includegraphics[width=0.24\linewidth]{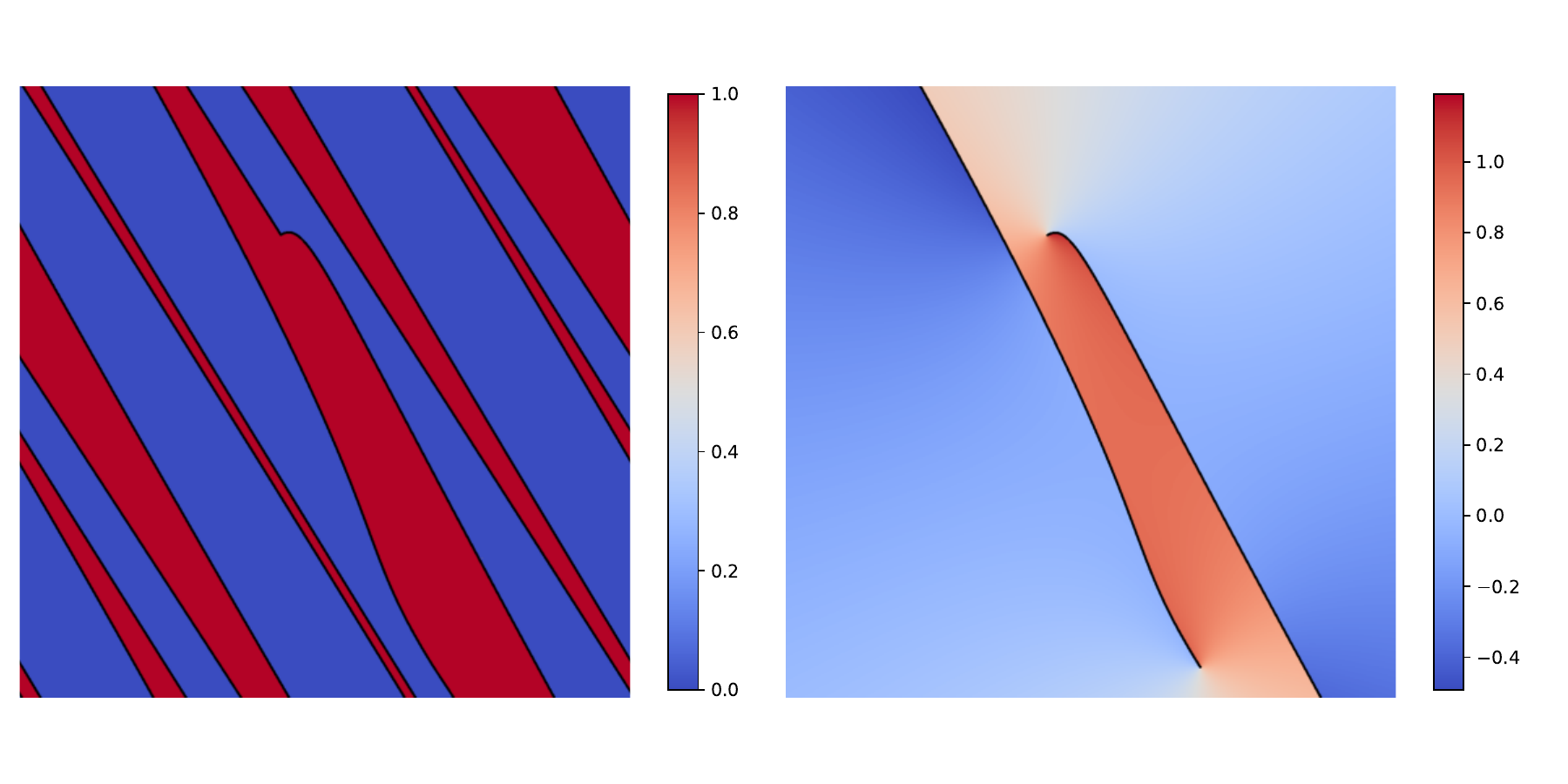}}
    \subfigure[115]{
    \includegraphics[width=0.24\linewidth]{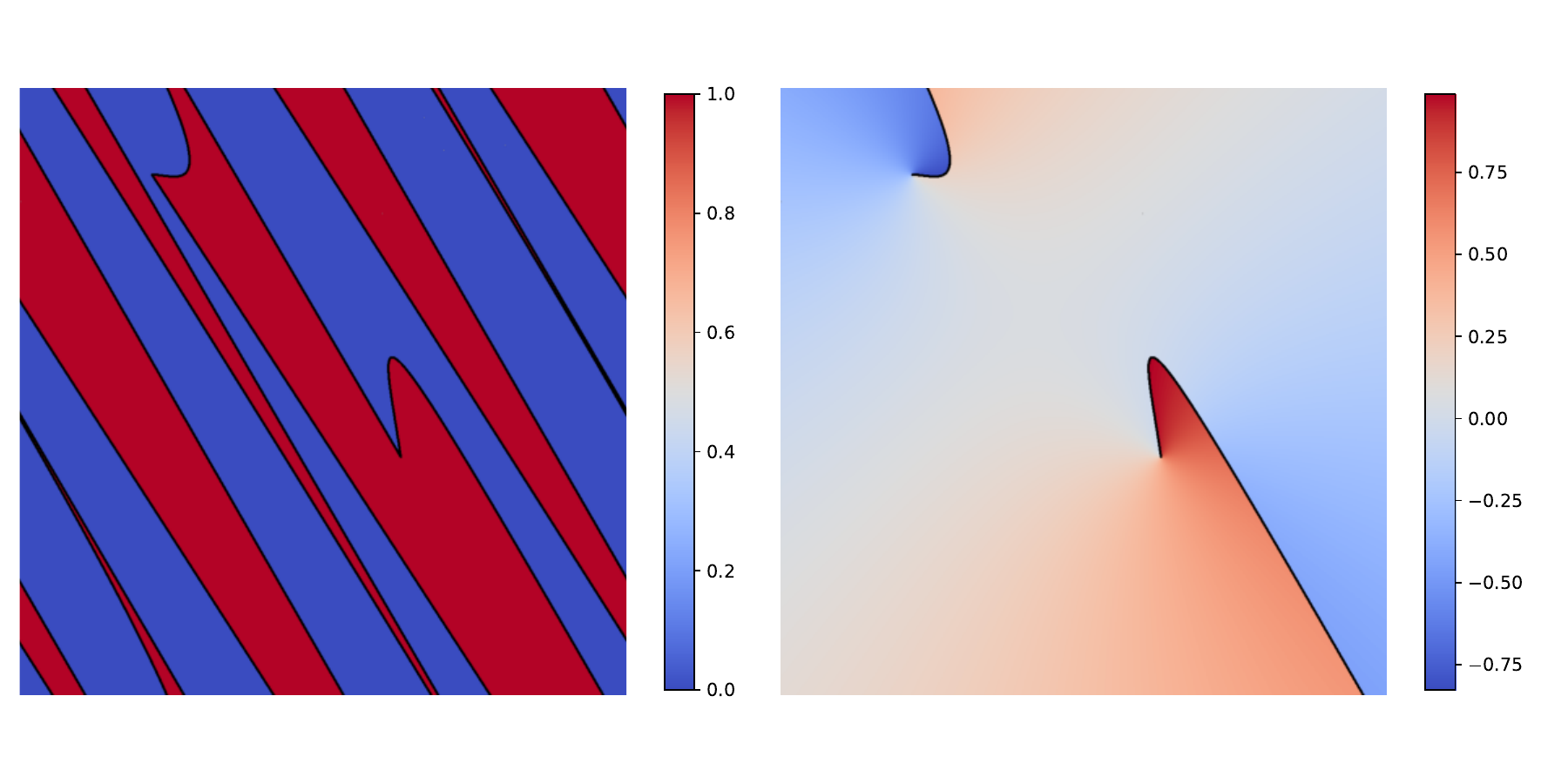}}
    \subfigure[132]{
    \includegraphics[width=0.24\linewidth]{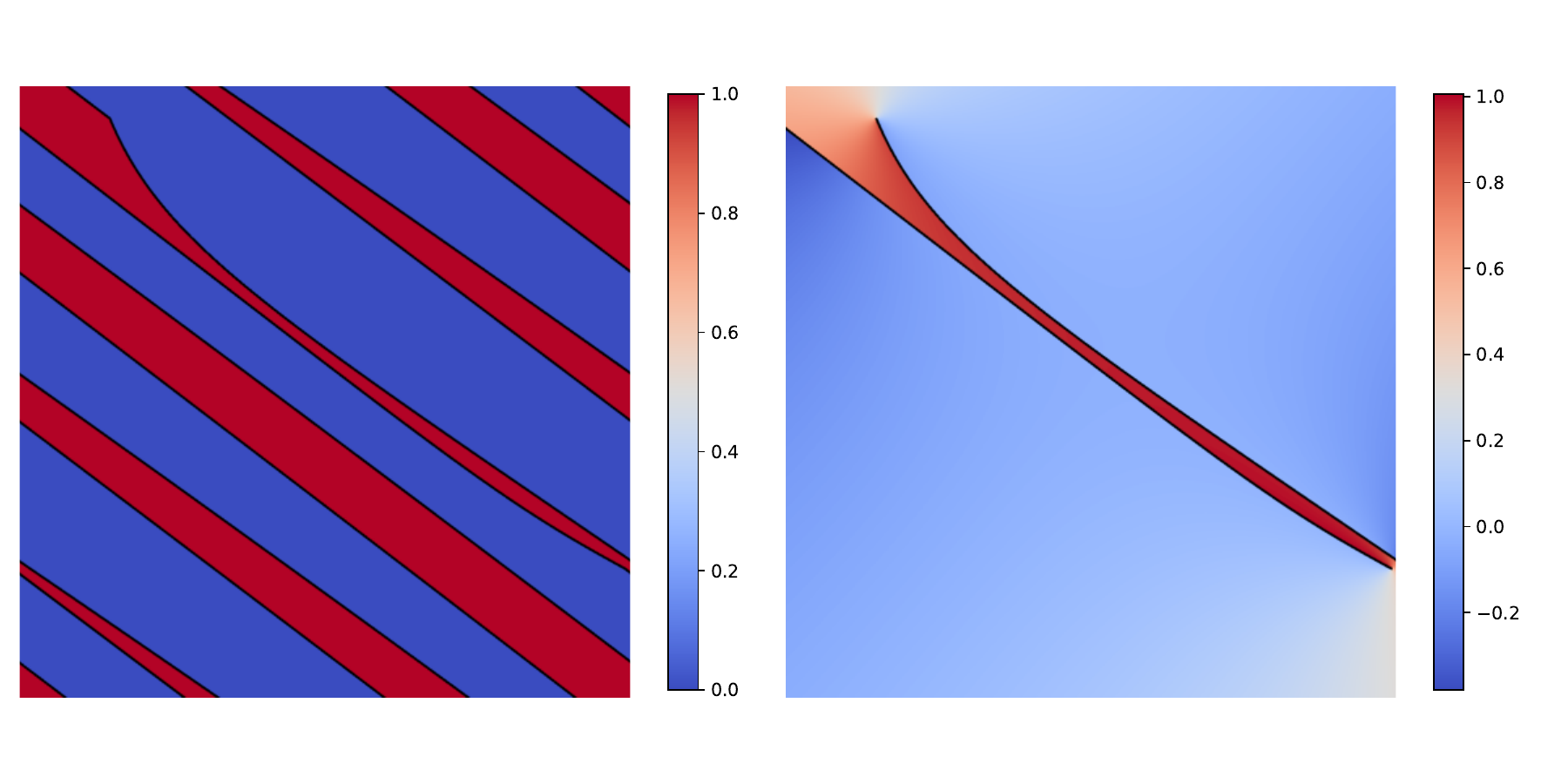}}
    \subfigure[146]{
    \includegraphics[width=0.24\linewidth]{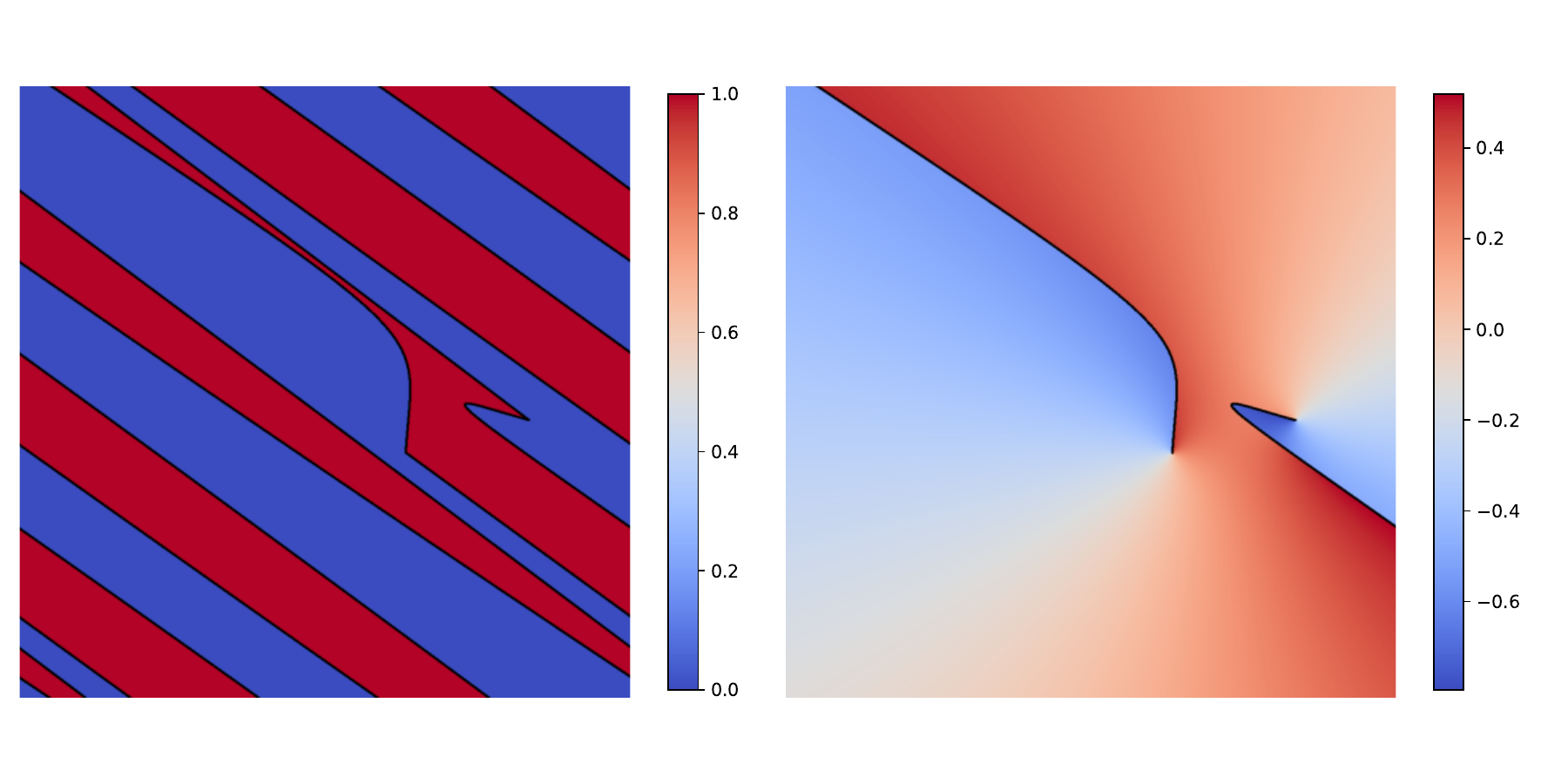}}
    \subfigure[151]{
    \includegraphics[width=0.24\linewidth]{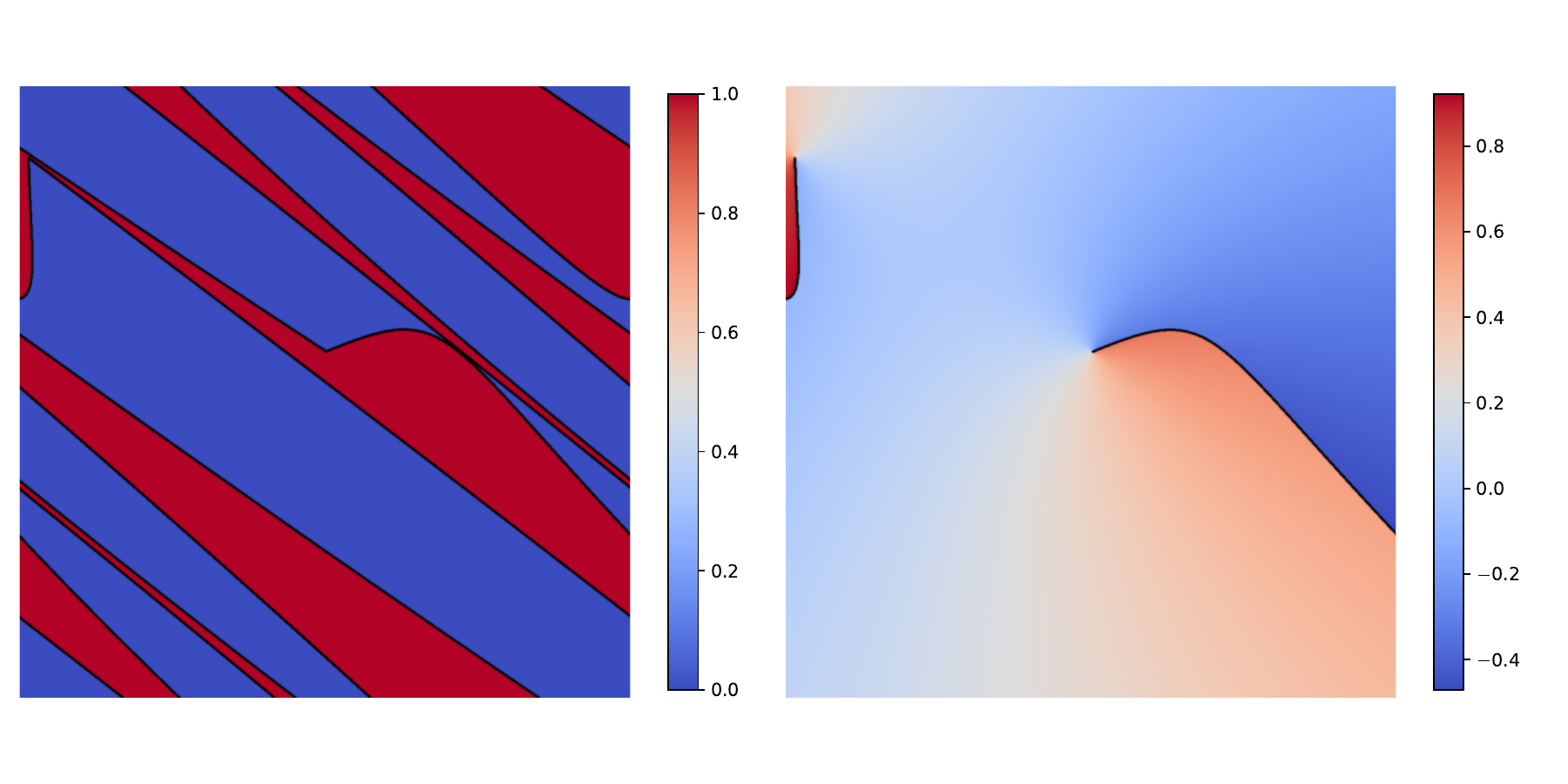}}
    \subfigure[167]{
    \includegraphics[width=0.24\linewidth]{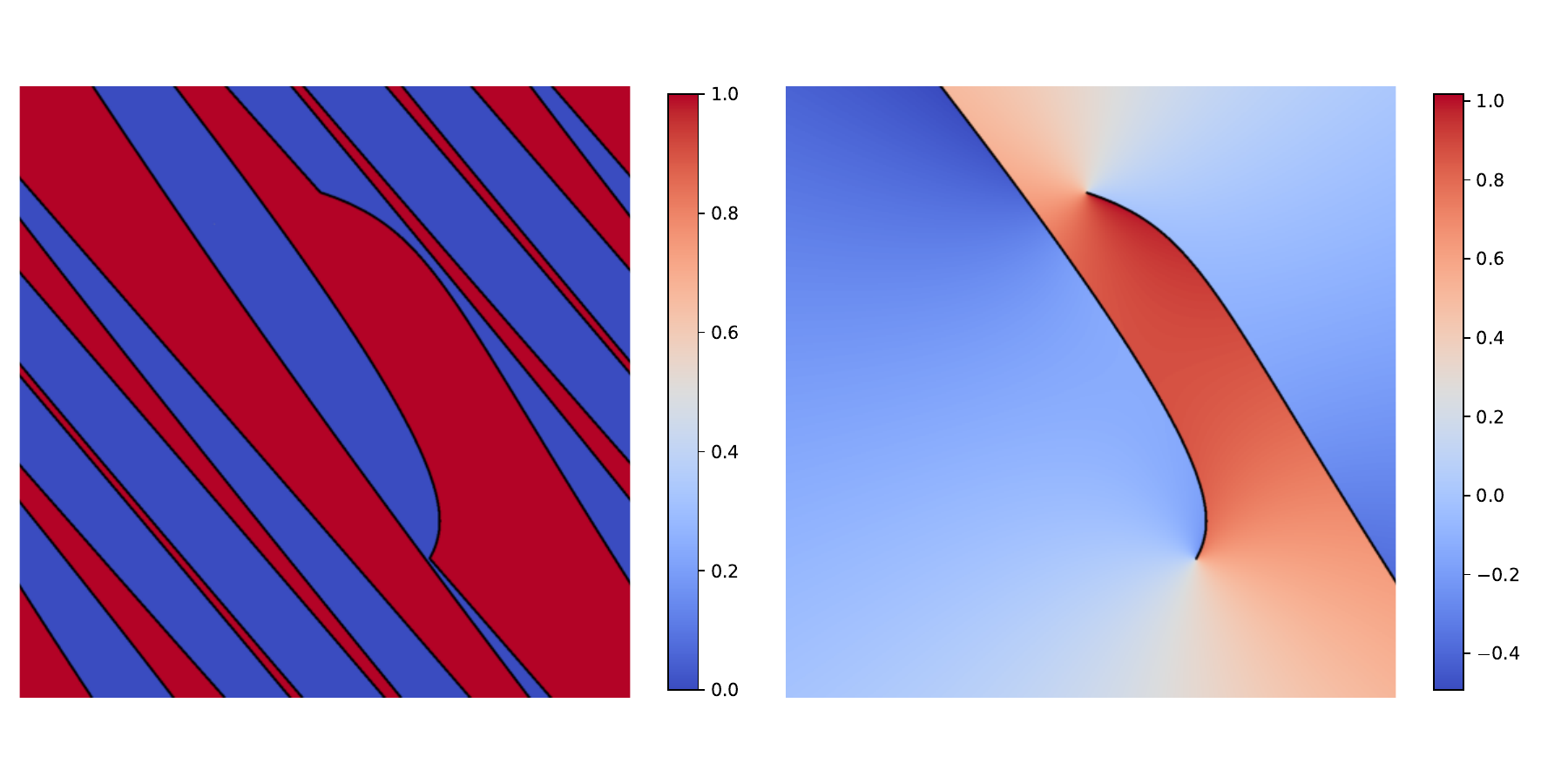}} \\
    \subfigure[171]{
    \includegraphics[width=0.24\linewidth]{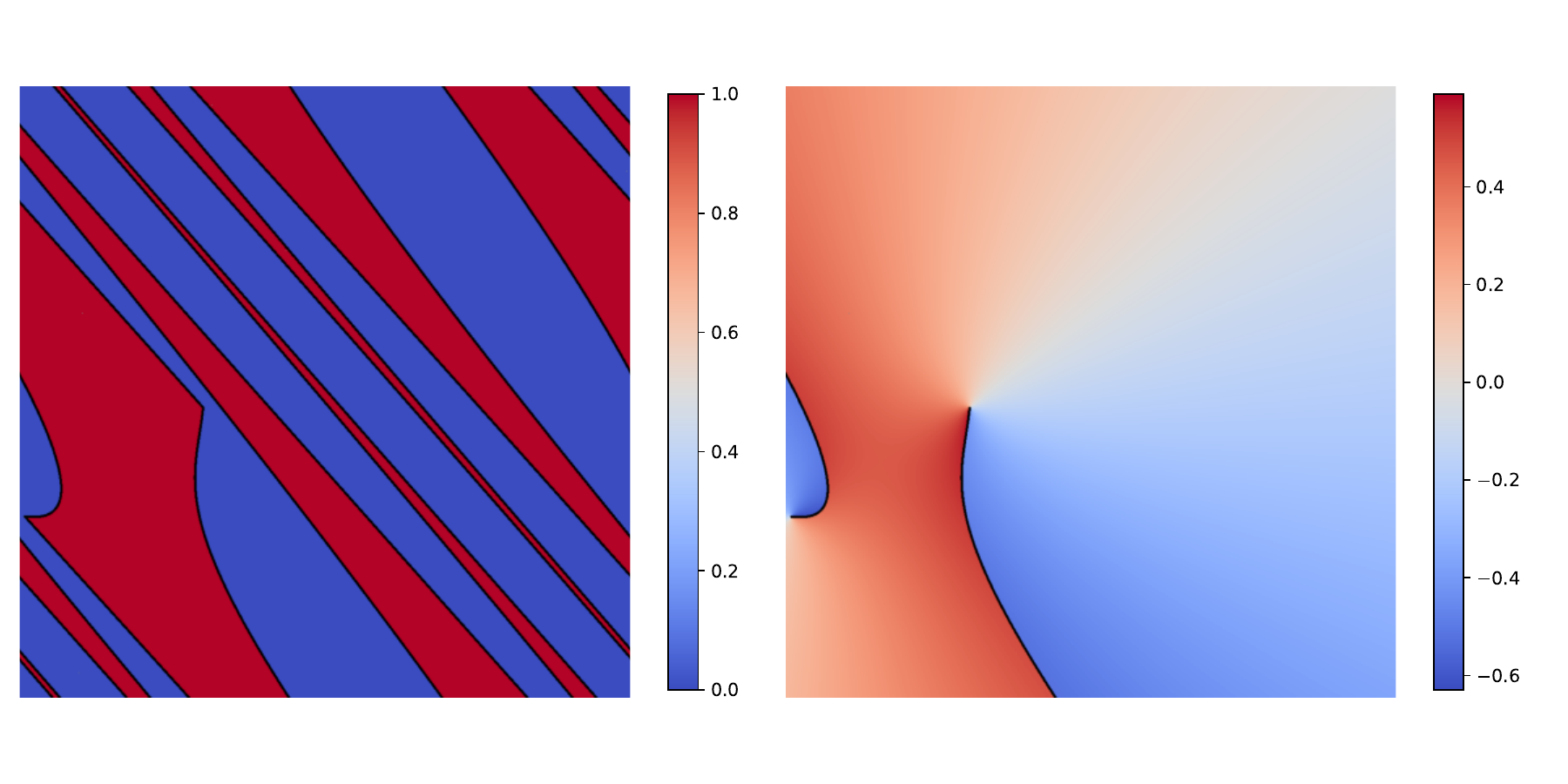}}
    \subfigure[185]{
    \includegraphics[width=0.24\linewidth]{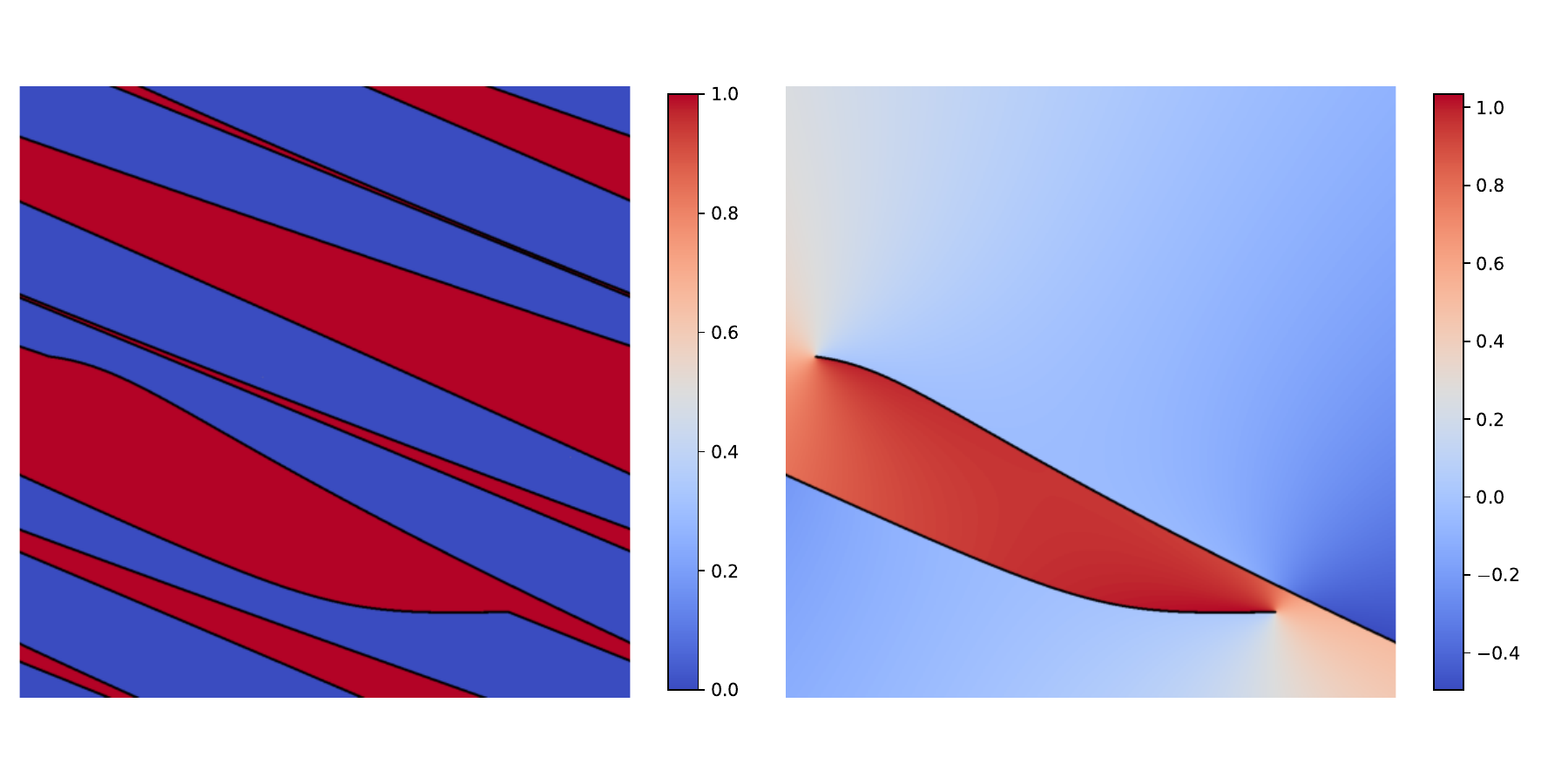}}
    \subfigure[211]{
    \includegraphics[width=0.24\linewidth]{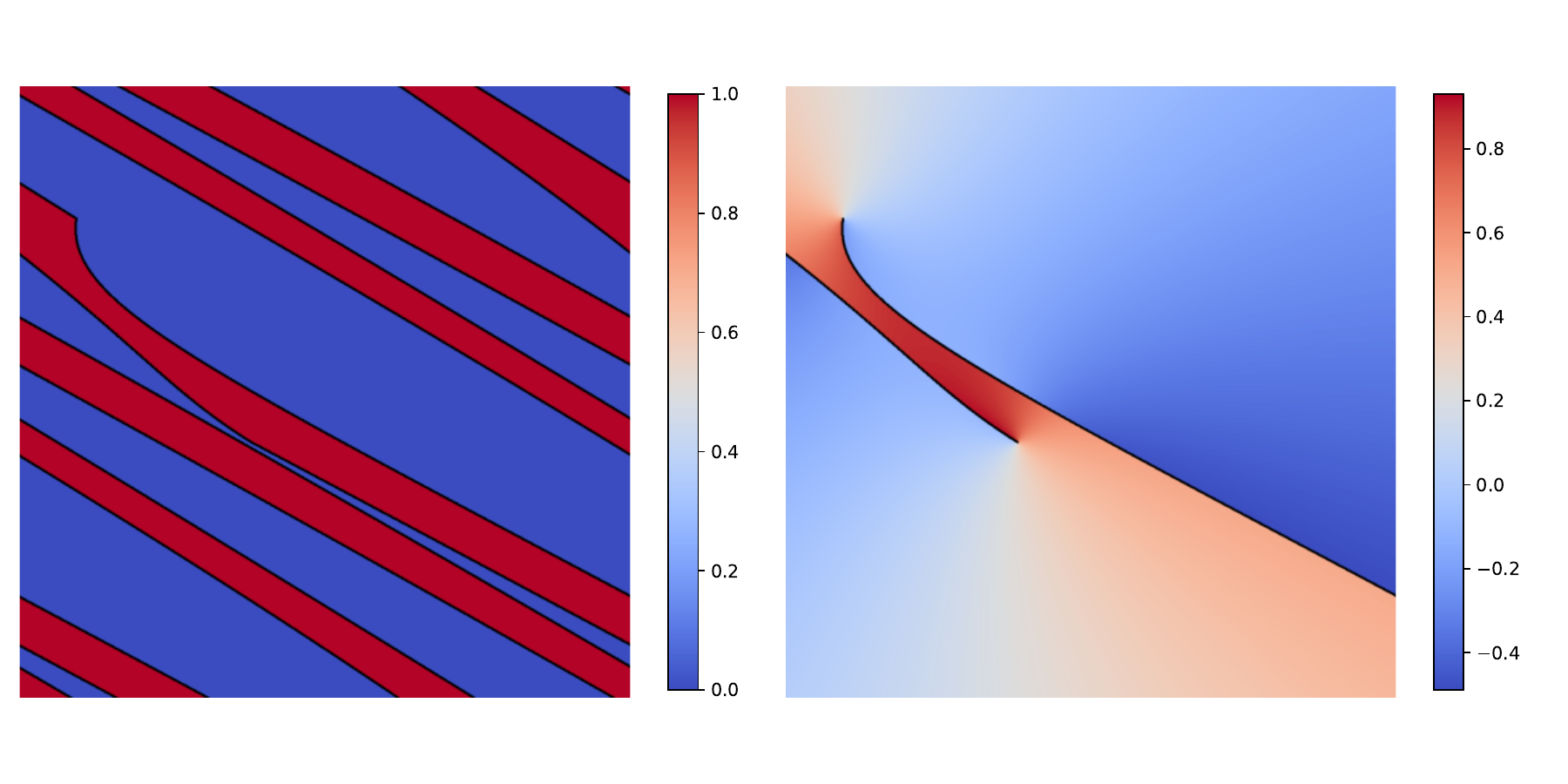}}
    \subfigure[234]{
    \includegraphics[width=0.24\linewidth]{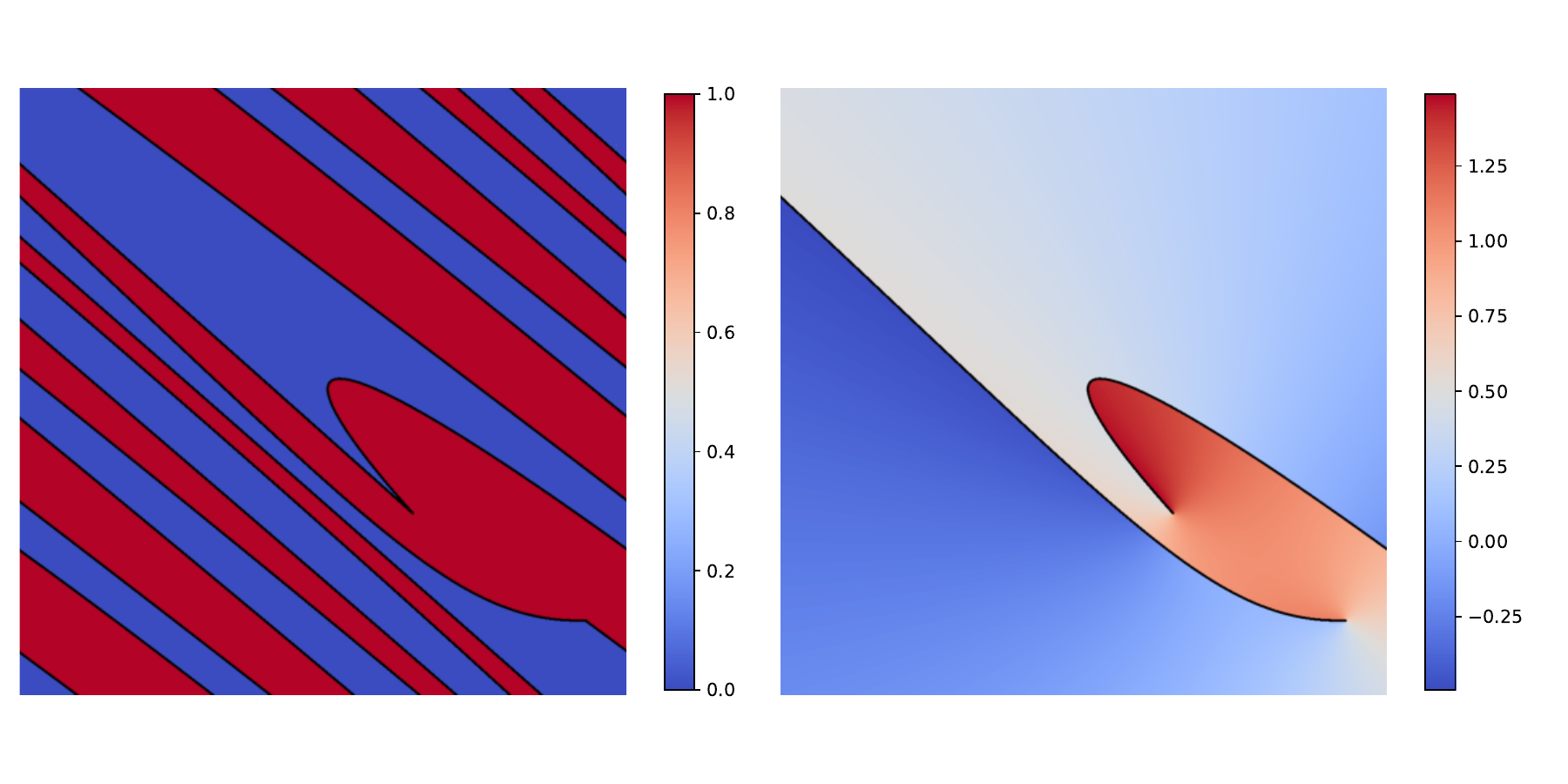}} \\
    \caption{Samples from the periodic dataset B. For each subfigure, the left part is the winding number computed by the periodic method in covering space and the right is the simply computed winding number in parametric domain.}
    \label{fig:periodic_result_b}
\end{figure*}

\appendix
\section{Proofs of the Ellipse Bound}
\subsection{Boundedness and the Rational Case}
\label{app:ellipse}

In this section, we present the detailed derivation of the ellipse bound.
Firstly, we prove the lemma \ref{thm:ellipse}

\begin{proof}
Let $p = \gamma(t)$ for some $t \in [0,1]$. Suppose, for contradiction, that $p \notin \Omega_\gamma$. Then
\begin{equation}
    d(p, \gamma(0)) + d(p, \gamma(1)) > B.
\end{equation}
However, since the curve length between $\gamma(0)$ and $\gamma(t)$ is $\int_0^t \|\gamma'(u)\|\,du$ and similarly for $\gamma(t)$ to $\gamma(1)$, we have
\begin{equation}
    d(p, \gamma(0)) + d(p, \gamma(1))
    \le \int_0^t \|\gamma'(u)\|\,du + \int_t^1 \|\gamma'(u)\|\,du
    \le B,
\end{equation}
which contradicts the assumption. Hence, $\gamma \subset \Omega_\gamma$.
\end{proof}

For a polynomial Bézier curve defined by control points $\{P_i\}_{i=0}^n$, an upper bound for its derivative norm can be obtained from its \emph{hodograph}.  
The hodograph is itself a Bézier curve with control points $\{n (P_i - P_{i-1})\}_{i=1}^n$, yielding the bound
\[
B = \max_{1 \le i \le n} n \|P_i - P_{i-1}\|.
\]

For a rational Bézier curve with control points $\{P_i\}_{i=0}^n$ and weights $\{w_i\}_{i=0}^n$, similar derivative bounds are available in the literature.  
Following \citet{floater1992derivatives}, two possible bounds for $\|\gamma'(t)\|$ are:
\begin{equation}
    \|\gamma'(t)\| \le n \frac{W}{w} \max_{i,j} \|P_i - P_j\|,
\end{equation}
\begin{equation}
    \|\gamma'(t)\| \le n \frac{W^2}{w^2} \max_i \|P_i - P_{i-1}\|,
\end{equation}
where $W = \max_i w_i$ and $w = \min_i w_i$.

\subsection{Termination}
\label{proof:termination}

In this subsection, we proof the Lemma~\ref{thm:termination}. The key insight is that each subdivision halves the diameter of the bounding ellipse.

\begin{proof}
Let the recursion depth be $m$. At this level, the diameter of the bounding ellipse is $B / 2^m$.

\begin{itemize}
    \item If $d > \epsilon$ and $m \ge \log(B / \epsilon)$, the test point $p$ cannot lie within any ellipse at this depth, and the recursion terminates.
    \item If $d \le \epsilon$, let $q$ be the closest point on $\gamma$ to $p$.  There exists an endpoint $v$ of the current subsegment such that
\[
\|v - q\| < \frac{B}{2^m}.
\]
Thus,
\[
\|p - v\| \le \|p - q\| + \|q - v\| < d + \frac{B}{2^m}.
\]
When $m \ge \log(B / \epsilon) + 1$, we have $\frac{B}{2^m} < \epsilon / 2$, so
\[
\|p - v\| < d + \frac{\epsilon}{2}.
\]
    \begin{itemize}
        \item If $d \le \epsilon / 2$, then $\|p - v\| < \epsilon$, and the algorithm reports \texttt{"boundary"}.
        \item If $d > \epsilon / 2$, the ellipse diameter is already smaller than $\epsilon / 2$, meaning the test point lies outside all bounding ellipses, and recursion terminates.  
    \end{itemize}
\end{itemize}
\end{proof}

\section{Proofs of the Periodic Case}

Throughout, a \emph{simple loop} refers to a loop without self-intersections.

\subsection{Proof of Proposition~\ref{thm:main}}
\label{proof:main}

\begin{proposition}
\label{thm:cylinder_zero_case}
Let $S$ be a surface homeomorphic to a cylinder, and let $\gamma$ be a \emph{simple} loop on $S$. Then the homology coefficient $n$ of $\gamma$ satisfies $n \in \{0, \pm 1\}$.
\end{proposition}

\begin{proof}
Let $p$ denote the projection map from the cylinder to $S^1$.  
If $|n| \ge 2$, then the composed mapping $p \circ \gamma : S^1 \to S^1$ is non-injective.  
Hence, there exist $t_1, t_2 \in S^1$ such that $p \circ \gamma(t_1) = p \circ \gamma(t_2)$, implying that $\gamma(t_1)$ and $\gamma(t_2)$ lie on the same vertical line.  
This results in a self-intersection of $\gamma$, contradicting its simplicity.
\end{proof}

\begin{proposition}
\label{thm:zero_case}
Let $S$ be a surface homeomorphic to a torus, and let $\gamma$ be a \emph{simple} loop on $S$ with homology coefficients $(n, m)$.  
If $n = 0$, then $m \in \{0, \pm 1\}$; similarly, if $m = 0$, then $n \in \{0, \pm 1\}$.
\end{proposition}

\begin{proof}
We prove only the case $n = 0$ and $m \ge 2$; the other cases are symmetric.  
Lift $\gamma$ to the universal covering space, obtaining a simple curve $\tilde{\gamma}$.  
Let $H$ be the subgroup of deck transformations generated by $(0,1)$; then the quotient $C = \mathbb{R}^2 / H$ is homeomorphic to an infinite cylinder.  
Denote the quotient map by $q$.  
The image $q(\tilde{\gamma})$ is a loop with homology coefficient $m$.  
By Proposition~\ref{thm:cylinder_zero_case}, since $m \ge 2$, $q(\tilde{\gamma})$ must self-intersect, implying that $\gamma$ is not simple.
\end{proof}

We now recall two classical results from topology \citep{farb2011primer} concerning simple loops and their intersection properties.  

\begin{proposition}
\label{thm:coprime}
Let $S$ be a surface homeomorphic to a torus, and let $\alpha$ and $\beta$ be generators of $H_1(S, \mathbb{Z})$.  
If a simple loop $\gamma$ on $S$ satisfies $\gamma = n\alpha + m\beta$, then $m$ and $n$ are coprime.
\end{proposition}

\begin{proposition}
\label{thm:intersection_number}
Let $S$ be a surface homeomorphic to a torus.  
If $\gamma_1$ and $\gamma_2$ are two loops on $S$ with homology coefficients $(n_1, m_1)$ and $(n_2, m_2)$, respectively, then the minimal intersection number between them is
\begin{equation}
    i\big((n_1, m_1), (n_2, m_2)\big) = |n_1 m_2 - n_2 m_1|.
\end{equation}
\end{proposition}

From Proposition~\ref{thm:intersection_number}, we obtain the following corollary.

\begin{proposition}
\label{thm:main_pre}
Let $S$ be a surface homeomorphic to a torus, and let $\{\gamma_i\}$ be a collection of disjoint simple loops on $S$ with homology coefficients $(n_i, m_i)$.  
Then each loop satisfies either $(n_i, m_i) = (0, 0)$ or $(n_i, m_i) = \pm(a, b)$ for some fixed integers $a, b$.
\end{proposition}

\begin{proof}
If either $n_i = 0$ or $m_i = 0$, Propositions~\ref{thm:cylinder_zero_case} and~\ref{thm:zero_case} yield the result.  
Now suppose that all coefficients are nonzero.  

By Proposition~\ref{thm:intersection_number}, if $\gamma_i$ and $\gamma_j$ are disjoint, then $n_i m_j - n_j m_i = 0$.  
Since all terms are nonzero, this implies $n_i / m_i = n_j / m_j$.  
Because $n_i$ and $m_i$ are coprime (Proposition~\ref{thm:coprime}), we conclude that $(n_i, m_i) = \pm(n_j, m_j)$.
\end{proof}

Finally, Proposition~\ref{thm:main_pre} allows us to classify all loops, completing the proof of Proposition~\ref{thm:main}.

\subsection{Proof of Proposition~\ref{thm:traverse_version2}}
\label{proof:traverse}

Proposition~\ref{thm:traverse_version2} can be understood as a partitioning of all lifted curves into disjoint periodic families.  
A concise proof is as follows.

\begin{proof}
Let $c_{i,j} = \gamma + i\vec{e}_1 + j\vec{e}_2$.  
For all integers $i,j$, each $c_{i,j}$ is contained in the right-hand side of the equation.  
Write $i = ap + b$ with $a \in \mathbb{Z}$ and $0 \le b < p$.  
Then
\[
c_{i,j} = c_{b,\, j - aq} + ap\vec{e}_1 + aq\vec{e}_2 \subset \gamma_{b,\, j - aq}.
\]
Hence, the left-hand side is contained in the right-hand side.  
Conversely, it is straightforward to verify that the right-hand side is contained in the left-hand side, completing the proof.
\end{proof}




\section{Implementation Details}
\label{sec:code}

In this section, we provide pesudo code of proposed algorithms and introduce some tedious details of our algorithm. We also refer readers to our C++ implementation (Anonymous).

\subsection{Periodic Winding Number Computation}

In this section, we list the pesudo code of winding number computation method.

In our implementation, to avoid translating parametric curves, we use the following property of winding number: Let $\vec{v}$ be a two dimensional vector, $p$ is the test point and $\gamma$ is a curve, then we have 
\begin{equation}
    \omega(p, \gamma + \vec{v}) = \omega(p - \vec{v}, \gamma)
\end{equation}

For uni-periodic case, algorithm \ref{alg:periodic} shows the pesudo code of the method computing winding number with respect to a periodically extended curve (section \ref{sec:periodic_curves}) and algorithm \ref{alg:uni_contriactible} illustrated the method handling contractible case.

For bi-periodic case, algorithm \ref{alg:bi_contriactible} demonstrates the contractible case and algorithm \ref{alg:biperiodic} presents the bi-periodic winding number computation implementation. 

\subsection{Pairing Curves in Bi-periodic Case}

In this section, we describe the algorithm detial of pairing the extended curves properly. Given a set of lift curves, we firstly split them to three set as defined in proposition \ref{thm:main}. Suppose $A = \{\alpha_i\}_{i = 1}^m$ and $B = \{\beta_i\}_{i = 1}^m$, the output of pairing algorithm is two curve sets $\{\tilde{\gamma}_i\}$ and $\{\bar{\gamma}_i\}$, such that $\tilde{\gamma}_i = \alpha_{i, (r, s)}$ and $\bar{\gamma}_i = \beta_{l_i, (r, s)} + \vec{v}_i$ is a proper pair. The term \textit{proper} means that $\bar{\gamma}_{i, (r, s)}$ is the first periodically extended curve on the left of $\tilde{\gamma}_i$.

If we directly set $\bar{\gamma}_i = \beta_i$, the band defined by periodically extended loop pairs may overlap, as illustrated in Figure \ref{fig:pair_case}. We can solve this issue by proper pairing strategy.

\begin{figure}[h]
    \includegraphics[width=0.9\linewidth]{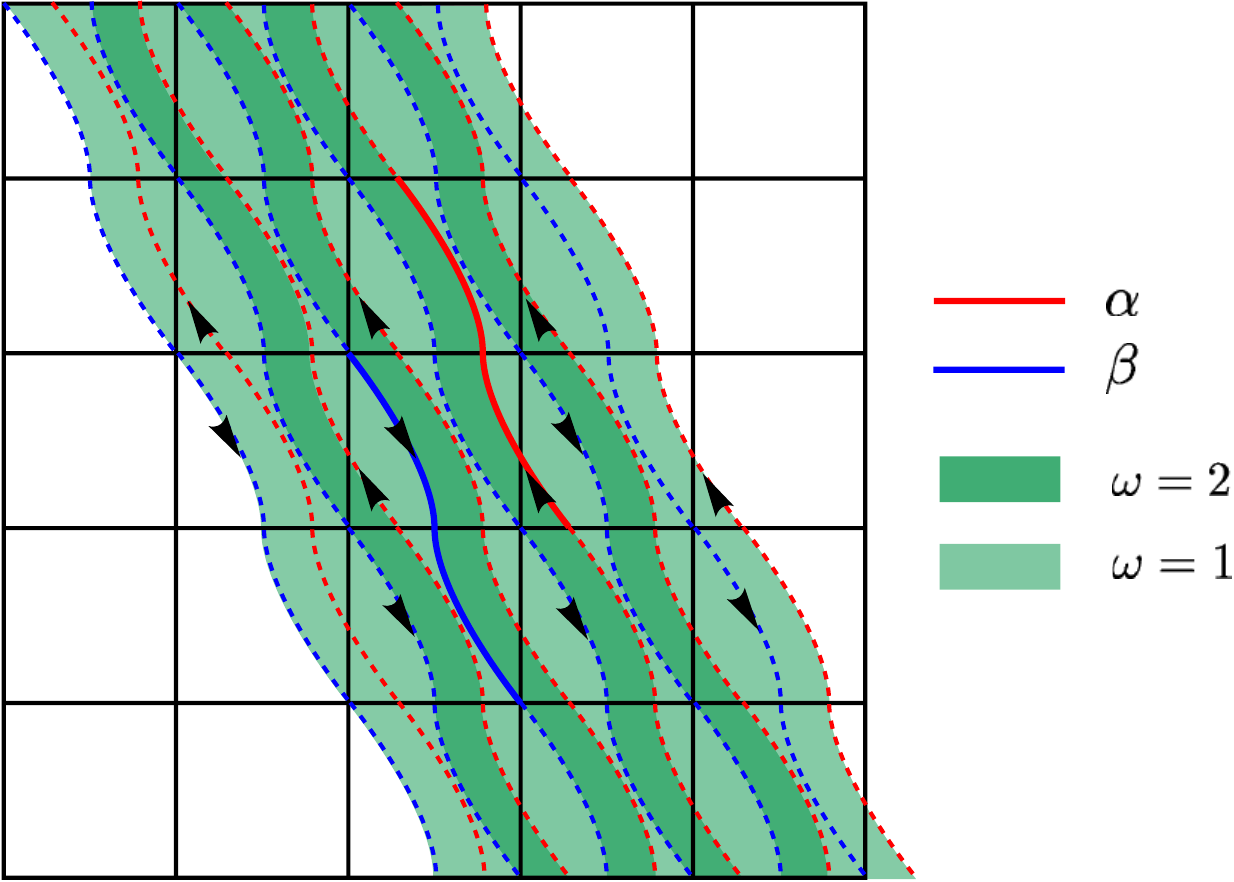}
    \caption{If the original lifted curve $\alpha$ and $\beta$ are in improper place, directly use pairing will lead to overlap between regions defined by $\alpha_{r, s}$ and $\beta_{r, s}$. The original lift of $\alpha$ and $\beta$ are represented by solid curves and other copies are illustrated by dashed curves. Winding number field $\omega$ defined by curves are plotted, overlapped regions have $\omega = 2$.}
    \label{fig:pair_case}
\end{figure}

Firstly, we define a \textit{ToLeft} predicate (see algorithm \ref{alg:to_left}) between curves $\gamma_{k, (i, j)}$. We say that $\gamma_{k^\prime, (i^\prime, j^\prime)}$ is \textit{on the left of} $\gamma_{k, (i, j)}$ if for some point $p \in \gamma_{k^\prime, (i^\prime, j^\prime)}$, the winding number $\omega(p, \gamma_{k, (i, j)})$ is positive. In our setup, the ToLeft predicate defines a partial order in set $\{\gamma_{k, (i, j)}\}_{i, j \in \mathbb{Z}}$ since these curves are insersection-free and periodically extended with same extenison per period.

With this predicate, finding $\bar{\gamma}_k$ likes the process of finding "smallest number". Algorithm \ref{alg:pair_loop} shows an algorithm find the smallest $\beta_{(i, j)}$ on the left of $\alpha_{(0, 0)}$ and Algorithm \ref{alg:pair_loops} generate the pairs $\tilde{\gamma}_i$ and $\bar{\gamma}_i$ properly.

\begin{algorithm}[H]
    \caption{PeriodicExtendedWindingNumber}\label{alg:periodic}
    \begin{algorithmic}
    \Require 2D path $\gamma$, test point $p$, tolerance value $\epsilon$.
    \Ensure Winding number of the pre-images of $\gamma$ w.r.t. $p$.
    \State $l \gets LeftMostContorlPoint(\gamma).x$
    \State $r \gets RightMostControlPoint(\gamma).x$ 
    \Comment {{\color{blue} Use $y$ coordinates if $|l - r| = 0$. Here we omit those $y$-code.}}

    \State $\vec{v} \gets \gamma(1) - \gamma(0)$ 

    \For{$k > 0$}
        \State $p^\prime = p - k \vec{v}$
        \If{$p.x - k < l$}
            \State \textbf{break}
        \EndIf
        \State $\omega \gets \omega + WindingNumber(p^\prime, \gamma, \epsilon)$
        \State $\omega \gets \omega - WindingNumberLineSegment(p^\prime, \gamma(0), \gamma(1))$
    \EndFor

    \For{$k > 0$}
        \State $p^\prime \gets p + k \vec{v}$
        \If{$p.x + k > r$}
            \State \textbf{break}
        \EndIf
        \State $\omega \gets \omega + WindingNumber(p^\prime, \gamma, \epsilon)$
        \State $\omega \gets \omega - WindingNumberLineSegment(p^\prime, \gamma(0), \gamma(1))$
    \EndFor

    \State $\omega \gets \omega + WindingNumberLine(p^\prime, \gamma(0), \gamma(1))$

    \State \textbf{Return} $\omega$

    \end{algorithmic}
\end{algorithm}

\begin{algorithm}[H]
    \caption{UniperiodicContractibleWindingNumber}\label{alg:uni_contriactible}
    \begin{algorithmic}
    \Require 2D closed loop $\gamma$, test point $p$, tolerance value $\epsilon$.
    \Ensure Winding number of the pre-images of $\gamma$ w.r.t. $p$.

    \State $l \gets LeftMostContorlPoint(\gamma).x$
    \State $r \gets RightMostControlPoint(\gamma).x$ 
    \Comment {{\color{blue} Use $y$ coordinates if $|l - r| = 0$. Here we omit those $y$-code.}}

    \For{$k > 0$}
        \State $p^\prime = p - k \vec{v}$
        \If{$p.x - k < l$}
            \State \textbf{break}
        \EndIf
        \State $\omega \gets \omega + WindingNumber(p^\prime, \gamma, \epsilon)$
    \EndFor

    \For{$k > 0$}
        \State $p^\prime \gets p + k \vec{v}$
        \If{$p.x + k > r$}
            \State \textbf{break}
        \EndIf
        \State $\omega \gets \omega + WindingNumber(p^\prime, \gamma, \epsilon)$
    \EndFor

    \State \textbf{Return} $\omega$

    \end{algorithmic}
\end{algorithm}

\begin{algorithm}[H]
    \caption{BiperiodicContractibleWindingNumber}\label{alg:bi_contriactible}
    \begin{algorithmic}
    \Require 2D closed loop $\gamma$, test point $p$, tolerance value $\epsilon$.
    \Ensure Winding number of the pre-images of $\gamma$ w.r.t. $p$.

    \State $l \gets LeftMostContorlPoint(\gamma).x$
    \State $r \gets RightMostControlPoint(\gamma).x$ 
    \State $t \gets TopMostControlPoint(\gamma).y$
    \State $b \gets BottomMostControlPoint(\gamma).y$

    \State $h \gets \lceil t - b \rceil, v \gets \lceil r - l \rceil$
    \State $\omega \gets 0$
    \For{$i$ \textbf{in} $[-h, h]$}
        \For{$j$ \textbf{in} $[-v, v]$}
            \State $p^\prime \gets p - (i , j)$
            \State $\omega \gets \omega + WindingNumber(p^\prime, \gamma)$
        \EndFor
    \EndFor

    \State \textbf{Return} $\omega$

    \end{algorithmic}
\end{algorithm}

\begin{algorithm}[H]
    \caption{PeriodicPairWN}\label{alg:periodic_pair}
    \begin{algorithmic}
    \Require 2D path $\gamma_1$, $\gamma_2$, test point $p$, tolerance value $\epsilon$.
    \Ensure Winding number of the pre-images of $\gamma$ w.r.t. $p$.

    \State $\omega_1 \gets PeriodicExtendedWindingNumber(p, \gamma_1, \epsilon)$

    \State $\omega_2 \gets PeriodicExtendedWindingNumber(p, \gamma_2, \epsilon)$

    \State \textbf{Return} $\omega_1 + \omega_2$
    \end{algorithmic}
\end{algorithm}

\begin{algorithm}[H]
    \caption{ComputeBiperiodicWindingNumber}\label{alg:biperiodic}
    \begin{algorithmic}
    \Require test point $p$, tolerance value $\epsilon$.
    \Ensure Winding number of the pre-images of $\gamma$ w.r.t. $p$.
    
    \State $\omega \gets 0$


    \For{All pairs $\tilde{\gamma}_k, \bar{\gamma}_k$}
        \For{$i$ \textbf{in} $[0, n)$}
            \For{$j$ \textbf{in} $[-|m|, |m|]$}
                \State $\omega \gets \omega + PeriodicPairWN(p - (i, j), \tilde{\gamma}_k, \bar{\gamma}_k)$
            \EndFor
        \EndFor
    \EndFor

    \State \textbf{Return} $\omega$

    \end{algorithmic}
\end{algorithm}

\begin{algorithm}[H]
    \caption{ToLeft}\label{alg:to_left}
    \begin{algorithmic}
    \Require 2D path $\gamma_1$, $\gamma_2$.
    \Ensure \textbf{true} if $\gamma_1$ is on the left of $\gamma_2$, otherwise \textbf{false}.

    \State $\omega = PeriodicExtendedWindingNumber(\gamma_1(0.5), \gamma_2)$

    \State \textbf{Return} $\omega > 0$
    \end{algorithmic}
\end{algorithm}

\begin{algorithm}[H]
    \caption{PairLoop} \label{alg:pair_loop}
    \begin{algorithmic}
        \Require 2D path $\alpha \in A$, $\beta \in B$.
        \Ensure $\vec{v} \in \mathbb{R}^2$, such that the region defined by the lift curves of $\alpha$ and $\beta + v$ containes no other preimages of $\alpha$ or $\beta$

        \State $p, q \gets$ homology coefficients of set A 

        \State $\vec{v} = q \vec{e}_2$

        \For{$i$ \textbf{in} $[0, p)$}
            \For {$j$ \textbf{in} $[-q, q]$}
                \If {$ToLeft(\beta + (i, j), \alpha)$ and $ToLeft(\beta + (i, j), \beta + \vec{v})$}
                    \State $\vec{v} \gets (i, j)$
                \EndIf
            \EndFor
        \EndFor
        \State \textbf{Return} $\vec{v}$
    \end{algorithmic}
\end{algorithm}

\begin{algorithm}[H]
    \caption{PairLoops} \label{alg:pair_loops}
    \begin{algorithmic}
        \Require 2D path $\gamma_i$.
        \Ensure Winding number of the pre-images of $\gamma$ w.r.t. $p$.
    
        \State $A, B, C \gets SplitLoops(\{\gamma_i\})$

        \State $n \gets SizeOf(A)$
        \For{$i$ \textbf{in} $[0, n)$}
            \State $\tilde{\gamma}_i = \alpha_i = A[i]$

            \State $\bar{\gamma}_i = NULL$
            \State $toRemove = NULL$

            \For{$\beta$ \textbf{in} $B$}
                \State $\bar{\gamma} \gets PairLoop(\alpha_i, \beta)$
                \If {$\bar{\gamma}_i$ is $NULL$ or $ToLeft(\bar{\gamma}, \bar{\gamma}_i)$}
                    \State $\bar{\gamma}_i \gets \bar{\gamma}$
                    \State $toRemove \gets \bar{\gamma}$
                \EndIf
            \EndFor

            \State Remove $\alpha_i$ from $A$
            \State Remove $toRemove$ from $B$
        \EndFor

        \State \textbf{Return} $\{\tilde{\gamma}_i\}$ and $\{\bar{\gamma}_i\}$
        \end{algorithmic}
\end{algorithm}

\end{document}